\documentclass[11pt]{article}

\usepackage[colorlinks,citecolor=blue,bookmarks=true]{hyperref}

\usepackage{amsmath} 
\usepackage{amsthm} 
\usepackage{amssymb}	
\usepackage{graphicx} 
\usepackage{multirow}
\usepackage[dvipsnames]{xcolor}
\usepackage[dvips,letterpaper,margin=1in]{geometry}
\usepackage[capitalize,noabbrev]{cleveref}
\usepackage{subcaption}

\usepackage[utf8]{inputenc}
\usepackage[english]{babel}

\usepackage{algorithm}
\usepackage{algpseudocode}
\makeatletter
\renewcommand{\ALG@name}{Algorithm}
\makeatother

\usepackage{bm}
\usepackage{bbm}
\usepackage{diagbox}
\usepackage{mathtools}
\usepackage{ytableau}

\usepackage{tabularx, booktabs}

\newtheorem{theorem}{Theorem}[section]

\newtheorem{lemma}[theorem]{Lemma}
\newtheorem{corollary}[theorem]{Corollary}
\newtheorem{proposition}[theorem]{Proposition}
\newtheorem{fact}[theorem]{Fact}

\newtheorem{definition}[theorem]{Definition}

\newtheorem{remark}[theorem]{Remark}
\newtheorem{problem}[theorem]{Problem}
\newtheorem{question}[theorem]{Question}
\crefname{question}{Question}{Questions}

\newtheorem{notation}[theorem]{Notation}

\newcommand{\braket}[2]{\left< #1 \vphantom{#2} \middle| #2 \vphantom{#1} \right>} 

\DeclarePairedDelimiter\abs{\lvert}{\rvert}

\DeclarePairedDelimiter\ceil{\lceil}{\rceil}
\DeclarePairedDelimiter\floor{\lfloor}{\rfloor}
\DeclarePairedDelimiter\ket{\lvert}{\rangle}
\DeclarePairedDelimiter\bra{\langle}{\rvert}

\newcommand{\E}{\mathop{\mathbb{E}\/}}

\newcommand{\tr} {\operatorname{tr}}

\newcommand{\supp} {\operatorname{supp}}
\newcommand{\spanspace} {\operatorname{span}}

\newcommand{\Co}{\mathbb C}

\newcommand{\ketbra}[2]{\ensuremath{\ket{#1}\!\bra{#2}}}
\renewcommand{\braket}[2]{\ensuremath{\langle {#1} \vert {#2} \rangle}}
\newcommand{\kett}[1]{|#1\rangle\!\rangle}
\newcommand{\bbra}[1]{\langle\!\langle#1|}
\newcommand{\kettbbra}[2]{\ensuremath{\kett{#1}\!\bbra{#2}}}
\newcommand{\bbrakett}[2]{\ensuremath{\langle\!\langle{#1}\vert{#2}\rangle\!\rangle}}

\allowdisplaybreaks

\DeclarePairedDelimiter\parens{\lparen}{\rparen}
\DeclarePairedDelimiter\norm{\lVert}{\rVert}
\DeclarePairedDelimiter\braces{\lbrace}{\rbrace}
\DeclarePairedDelimiter\bracks{\lbrack}{\rbrack}

\newcommand{\calA}{\mathcal{A}}

\newcommand{\calE}{\mathcal{E}}
\newcommand{\calF}{\mathcal{F}}

\newcommand{\calN}{\mathcal{N}}

\newcommand{\calV}{\mathcal{V}}

\newcommand{\bb}{\boldsymbol{b}}

\newcommand*{\Tr}{\mathrm{tr}}
\newcommand*{\eps}{\varepsilon}

\newcommand*{\ex}[1]{\mathbb{E}\left[#1 \right]}

\newcommand{\qchannel}{\textbf{\textup{QChan}}}
\newcommand{\isochannel}{\textbf{\textup{ISO}}}
\newcommand{\dilation}{\textbf{\textup{Dilation}}}
\newcommand{\contract}{\textbf{\textup{Contract}}}

\renewcommand{\bb}{\begin{equation}\begin{aligned}\hspace{0pt}}
\newcommand{\ee}{\end{aligned}\end{equation}}

\newcommand{\EE}[1]{\underset{\scaleobj{.8}{#1}}{\mathbb{E}\,}}

\usepackage{dsfont}
\usepackage{scalerel}

\renewcommand{\Pr}{\mathbb{P}}

\newcommand{\sr}{\kappa}

\begin{document}

\title{Quantum channel tomography: optimal bounds\\ and a Heisenberg-to-classical phase transition\footnote{This paper subsumes prior papers~\cite{chen2025quantum,oufkir2026improved,chen2026optimal}.}
}
\date{}

\author{
Kean Chen\thanks{University of Pennsylvania, USA. Email: \texttt{keanchen.gan@gmail.com}}\and
Filippo Girardi\thanks{Scuola Normale Superiore, Italy. Email: \texttt{filippo.girardi@sns.it}}\and
Aadil Oufkir\thanks{University Mohammed VI Polytechnic, Morocco. Email: \texttt{aadil.oufkir@gmail.com}}\and
Nengkun Yu\thanks{Stony Brook University, USA. Email: \texttt{nengkun.yu@stonybrook.edu}}\and
Zhicheng Zhang\thanks{Griffith University, Australia. Email: \texttt{iszczhang@gmail.com}}
}

\maketitle


\begin{abstract}  
How many black-box queries to a quantum channel are needed to learn its full classical description? This question lies at the heart of quantum channel tomography (also known as quantum process tomography), a fundamental task in the characterization and validation of quantum hardware. Despite extensive prior work, the optimal query complexity for quantum channel tomography is far from fully understood.

In this paper, we study tomography of an unknown quantum channel with input dimension $d_1$, output dimension $d_2$, and Kraus rank at most $r$, to within error $\varepsilon$. We identify the dilation rate $\tau = r d_2 / d_1$ (which always satisfies $\tau\geq 1$ due to the trace preservation of quantum channels) as a key parameter, and establish that the optimal query complexity of channel tomography exhibits distinct scaling laws across three regimes of $\tau$.
\begin{itemize}
    \item 
    In the boundary regime ($\tau = 1$): we show that the query complexity is $\Theta(r d_1 d_2/\varepsilon)$ for Choi trace norm error $\varepsilon$, and is upper bounded by $O(\min\{r d_1^{1.5} d_2/\varepsilon, r d_1 d_2/\varepsilon^2\})$ and lower bounded by $\Omega(r d_1 d_2/\varepsilon)$ for diamond norm error $\varepsilon$.
    \item 
    In the away-from-boundary regime ($\tau \geq 1+\Omega(1)$):
    we show that the query complexity is $\Theta(r d_1 d_2/\varepsilon^2)$ for both Choi trace norm and diamond norm errors $\varepsilon$.
\end{itemize}
Our results uncover a sharp Heisenberg-to-classical phase transition in the query complexity of quantum channel tomography: at $\tau=1$, the optimal query complexity exhibits Heisenberg scaling $1/\varepsilon$, whereas for $\tau\geq 1+\Omega(1)$, it exhibits classical scaling $1/\varepsilon^2$.
In addition, we show that in the near-boundary regime ($1< \tau < 1+o(1)$), the query complexity exhibits a mixture of Heisenberg and classical scaling behaviors.
\end{abstract}

\tableofcontents
\newpage

\section{Introduction}\label{sec-770121}

Quantum channel tomography (also known as quantum process tomography) asks how to reconstruct an unknown quantum process from experimental data.
    Specifically, given query access to a quantum channel $\calE$, the goal is to learn a full classical description of $\calE$ to within some prescribed error, with high probability (say, at least $2/3$). This task is central to the characterization and validation of quantum devices, and it has been studied extensively for nearly three decades~\cite{chuang1997prescription,poyatos1997complete,leung2000towards,d2001quantum,mohseni2008quantum,knee2018quantum,kliesch2019guaranteed,bouchard2019quantum,surawy2022projected,Oufkir_2023,oufkir2023adaptivity,chen2023unitarity,huang2023learning,chen2023learnability,pmlr-v195-fawzi23a,raza2024online,vasconcelos2024learning,fanizza2024learning,caro2024learning,rosenthal2024quantum,zhao2024learning,chen2024tight,wadhwa2025learning,chen2025top,chen2025efficient,pmlr-v291-chen25c,wadhwa2025agnostic,zambrano2025fast,angrisani2025learning,yoshida2025quantum}. 
    
    A key question is the optimal query complexity of quantum channel tomography.
    We consider an unknown channel with input dimension $d_1$, output dimension $d_2$, and Kraus rank at most $r$.
    Several important special cases are by now well understood:
    \begin{itemize}
        \item 
        When the input dimension $d_1=1$, the task reduces to \textit{quantum state tomography}.
        The optimal query complexity for pure-state tomography was established in~\cite{hayashi1998asymptotic,bruss1999optimal,keyl1999optimal}. The mixed-state case was settled later by Haah, Harrow, Ji, Wu, and Yu~\cite{Haah_2017} and by O'Donnell and Wright~\cite{10.1145/2897518.2897544}.
        Further refinements and extensions appear in~\cite{o2017efficient,GKKT20,yuen2023improved,chen2023does,chen2024adaptivity,scharnhorst2025optimal,pelecanos2025debiased,pelecanos2025mixedstatetomographyreduces}.
        \item 
        When the unknown channel is a \textit{unitary channel} (i.e., $d_1=d_2=d$ and $r=1$), Haah, Kothari, O'Donnell, and Tang~\cite{haah2023query} established that the optimal query complexity is $\Theta(d^2/\varepsilon)$, where $\varepsilon$ is the diamond norm error. Notably, in this case, Heisenberg scaling $1/\varepsilon$ is achievable.
        \item 
        When the unknown channel is an \textit{isometry channel},
        Yoshida, Miyazaki, and Murao~\cite{yoshida2025quantum} established a lower bound of $\widetilde{\Omega}((d_2-d_1)d_1 /\varepsilon^{2})$ for Choi trace norm error $\varepsilon$, where $\widetilde{\Omega}$ suppress the logarithmic factors.\footnote{Since we consider the tomography with success probability at least $2/3$, the lower bound in \cite{yoshida2025quantum} applies to our setting, provided that the success probability is amplified to $1-O(\varepsilon^2)$, which incurs an additional logarithmic factor in $\varepsilon$.}
        \item 
        When only \textit{non-adaptive incoherent queries} are allowed, Oufkir~\cite{Oufkir_2023,oufkir2023adaptivity} established near-optimal bounds of $\widetilde{\Theta}(d_1^3 d_2^3/\varepsilon^2)$ for full-Kraus-rank (i.e., $r=d_1d_2$) channel tomography under diamond norm error $\varepsilon$, generalizing the algorithm in~\cite{surawy2022projected}.
        \item 
        A folklore approach, based on optimal mixed-state tomography in trace norm applied to the Choi state of the unknown channel, yields a query upper bound of $O(rd_1^3d_2/\varepsilon^2)$ for quantum channel tomography under diamond norm error.
    \end{itemize}

Despite this progress, the optimal query complexity of \textit{general} quantum channel tomography is still far from fully understood. In particular, it remains unclear
how the dimension parameters $d_1,d_2,r$ interact with the error parameter $\varepsilon$ in the optimal query complexity.
A central question is to understand when this optimal query complexity has Heisenberg scaling $1/\varepsilon$ (compared to classical scaling $1/\varepsilon^2$), while simultaneously maintaining the optimal dependence on the dimension parameters. In other words:
\begin{question}
    \label{ques:main}
    \textit{When and how does the optimal query complexity of quantum channel tomography exhibit Heisenberg scaling?}
\end{question}
\vspace{1mm}

\subsection{Our results}

To answer the above question, we prove new upper and lower bounds on the query complexity of quantum channel tomography.
In particular, we identify the \textit{dilation rate} $\tau = rd_2/d_1$ as a key parameter. 
Note that $\tau\geq 1$ always holds since quantum channels are trace-preserving.
Depending on the value of $\tau$, this leads to three parameter regimes:
boundary regime ($\tau=1$), near-boundary regime ($1< \tau < 1+o(1)$), and away-from-boundary regime ($\tau \geq 1+\Omega(1)$).
We establish multiple optimal bounds 
and identify when the optimal query complexity exhibits Heisenberg scaling: the dependence $1/\varepsilon$ is achievable in the boundary regime,
but is no longer possible once one moves beyond the boundary regime; and in the away-from-boundary regime, only the classical scaling $1/\varepsilon^2$ is possible.
An illustration of such Heisenberg-to-classical phase transition is shown in \Cref{fig:phase-diagram}.

In the following, we characterize in detail how the query complexity behaves across these regimes.
Our results cover both the Choi trace norm\footnote{By the Choi trace norm, we mean the trace norm of normalized Choi states; see \cref{def-3301809}.} and the diamond norm, which are standard metrics in quantum information theory for quantifying average-case and worst-case errors, respectively.

\begin{figure}[t]
  \centering
  \def\svgwidth{\linewidth}
  \resizebox{.87\textwidth}{!}{
  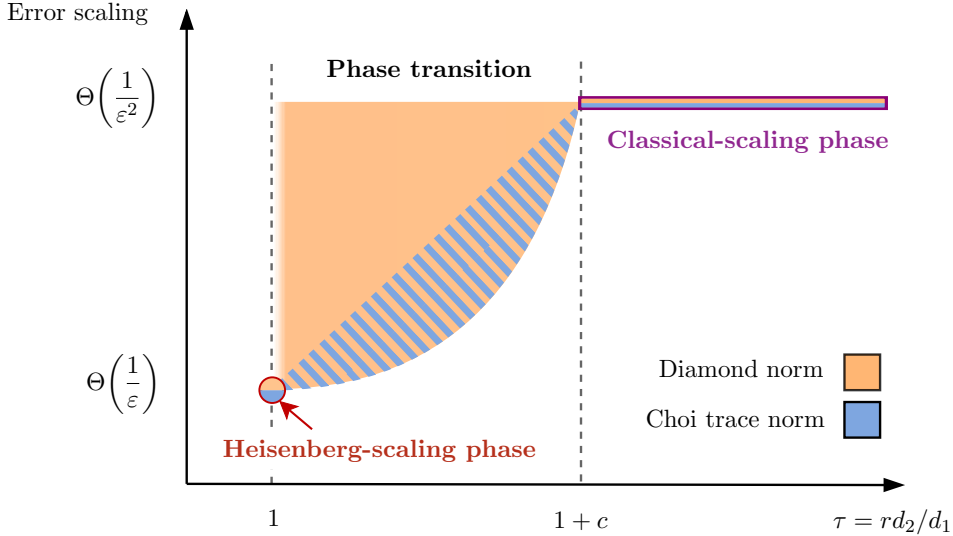}
  \caption{Heisenberg-to-classical phase transition in the query complexity of quantum channel tomography. Here, $c\in (0,1]$ can be an arbitrary constant. For simplicity of presentation, we focus on the scaling of $\varepsilon$, by taking $d_1$ to be a constant and $\tau$ to be an independent parameter. 
  Orange (resp.\ blue) regions indicate where the optimal query complexity of channel tomography should lie, according to our upper and lower bounds, under diamond norm error (resp.\ Choi trace norm error).
  The optimal query complexity exhibits Heisenberg scaling $1/\varepsilon$ in the boundary regime and classical scaling $1/\varepsilon^2$ in the away-from-boundary regime. In the near-boundary regime, it exhibits a mixture of Heisenberg and classical scalings.} 
  \label{fig:phase-diagram}
\end{figure}

\subsubsection{Boundary to away-from-boundary: a phase transition}

For boundary and away-from-boundary regimes, we establish the following query complexity bounds for quantum channel tomography and observe a Heisenberg-to-classical phase transition.
\begin{theorem}[Boundary and away-from-boundary regimes, \Cref{coro-12120135,coro-12120134,thm-390442,thm-390512} restated]\label{thm-2150126}
Consider an unknown quantum channel $\mathcal{E}$ with input dimension $d_1$, output dimension $d_2$, and Kraus rank at most $r$. Let $\tau=rd_2/d_1\geq 1$ be the dilation rate. Then, the query complexity of the tomography of $\calE$ to within error $\varepsilon$ has the following bounds.
\begin{itemize}
    \item In the boundary regime $\tau=1$: we establish matching upper and lower bounds of $\Theta\!\left(\frac{rd_1d_2}{\varepsilon}\right)$ for Choi trace norm error $\varepsilon$; and also establish an upper bound of $O\!\left(\min\!\left\{\frac{rd_1^{1.5}d_2}{\varepsilon},\frac{rd_1d_2}{\varepsilon^2}\right\}\right)$ and a lower bound of $\Omega\!\left(\frac{rd_1d_2}{\varepsilon}\right)$ for diamond norm error $\varepsilon$.
    \item In the away-from-boundary regime $\tau \geq 1+c$, for $c\in (0,1]$ an arbitrary (but fixed) constant: we establish matching upper and lower bounds of $\Theta(\frac{rd_1d_2}{\varepsilon^2})$ for both diamond norm and Choi trace norm errors $\varepsilon$.
\end{itemize}
\end{theorem}

\begin{table}[ht]
\centering
\setlength{\tabcolsep}{2.0mm}{{\renewcommand{\arraystretch}{1.0}
\begin{tabular}{|c|c|c|}
\hline
 & \begin{tabular}[c]{@{}c@{}}Boundary\\ $\tau=1$\end{tabular} & \begin{tabular}[c]{@{}c@{}}Away-from-boundary\\ $\tau\geq 1+c$\end{tabular} \\ \hline
Upper bounds &   
\begin{tabular}[c]{@{}c@{}}\noalign{\smallskip}
$O\!\left(\dfrac{rd_1d_2}{\varepsilon}\right)^\dag$, $O\!\left(\min\!\left\{\dfrac{rd_1^{1.5}d_2}{\varepsilon},\dfrac{rd_1d_2}{\varepsilon^2}\right\}\right)^\ddagger$ \\ \noalign{\smallskip}
 \cref{coro-12120135}\end{tabular} 
& \begin{tabular}[c]{@{}c@{}}\noalign{\smallskip}
$O\parens*{\dfrac{rd_1d_2}{\varepsilon^2}}$ \\ \noalign{\smallskip}
\cref{coro-12120134} \end{tabular} \\ \hline
Lower bounds&  \begin{tabular}[c]{@{}c@{}}\noalign{\smallskip}
$\Omega\parens*{\dfrac{rd_1d_2}{\varepsilon}}$ \\ \noalign{\smallskip}
\cref{thm-390442}\end{tabular}   & \begin{tabular}[c]{@{}c@{}}\noalign{\smallskip}
$\Omega\parens*{\dfrac{rd_1d_2}{\varepsilon^2}}$\\ \noalign{\smallskip}
\cref{thm-390512} \end{tabular} \\ \hline
\end{tabular}}
}
\caption{Boundary and away-from-boundary regimes. Here $c\in (0,1]$ can be an arbitrary constant.
All bounds hold for both diamond norm and Choi trace norm errors except: $^\dag$holds for Choi trace norm error; $^\ddagger$holds for diamond norm error.
}
\label{tab:comparison}
\end{table}

\Cref{tab:comparison} summarizes the results in \Cref{thm-2150126}. 
We note that \cref{thm-2150126} is optimal in the following senses:
\begin{itemize}
\item The dependence on all parameters $d_1,d_2,r$ and $\varepsilon$ is optimal in the away-from-boundary regime, for both diamond norm and Choi trace norm errors.
\item The dependence on all parameters $d_1,d_2,r$ and $\varepsilon$ is optimal in the boundary regime for Choi trace norm error.
\item The dependence on $d_1$, $d_2$, and $r$ is optimal in all regimes for both constant diamond norm and Choi trace norm errors.
\end{itemize}
It also shows, from the boundary regime to the away-from-boundary regime,
the optimal query complexity exhibits a phase transition from Heisenberg scaling $1/\varepsilon$ to classical scaling $1/\varepsilon^2$.

As a special case of \Cref{thm-2150126}, we consider quantum channels with equal input and output dimensions, i.e., $d_1 = d_2 = d$, which are simply called $d$-dimensional quantum channels.
In this case, $\tau=r$ is always a positive integer and $\tau=1$ if and only if $\mathcal{E}$ is a unitary channel.
Using \cref{thm-2150126} in conjunction with the unitary channel tomography results by Haah, Kothari, O'Donnell, and Tang~\cite{haah2023query}, we obtain a complete characterization of the query complexity for tomography of 
$d$-dimensional quantum channels.
\begin{corollary}[Tomography of $d$-dimensional quantum channels]\label{equaldimension}
The query complexity of tomography of $d$-dimensional quantum channels with Kraus rank at most $r$, and to within either diamond norm error $\varepsilon$ or Choi trace norm error $\varepsilon$, is
\[
\Theta\!\left(\frac{rd^2}{\varepsilon^{\min\{r,2\}}}\right).
\]
\end{corollary}

As another illuminating and well-studied special case of \Cref{thm-2150126}, we consider tomography of quantum channels with input dimension $d_1=1$, which reduces to quantum state tomography. In this regime, our result recovers the recent optimal sample lower bound~\cite{scharnhorst2025optimal} (without logarithmic factors) for quantum state tomography, matching the upper bound established in~\cite{10.1145/2897518.2897544}.
\begin{corollary}[Optimal lower bound for state tomography]
Tomography of a $d$-dimensional mixed state with rank at most $r$, to within trace norm error $\varepsilon$, requires $\Omega\!\left(\frac{dr}{\varepsilon^2}\right)$ samples.
\end{corollary}
Our proof of this lower bound proceeds via a different line of reasoning than that in~\cite{scharnhorst2025optimal} and may offer complementary insights.

The following result underpins our upper bounds in \Cref{thm-2150126} by establishing a connection between the query complexities in two different access models: given access to an unknown channel itself, or given access to one of its Stinespring dilations. We say a quantum algorithm is a tester if it makes queries to an unknown quantum channel and outputs a classical outcome. A tester is a parallel tester if its queries to the unknown channel can be made in parallel. Formal definitions can be found in \Cref{sub:tester}.

\begin{theorem}[Dilation does not help for parallel testers, \cref{coro-12122150} restated]
\label{thm-2150136}
If there exists a parallel (possibly coherent) tester that solves a channel estimation task using $n$ queries to an arbitrary dilation of an unknown quantum channel $\mathcal{E}$, then there exists a parallel tester that solves this task using $n$ queries to $\mathcal{E}$ itself.
\end{theorem}

We note that \cref{thm-2150136} partially answers a conjecture from Tang, Wright, and Zhandry~\cite{tang2025conjugate}, which asserts that access to channel dilations does not help.

Moreover, we also obtain the following Heisenberg-scaling query upper bound for quantum state tomography with state-preparation channels, by applying \cref{thm-2150136} to the pure-state tomography algorithm due to Chen~\cite{chen2025inverse}, which achieves the Heisenberg scaling using parallel queries without inverses. 

\begin{corollary}[State tomography with state-preparation channels, \Cref{coro-12120136} restated]\label{coro-12131443}
When $\tau =1$, tomography of the mixed state $\mathcal{E}(\ketbra{0}{0})$ to within trace norm error $\varepsilon$ can be done using $O\!\left(\min\!\left\{\frac{d_1^{1.5}}{\varepsilon},\frac{d_1}{\varepsilon^2}\right\}\right)$ queries to $\mathcal{E}$.
\end{corollary}

\subsubsection{Near-boundary: new upper and lower bounds}

In the near-boundary regime, we prove new upper and lower bounds on the query complexity of quantum channel tomography,
which exhibits a mixture of Heisenberg and classical scalings. 

\begin{theorem}[Near-boundary regime, \Cref{coro-2270246,coro-12120134,thm-390442,coro-3270443,thm-390511} restated]
\label{thm:near-boundary}
Consider an unknown quantum channel $\mathcal{E}$ with input dimension $d_1$, output dimension $d_2$, and Kraus rank at most $r$. Let $\tau=rd_2/d_1\geq 1$ be the dilation rate. Then, for $\tau\in(1,\frac{4}{3})$ where $\tau$ can be arbitrarily close to $1$, the query complexity of the tomography of $\calE$ to within error $\varepsilon$ has the following bounds.
\begin{itemize}
    \item 
     For Choi trace norm error $\varepsilon$: we establish an upper bound of $O\!\left(\frac{d_1^2}{\varepsilon}+\frac{(\tau-1)d_1^2}{\varepsilon^2}\right)$ and a lower bound of $\Omega\parens*{\frac{d_1^2}{\varepsilon}+\frac{(\tau-1)^2}{\varepsilon^2}}$.
     \item 
     For diamond norm error $\varepsilon$: we establish an upper bound of $O\parens*{\frac{d_1^2}{\varepsilon^2}}$ and a lower bound of $\Omega\parens*{\frac{d_1^2}{\varepsilon} +\frac{(\tau-1)^2d_1^2}{\varepsilon^2}}$.
    \end{itemize}
\end{theorem}

\begin{table}[ht]
\centering
\setlength{\tabcolsep}{2.0mm}{{\renewcommand{\arraystretch}{1.0}
\begin{tabular}{|c|c|c|}
\hline
             & Choi trace norm & Diamond norm \\ \hline
Upper bounds &  \begin{tabular}[c]{@{}c@{}}\noalign{\smallskip}
$O\left(\dfrac{d_1^2}{\varepsilon}+\dfrac{(\tau-1)d_1^2}{\varepsilon^2}\right)$ \\ \noalign{\smallskip}
\cref{coro-2270246} \end{tabular}   & \begin{tabular}[c]{@{}c@{}}\noalign{\smallskip}
$O\parens*{\dfrac{d_1^2}{\varepsilon^2}}$   \\ \noalign{\smallskip}
 \cref{coro-12120134} \end{tabular} \\ \hline
Lower bounds & \begin{tabular}[c]{@{}c@{}}\noalign{\smallskip}
 $\Omega\parens*{\dfrac{d_1^2}{\varepsilon}+\dfrac{(\tau-1)^2}{\varepsilon^2}}$ \\ \noalign{\smallskip}
\cref{thm-390442,coro-3270443} \end{tabular}   & \begin{tabular}[c]{@{}c@{}}\noalign{\smallskip}
$\Omega\parens*{\dfrac{d_1^2}{\varepsilon} +\dfrac{(\tau-1)^2d_1^2}{\varepsilon^2}}$  \\ \noalign{\smallskip}
\cref{thm-390442,thm-390511}\end{tabular}        \\ \hline
\end{tabular}}
}
\caption{Near-boundary regime, i.e., $\tau\in (1,\frac{4}{3})$ can be arbitrarily close to $1$.}
\label{tab:near-boundary}
\end{table}

\Cref{tab:near-boundary} summarizes the results in \Cref{thm:near-boundary}.
We note that the $1/\varepsilon^2$ terms in the lower bounds in the near-boundary regimes imply that the query complexity of quantum channel tomography no longer exhibits Heisenberg scaling $1/\varepsilon$ once it leaves the boundary regime, i.e., even for $\tau$ arbitrarily close to $1$, as long as it is independent of $\varepsilon$.
This transition was illustrated in \Cref{fig:phase-diagram}.

\subsection{Overview of techniques}\label{sec-790210}

\paragraph{Upper bounds.}
Our upper bounds for quantum channel tomography are directly inspired by the recent work of Pelecanos, Spilecki, Tang, and Wright~\cite{pelecanos2025mixedstatetomographyreduces}, who showed that mixed-state tomography can be reduced to pure-state tomography while still achieving optimal performance. 
In this paper, we prove in \cref{thm-1221159} that any channel tomography algorithm that can make parallel queries to a Stinespring dilation of an unknown channel $\mathcal{E}$ can be faithfully simulated by an algorithm that queries only $\mathcal{E}$ itself. 
This enables a reduction from general quantum channel tomography to isometry channel tomography, which is a more tractable problem.
As a consequence, by combining this reduction with existing parallel-query algorithms for isometry tomography~\cite{yang2020optimal,yoshida2025quantum} and with a generalization of the base tomography algorithm from \cite{haah2023query} (see \cref{sec:isometry_channel_tomography}), we obtain the upper bounds presented in this paper.

\cref{thm-1221159} is proved by constructing a ``local'' tester (see \cref{eq-1242355,eq-1242217}) that can faithfully simulate the ``global'' tester with access to a random Stinespring dilation, using representation-theoretic tools together with the quantum comb formalism~\cite{chiribella2009theoretical,bavaresco2022unitary}.
This result can be viewed as a generalization of the earlier \textit{local test} result by \cite{chen2024local}, which studies the local test for quantum states and shows that access to purifications does not help for mixed-state testing.
Related results trace back to~\cite[Theorem 35]{soleimanifar2022testing}, and were recently strengthened in an algorithmic sense by \cite{tang2025conjugate}, which explicitly constructs an algorithm for generating random purifications of a mixed state. 
Intuitively, local test and random purification can be viewed as dual concepts in the Heisenberg and Schr{\"o}dinger pictures, respectively.
More specifically, random purification techniques convert access to a state into access to a corresponding purification.
Conversely, local test techniques convert a global tester (which has access to a purification of a state) into a local tester (which only has access to the state itself), such that the local tester can faithfully simulate the global tester. 
Thus, random purification describes conversions at the level of states, namely in the Schr{\"o}dinger picture, whereas local test describes conversions at the level of testers (i.e., POVMs), namely in the Heisenberg picture. 
This line of work is further extended in \cite{girardi2025random,walter2025random,mele2025random,Girardi2025Dec,yoshida2025random}. 

In particular, \cite{Girardi2025Dec,yoshida2025random} explicitly construct random Stinespring dilation superchannels that convert parallel queries to a quantum channel into parallel queries to its random Stinespring dilation. These results can also be viewed as dual to our local test result for quantum channels (i.e., \cref{thm-1221159}), with respect to the Schr{\"o}dinger and Heisenberg pictures. We remark that, from an algorithmic sense, they additionally provide explicit and efficient circuit implementations of this conversion.

\vspace{5mm}
\begin{figure}[ht]
  \centering
  \def\svgwidth{\linewidth}
  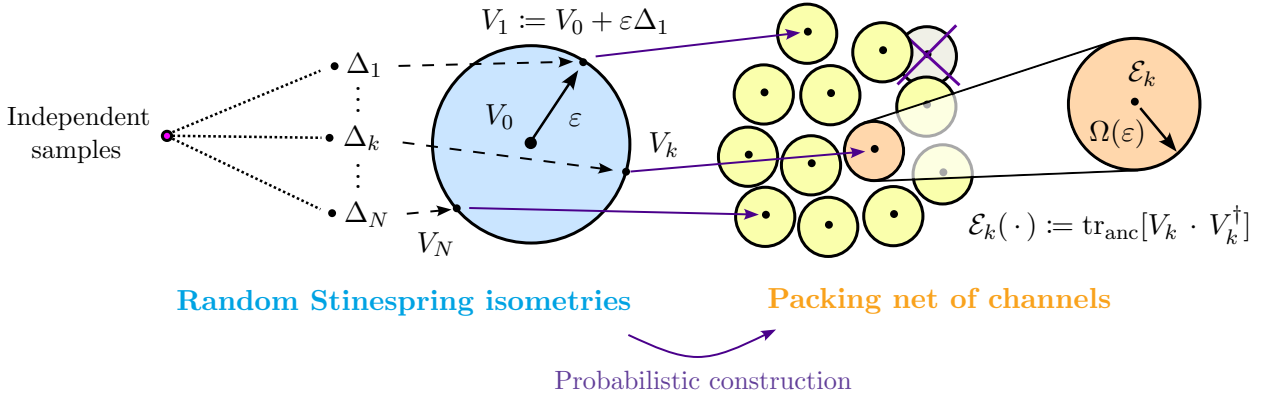
  \caption{A pictorial representation of our constructions of channel packing nets.} 
  \label{fig:construction}
\end{figure}

\paragraph{Lower bounds.} To prove our lower bounds, we proceed using a two-stage method. The first stage consists of constructing sufficiently large packing nets of quantum channels with specific structures and desired properties. 
The second stage consists of proving lower bounds on the query complexity of algorithms that can learn, and thus discriminate between, the quantum channels in the packing nets.

Existing results on nets of the set of isometries \cite{Szarek1997Jan} lead to  packing nets of the set of quantum channels of input dimension $d_1$, output dimension $d_2$ and Kraus rank $r$ of cardinality satisfying $\ln |\calN|=\Omega(rd_1d_2)$. This along with a polynomial method can be used to obtain the general lower bound $\Omega(rd_1d_2)$ \cite{Girardi2025Dec}, which is tight in the dimensions and the Kraus rank but does not capture the dependency on error parameter $\eps$. To make the dependence on $\eps$ precise and optimal, we opt for carefully designed packing nets and a new proof strategy. We distinguish between the boundary and non-boundary regimes.

Roughly speaking, our constructions in the non-boundary regimes share the following idea (see \cref{fig:construction}). To ensure complete-positivity, we construct the channels via their Stinespring isometries. The isometries are chosen of the following  form
\begin{equation}\label{eq:form-construction}
    V = V_0 +\eps \Delta,
\end{equation}
where $V_0$ corresponds to a center map and $\Delta$ is a noisy perturbation. The center supermap $\mathcal{V}_0=V_0(\cdot) V_0^\dag$ admits Kraus operators that are orthogonal and have almost the same Hilbert–Schmidt norm. 
The noisy perturbation is constructed using the Haar randomness applied on certain output space of $\Delta$. The  concentration phenomenon in high-dimensions permits us to prove that two independent channels constructed in this way will be $\eps$-far apart in the Choi trace norm (or diamond norm) with a probability that is at least $1- \exp(-Crd_1d_2)$ with $C$ a universal constant. A probabilistic argument then concludes existence of the packing net of cardinality satisfying $\ln |\calN| = \Omega(rd_1d_2)$. A crucial feature of our construction that leads to classical scaling is that the center and perturbation components have orthogonal images. This can give insights to the phase transition behavior, since at the boundary such construction with the orthogonality  feature is not possible. 

In the boundary regime, the construction has a similar form as in \cref{eq:form-construction}. The center map is close to identity. The perturbation, unlike in the non-boundary regime, is a Haar conjugated anti-Hermitian traceless diagonal matrix with entries $\pm \mathrm{i}$. This construction is inspired by Paninski distribution \cite{Paninski2008Oct} that was used to prove optimality of the mixedness testing problem. 
As stated before, this construction does not exhibit the orthogonality feature. Nevertheless, it has the property that the center and perturbation matrices are orthogonal in Hilbert-Schmidt inner product. This is sufficient to establish the Heisenberg-scaling lower bound.


We now proceed to the second stage of the lower bound proof. We study two structured families of isometries, referred to as type I and type II instances, corresponding to the boundary and non-boundary regimes, respectively.
Both families take the perturbative form in \cref{eq:form-construction}, but with different choices of $V_0$ and $\Delta$. 
Then, we show that discriminating a set of isometries of type I or type II requires at least $\Omega(f/\varepsilon)$ or $\Omega(f/\varepsilon^2)$ queries, respectively, where $f$ depends on both the cardinality of the set and the dimensions of the isometries.
By construction, the channels in the packing nets described above admit Stinespring dilations that are exactly of type I or type II. It follows that the hardness of isometry discrimination transfers directly to channel discrimination: any algorithm that can distinguish among the channels in the packing net can also distinguish among the corresponding dilation isometries, simply by discarding the ancilla system. Since a tomography algorithm with sufficiently high precision can solve the discrimination problem, these query lower bounds also apply to quantum channel tomography.

We now elaborate on the hardness of the isometry discrimination problem.
We formulate the task of discriminating among a set $\mathcal{N}$ of isometries within the quantum comb framework~\cite{chiribella2008quantum,chiribella2009theoretical,bavaresco2022unitary}. In this framework, any discrimination algorithm making $n$ adaptive and coherent queries to an unknown channel can be represented by a sequential tester $\{T_V\}_{V\in\mathcal{N}}$, where $\sum_{V\in\mathcal{N}} T_V$ forms an $(n+1)$-comb with $1$-dimensional input and output spaces. 
Then, the success probability of the discrimination problem is 
$\frac{1}{|\mathcal{N}|}\sum_{V\in\mathcal{N}} T_V\star \kettbbra{V}{V}^{\otimes n}$, where $\star$ denotes the link product (see \cref{def-720255}), representing the contraction between combs and $\kettbbra{V}{V}$ is the Choi operator of the isometry channel $\mathcal{V}$.
The technical heart of the proof is a decomposition of $\kett{V}^{\otimes n}$ into orthogonal sectors indexed by either where (for type I isometry $V$) or how many (for type II isometry $V$) perturbative components (i.e., $\Delta$) appear.
By collecting the Haar average over the orbit of the perturbation for each sector, we can construct an operator $\lambda\cdot \Gamma$ that can universally upper bound (w.r.t. the L\"owner order) $\kettbbra{V}{V}^{\otimes n}$ for all $V\in\mathcal{N}$, where $\lambda>0$ is a scalar depending on $n,d,\varepsilon$ and $\Gamma$ is an $n$-comb. 
Since $\sum_{V\in\mathcal{N}} T_V$ is an $(n+1)$-comb, its contraction with the $n$-comb $\Gamma$ evaluates to $1$.
It follows that the success probability, which must be at least $2/3$, is upper bounded by
\[\frac{1}{|\mathcal{N}|}\sum_{V\in\mathcal{N}} T_V \star (\lambda\cdot \Gamma) = \frac{\lambda}{|\mathcal{N}|}.\]
Solving the inequality $2/3\leq \lambda/|\mathcal{N}|$ yields the desired query lower bounds.

We remark that, when the orthogonal decomposition of $\kett{V}^{\otimes n}$ is based on how many perturbative components appear (the type II case), the property that $\Gamma$ is a comb relies on the orthogonality between the image of the perturbation $\Delta$ and that of the center map $V_0$. 
However, as mentioned above, this is not possible in the boundary regime.
Thus for the type I isometries, the decomposition is based on where perturbative components appear, which is less ``compact'' and ultimately yields a Heisenberg-scaling lower bound of $\Omega(1/\varepsilon)$ rather than the classical-scaling bound $\Omega(1/\varepsilon^2)$.

\subsection{Related work}

\paragraph{Quantum channel tomography.}
The problem of quantum channel tomography (also known as quantum process tomography) has been studied under a variety of models and error metrics. Early works~\cite{chuang1997prescription,poyatos1997complete} consider protocols that prepare input states from a complete set of basis states, perform state tomography on the corresponding output states, and reconstruct the channel using an inversion protocol. Later, \cite{leung2000towards,d2001quantum} use the Choi-Jamio{\l}kowski isomorphism~\cite{choi1975completely,jamiolkowski1972linear} to reduce channel tomography to Choi state tomography. These approaches, together with the direct scheme of \cite{Mohseni2006Oct}, are compared in terms of resource requirements in \cite{mohseni2008quantum}. 

Subsequent works \cite{kahn2007fast,kliesch2019guaranteed,surawy2022projected,haah2023query,Oufkir_2023,rosenthal2024quantum,Girardi2025Dec} focus on obtaining rigorous bounds on the query complexity with explicit dependence on the estimation error. For example, \cite{kliesch2019guaranteed,surawy2022projected} propose randomized protocols for channel tomography and establish query complexity bounds under the Choi-Frobenius norm and the Choi trace norm errors. One of the protocols in \cite{surawy2022projected} is generalized by \cite{Oufkir_2023}, which shows that it achieves near-optimal query complexity for full-Kraus-rank channel tomography under the diamond norm error in the non-adaptive single-copy setting.
Complementing these upper bounds, \cite{rosenthal2024quantum} establishes a query lower bound of $\Omega(d_1^2d_2^2/\log(d_1d_2))$ for full-Kraus-rank (i.e., $r=d_1d_2$) channel tomography under constant Choi trace norm error. Moreover, \cite{Girardi2025Dec} proves a query lower bound of $\Omega(rd_1d_2)$ for channel tomography under constant diamond norm error.

\paragraph{Unitary channel tomography.}
Unitary channel tomography is a special case of general channel tomography where the input and output dimensions are equal and the Kraus rank is $1$. It is studied under different assumptions, models, and error metrics \cite{Acin2001Oct,Peres2002Jul,Bagan2004Sep,Chiribella2004Oct,Hayashi2006May,Chiribella2005Oct,kahn2007fast,yang2020optimal,haah2023query,grewal2025query,yoshida2025quantum}. In particular, \cite{yang2020optimal} and \cite{yoshida2025quantum} study tomography of unitary and isometry channels under average-case channel fidelity; \cite{haah2023query} establishes the optimal query complexity $\Theta(d^2/\varepsilon)$ for unitary channel tomography under the diamond norm.
Moreover, the lower bound in \cite{yoshida2025quantum} with the upper bound in \cite{haah2023query} establishes that pure Heisenberg scaling is impossible for isometry channel tomography, except in the unitary case.

On the upper bound side, our result is based on a reduction from general channel tomography to isometry channel tomography, building on the base unitary tomography algorithm in~\cite{haah2023query} that makes parallel queries.
Through a reduction to the algorithm in~\cite{yang2020optimal}, we further show that Heisenberg scaling can be achieved even for non-unitary channels. 
We also leverage a reduction to the algorithm in~\cite{yoshida2025quantum} to obtain upper bounds for channel tomography under the Choi trace norm error in the near-boundary regime.
On the lower bound side, new packing-net constructions are needed to capture the full dependence on the channel dimensions, Kraus rank, and error. In contrast to \cite{haah2023query}, our packing nets have additional structure and leverage symmetry arguments. 

\paragraph{Quantum state tomography.}
Quantum state tomography (or more generally, quantum state learning~\cite{arunachalam2017guest,anshu2023survey}) is a special case of quantum channel tomography where the input dimension is $1$. The optimal sample complexity of $d$-dimensional rank-$r$ quantum state tomography under the trace norm error is established as $\tilde{\Theta}(dr/\varepsilon^2)$ in \cite{10.1145/2897518.2897544,Haah_2017}, and the logarithmic factor is removed only recently by \cite{scharnhorst2025optimal}. As noted above, quantum state tomography algorithms can be used to perform channel tomography via various reductions. In particular, pure state tomography \cite{CL14,KRT17,GKKT20} is a key ingredient in the optimal algorithm for unitary channel tomography in~\cite{haah2023query}. Moreover, using a different proof strategy, our results recover the upper bound of \cite{10.1145/2897518.2897544,Haah_2017} and the recently established lower bound of \cite{scharnhorst2025optimal} as special cases.

\paragraph{Independent work.}
Mele and Bittel~\cite{mele2025optimallearningquantumchannels} established the same upper bound $O(rd_1d_2/\varepsilon^2)$ for quantum channel tomography under diamond norm error. 
Their upper bound is concurrent and independent to Corollary 3.4 in the earlier paper~\cite{chen2025quantum}, which is subsumed by the present paper and restated here as \cref{coro-12120134}.
They also derived an explicit and non-trivial dependence on the failure probability.
We note that their method and ours differ substantially.
Specifically, they obtain the upper bound through an analysis of the tomography of Choi states, which leads to an explicit tomography algorithm.
In contrast, our approach uses the local test technique to simulate access to dilations of quantum channels, allowing us to reduce general channel tomography to a more tractable problem of isometry channel tomography.
Moreover, we establish additional upper bounds achieving Heisenberg scaling $1/\varepsilon$ for (possibly non-unitary) channels in the boundary regime $\tau=1$.

They also established a query lower bound. When the failure probability is constant, their lower bound is $\Omega\parens*{\frac{rd_1d_2}{\varepsilon^{\beta}} + \frac{m\log (3m)}{\varepsilon^2}}$ for diamond norm error $\varepsilon$, where $\beta=\frac{2rd_1d_2-d_1^2-r^2}{2(rd_1d_2-1)}$ and 
\begin{equation*}
    m=\max\braces*{\min\braces*{d_1, \floor*{\frac{r-\ceil*{d_1/d_2}}{2}}}, \,\,\mathbbm{1}_{r\geq \ceil*{d_1/d_2} + 1}, \,\, d_1 \cdot \mathbbm{1}_{r\geq 2 d_1}}.
\end{equation*}
We can see that $\beta\in (\frac{3}{8}, 1]$ and $m\in [0,d_1]$. 
Moreover, when $r \ge \lceil d_1/d_2 \rceil + 1$,\footnote{That is, just above the minimal Kraus rank $\lceil d_1/d_2 \rceil$ that a valid quantum channel can have, given the input dimension $d_1$ and output dimension $d_2$.} it follows that $m \ge 1$, thereby showing that pure Heisenberg scaling is unattainable for any quantum channel whose Kraus rank is not minimal.
However, when $d_1$ is not a multiple of $d_2$ and $r= \lceil d_1/d_2\rceil$ is minimal, their lower bound does not rule out Heisenberg scaling, although this case still lies in the non-boundary regime ($\tau>1$).
By comparison, our lower bounds rule out pure Heisenberg scaling across the entire non-boundary regime, and 
exhibit sharper parameter dependence
in all regimes (in particular, our lower bounds are optimal in the away-from-boundary regime).



\subsection{Discussion}
\label{sub:discussion}

In this paper, we study the query complexity of quantum channel tomography.
Our results answer \Cref{ques:main} as follows.
We identify the dilation rate $\tau = rd_2/d_1$ as a key parameter.
We show that when $\tau=1$, the optimal query complexity exhibits Heisenberg scaling $1/\varepsilon$; whereas for $\tau>1$, it no longer has Heisenberg scaling due to our lower bounds.
Moreover, we establish multiple optimal bounds in the boundary regime ($\tau=1$) and the away-from-boundary regime ($\tau \geq 1+\Omega(1)$), as well as new upper and lower bounds in the near-boundary regime ($1< \tau < 1+ o(1)$).
We conclude by highlighting several directions for future work.

A natural open problem is to determine the optimal query complexity of quantum channel tomography across all parameter regimes. This includes sharpening the upper and lower bounds in the near-boundary regime, as well as improving the bounds in the boundary regime for diamond norm error.

A second question, closely related to the first, is whether \Cref{thm-2150136} can be extended to sequential testers.
More precisely, suppose a tester that makes an arbitrary sequence of queries, not necessarily in parallel, to an arbitrary dilation of an unknown quantum channel $\calE$. Can one then construct another tester that makes the same number of queries to $\calE$ itself and solves the same estimation task? 

\paragraph{Recent developments.}
After the work in this paper, there have been several recent developments.
The second version of \cite{yoshida2025random} gives an approximate sequential dilation strategy with query overhead and also proves a no-go theorem for exact sequential dilation.
\cite{he2026optimal} gives a unitary tomography algorithm in diamond norm using $O(d^2/\varepsilon)$ parallel queries. This, combined with \cref{thm-2150136}, yields a query upper bound $O(rd_1d_2/\varepsilon)$ for diamond-norm boundary-regime channel tomography, matching our lower bound $\Omega(rd_1d_2/\varepsilon)$ in \cref{thm-2150126}.

\subsection{Organization}

The rest of the paper is organized as follows.
\begin{itemize}
\item In \cref{sec-770122}, we introduce the necessary preliminaries and background of the techniques used in this paper.
\item In \cref{sec-770126}, we use local test techniques to prove that dilation does not help in quantum channel tomography and estimation. 
This can reduce channel tomography to isometry tomography and thus be combined with the isometry tomography algorithms in \cref{sec:isometry_channel_tomography} to obtain our upper bounds.
\item In \cref{sec-3200450}, we define two classes of isometries, called type I and type II isometries, and prove Heisenberg-scaling and classical-scaling hardness results for the corresponding discrimination tasks, respectively.
\item In \cref{sec-3290355}, we construct sufficiently large packing nets of type I and type II isometries such that, even after tracing the ancilla system, the resulting channels remain well separated in Choi trace norm or diamond norm.
\item In \cref{sec-3200555}, we combine the hardness results from \cref{sec-3200450} and the packing nets from \cref{sec-3290355} to obtain our lower bounds.
\item In \cref{sec-770141}, we prove some technical lemmas deferred from earlier sections.
\end{itemize}

\section{Preliminaries}\label{sec-770122}

\paragraph{Notation.}
Given a Hilbert space $\mathcal{H}$, we denote by $\mathcal{L}(\mathcal{H})$ the set of linear operators on $\mathcal{H}$. Similarly, we denote by $\mathcal{L}(\mathcal{H}_0,\mathcal{H}_1)$ the set of linear operators from $\mathcal{H}_0$ to $\mathcal{H}_1$. Given two orthonormal bases for $\mathcal{H}_0$ and $\mathcal{H}_1$, respectively, we can represent every linear operator in $\mathcal{L}(\mathcal{H}_0,\mathcal{H}_1)$ by a $\dim(\mathcal{H}_1)\times \dim(\mathcal{H}_0)$ matrix, and, for such a matrix \(X\), we use \(\kett{X}\in \mathcal{H}_1\otimes\mathcal{H}_0\) to denote the vector obtained by flattening the matrix $X$. It is easy to see the following properties:
\[\kett{\ketbra{\psi}{\phi}}=\ket{\psi}\ket{\phi^*}, \quad\quad\quad \kett{XYZ}=X\otimes Z^\textup{T} \kett{Y},\]
where \(\ket{\phi^*}\) denotes the entry-wise complex conjugate of \(\ket{\phi}\)  and $Z^\textup{T}$ denotes the transpose of the matrix $Z$ with respect to a given orthonormal basis. 
The inner product can be rewritten as \(\bbrakett{X}{Y}=\tr(X^\dag Y)\).
For two  Hermitian operators $X,Y$, we write $X\leq Y$ to denote that $Y-X$ is positive semidefinite.

Let $n, m, d$ be positive integers such that $n\geq m$. Let $\mathcal{H}_1\cong\cdots\cong \mathcal{H}_n\cong\mathbb{C}^{d}$ be $n$ copies of a $d$-dimensional Hilbert space $\mathcal{H}$. 
Let $S\subseteq [n]=\{1,2,\ldots,n\}$ be a set of integers and let $\ket{\psi}\in\mathbb{C}^d$ be a state. We denote by
\[\ket{\psi}^{\otimes S}\]
the state $\ket{\psi}^{\otimes |S|}\in\bigotimes_{i\in S} \mathcal{H}_i$.
Therefore, if $\ket{\varphi}\in \mathbb{C}^d$ is another state, then we call
\[\ket{\psi}^{\otimes S} \otimes\ket{\varphi}^{\otimes [n]\setminus S}\]
the state $\bigotimes_{i=1}^n \ket{x_i}\in\bigotimes_{i=1}^n \mathcal{H}_i$, where $\ket{x_i}=\ket{\psi}$ for $i\in S$, and $\ket{x_i}=\ket{\varphi}$ for $i\notin S$.

\begin{table}[t]
    \begin{tabularx}{\textwidth}{p{0.19\textwidth}X>{\raggedleft\arraybackslash}l} 
    \toprule
    Symbol & Description & Introduced in \\
    \midrule
      $C_{\mathcal{E}}$ & Choi-Jamio{\l}kowski operator of the channel $\mathcal{E}$, (i.e., unnormalized Choi-Jamio{\l}kowski state). & \cref{eq:Choi}\\
      $\qchannel_{d_1,d_2}^r$ & Set of all quantum channels $\mathcal{E}:\mathcal{L}(\mathbb{C}^{d_1})\rightarrow \mathcal{L}(\mathbb{C}^{d_2})$ that have Kraus rank at most $r$. & \cref{not:notationQChan}\\ 
      $\isochannel_{d_1,d_2}$ & Set of isometry channels with input dimension $d_1$ and output dimension $d_2$, which is equivalent to $\qchannel_{d_1,d_2}^1$ & \textit{Ibid.}\\
      $\dilation_{r}(\mathcal{E})$ & Set of all dilations of the quantum channel $\mathcal{E}$ with an ancilla system of dimension $r$. & \cref{not:dilations} \\
      $\contract_{\mathcal{H}_3}(\mathcal{V})$ & Quantum channel defined as the contraction \mbox{$\rho\mapsto \tr_{\mathcal{H}_3}(V\rho V^\dag)$}, where $V:\mathcal{H}_1\rightarrow\mathcal{H}_2\otimes\mathcal{H}_3$ \mbox{is an isometry.} & \textit{Ibid.}\\
      $\mathcal{V}\sim\dilation_r(\mathcal{E})$ & The dilation $\mathcal{V}$ is sampled from Haar distributions on $\dilation_r(\mathcal{E})$. & \cref{not:Haar}\\
      $U\sim \mathbb{U}_d$ & The unitary $U$ is sampled from Haar distributions on $\mathbb{U}_d$. & \textit{Ibid.} \\
      $X\star Y$ & Link product between $X$ and $Y$. & \cref{def-720255} \\
      $\{T_i\}_i$ & Quantum channel tester. & \cref{sub:tester} \\
      $\{A_{\mathcal{E}}\}_{\mathcal{E}\in\qchannel_{d_1,d_2}^r}$ & Estimation task of quantum channels, i.e., set of classical outcomes that are regarded as correct answers when the unknown channel is $\mathcal{E}$. & \cref{not:est_task}\\
      \bottomrule
     \end{tabularx}
         \caption{Summary of the notation used in the paper.}
    \label{table1}
\end{table}

\subsection{Quantum channels}

A \emph{quantum channel} with input dimension $d_1$ and output dimension $d_2$ is a linear, completely positive and trace-preserving map $\mathcal{E}:\mathcal{L}(\mathbb{C}^{d_1})\rightarrow\mathcal{L}(\mathbb{C}^{d_2})$ (see, e.g., \cite{NC10,watrous2018theory,hayashi2017quantum}), which is also called a CPTP map. 

Any quantum channel $\mathcal{E}$, in the Kraus representation~\cite{kraus1983states}, can be written as
\[\mathcal{E}(\rho)=\sum_{i=1}^r E_i \rho E_i^\dag,\]
where $E_i: \mathbb{C}^{d_1}\rightarrow\mathbb{C}^{d_2}$ are non-zero linear operators satisfying
\[\sum_{i=1}^r E_i^\dag E_i=I_{d_1} \qquad\text{and}\qquad \tr(E_i^\dag E_j)=0 \quad \text{for} \quad i\neq j.\]
The integer $r$ is called \emph{Kraus rank} and $\{E_i\}$ are called \emph{Kraus operators}.
Since, for each $i$, $E_i^\dag E_i$ has rank at most $d_2$, and since these operators sum to $I_{d_1}$, one can easily see that $rd_2\geq d_1$.
In particular, a quantum channel having Kraus rank $r=1$ is an isometry $\mathcal{V}(\,\cdot\,)=V\,\cdot\, V^\dag$, where $V:\mathbb{C}^{d_1}\rightarrow\mathbb{C}^{d_2}$ is an isometry operator -- i.e. it satisfies $V^\dag V=I_{d_1}$ -- and it must hold that $d_2\geq d_1$.

\begin{notation}\label{not:notationQChan}
We denote by $\qchannel_{d_1,d_2}^r$ the set of all quantum channels $\mathcal{E}:\mathcal{L}(\mathbb{C}^{d_1})\rightarrow \mathcal{L}(\mathbb{C}^{d_2})$ having Kraus rank at most $r$.
In particular, we denote by $\isochannel_{d_1,d_2}$ the set of all isometry channels with input dimension $d_1$ and output dimension $d_2$, which is corresponds to $\qchannel_{d_1,d_2}^1$. 
\end{notation}

In the Choi-Jamio{\l}kowski representation~\cite{choi1975completely,jamiolkowski1972linear}, any channel $\mathcal{E}$ can be identified by its Choi-Jamio{\l}kowski operator as follows:
\begin{equation}\label{eq:Choi}
    C_\mathcal{E}=(\mathcal{E}\otimes \mathcal{I})(\kettbbra{I}{I})\in\mathcal{L}(\mathbb{C}^{d_2}\otimes\mathbb{C}^{d_1}),
\end{equation}
where we have denoted by $\kett{I}=\sum\limits_{i}\ket{i}\ket{i}\in \mathbb{C}^{d_1}\otimes \mathbb{C}^{d_1}$ an unnormalized maximally entangled state. We may simply call $C_\mathcal{E}$ Choi operator or (unnormalized) Choi state. 
Note that it it possible to write the Choi operator as $C_\mathcal{E}=\sum_{i=1}^r\kettbbra{E_i}{E_i}$, where $E_i$ are orthogonal Kraus operators and thus $\kett{E_i}$ are pairwise orthogonal vectors. Hence, the Kraus rank is equal to the rank of the Choi operator. As a consequence, one can easily see that $r\leq d_1d_2$.

\paragraph{Stinespring dilation.}
Given a quantum channel $\mathcal{E}$ with Kraus operators $\{E_i\}_{i=1}^r$, using its Stinespring dilation~\cite{stinespring1955positive}, we can also write $\mathcal{E}$ as
\begin{equation}\label{eq-1230155}
\mathcal{E}(\rho)=\tr_{\mathcal{H}_\mathrm{anc}}(V\rho V^\dag),
\end{equation}
where $\mathcal{H}_\mathrm{anc}\cong \mathbb{C}^{r}$ is an ancilla system and $V=\sum_{i=1}^r\ket{i}_\mathrm{anc}\otimes E_i$ is an isometry operator. Any isometry channel \mbox{$\mathcal{V}(\,\cdot\,)=V\,\cdot\,V^\dag$} satisfying \cref{eq-1230155} is called a dilation of $\mathcal{E}$. 
Two isometries $\mathcal{V}_1$ and $\mathcal{V}_2$ are two dilations of the same channel $\mathcal{E}$ if and only if they differ by a unitary on $\mathcal{H}_\mathrm{anc}$, namely, $V_2=(U\otimes I_{d_2}) V_1$ for a unitary $U:\mathcal{H}_\mathrm{anc}\rightarrow\mathcal{H}_\mathrm{anc}$.

\begin{notation}\label{not:dilations}
Given a quantum channel $\mathcal{E}$ with Kraus rank at most $r$, we denote by $\dilation_{r}(\mathcal{E})$ the set of all dilations of $\mathcal{E}$ with an ancilla system of dimension $r$. 
Given an isometry channel $\mathcal{V}:\mathcal{L}(\mathcal{H}_1)\rightarrow\mathcal{L}(\mathcal{H}_2\otimes\mathcal{H}_3)$, we denote by $\contract_{\mathcal{H}_3}(\mathcal{V})$ the quantum channel
\[\rho\mapsto \tr_{\mathcal{H}_3}(V\rho V^\dag).\]
\end{notation}

\paragraph{Haar distribution.}


Given any arbitrary quantum channel $\mathcal{E}\in\qchannel_{d_1,d_2}^r$, we define the Haar distribution on $\dilation_r(\mathcal{E})$ as follows: choose an arbitrary dilation $\mathcal{V}\in\dilation_r(\mathcal{E})$, and output $(\mathcal{U}\otimes \mathcal{I}_{d_2})\circ\mathcal{V}$, where $U \in\mathbb{U}_r$ is a Haar random unitary.
This procedure is well defined since the resulting distribution does not depend on the choice of the initial dilation $\mathcal{V}$. It is easy to see that such distribution is invariant under $\mathbb{U}_{r}$, i.e.
\[\Pr[A]=\Pr[\{(\mathcal{U}\otimes \mathcal{I}_{d_2})\circ\mathcal{V}\,|\, \mathcal{V}\in A\}],\]
for any unitary $U\in\mathbb{U}_r$ and any measurable set $A\subseteq \dilation_r(\mathcal{E})$.

\begin{notation}\label{not:Haar}
We denote by $\mathcal{V}\sim\dilation_r(\mathcal{E})$ and $U\sim \mathbb{U}_d$ two random variables $\mathcal{V}$ and $U$ that are sampled from the Haar distributions on $\dilation_r(\mathcal{E})$ and $\mathbb{U}_d$, respectively.
\end{notation}

\paragraph{Distance measures.}
In this paper, we are going to use the Choi trace norm and the diamond norm~\cite{AKN98,watrous2018theory} as the \textit{average-case} and \textit{worst-case} measures between quantum channels, respectively.
\begin{definition}\label{def-3301809}
Let $\mathcal{E},\mathcal{F}\in\qchannel_{d_1,d_2}^r$ be two quantum channels.
We define the \textit{Choi trace norm} between $\mathcal{E}$ and $\mathcal{F}$ as
\[\frac{1}{d_1}\|C_\mathcal{E}-C_\mathcal{F}\|_1,\]
where $\|\cdot\|_1$ denotes the Schatten $1$-norm (i.e. trace norm) and $C_{\mathcal{E}}$ is the (unnormalized) Choi state of $\mathcal{E}$.
We define the \textit{diamond norm} between $\mathcal{E}$ and $\mathcal{F}$ as
\[\|\mathcal{E}-\mathcal{F}\|_\diamond\coloneqq \sup_{\rho}\|(\mathcal{E}\otimes \mathcal{I}_{d_1})(\rho)-(\mathcal{F}\otimes\mathcal{I}_{d_1})(\rho)\|_1,\]
where the supremum is taken over all quantum states $\rho\in \mathbb{C}^{d_1}\otimes\mathbb{C}^{d_1}$.
\end{definition}

A simple relation between the diamond norm of quantum channels and trace norm of the corresponding Choi-Jamio{\l}kowski operators is the following: given two quantum channels $\mathcal{E},\mathcal{F}\in \qchannel_{d_1,d_2}^r$, we have (see, e.g., \cite[Proposition 50]{kliesch2021theory}):
\begin{align}
\frac{1}{d_1}\|C_\mathcal{E}-C_\mathcal{F}\|_1\leq \|\mathcal{E}-\mathcal{F}\|_\diamond 
\leq \|C_\mathcal{E}-C_\mathcal{F}\|_1.\label{eq-3302201}
\end{align}
As an average-case measure, the Choi trace norm is widely used in learning of quantum channels~\cite{kliesch2019guaranteed,surawy2022projected} and also in other tasks in quantum information theory~\cite{gs007,kliesch2021theory,rosenthal2024quantum}.
Moreover, we will also consider the \textit{channel fidelity}~\cite{raginsky2001fidelity} (or \textit{entanglement fidelity}~\cite{nielsen2002simple}). Specifically, given two quantum channels $\mathcal{E},\mathcal{F}\in \qchannel_{d_1,d_2}^r$, their channel fidelity $\mathrm{F}_\textup{ch}$ is defined as the fidelity of their normalized Choi states:
\[\mathrm{F}_\textup{ch}(\mathcal{E},\mathcal{F})\coloneqq \mathrm{F}\!\left(\frac{C_\mathcal{E}}{d_1},\frac{C_\mathcal{F}}{d_1}\right),\]
where $\mathrm{F}$ denotes the fidelity between quantum states.
Note that channel fidelity is closely related to the Choi trace norm by a straightforward lifting of the relation between fidelity and trace norm for quantum states.

\subsection{Quantum combs}\label{sec-6290132}
The framework of quantum combs~\cite{chiribella2008quantum,chiribella2009theoretical} provide a powerful description of higher-order transformations of quantum processes. More specifically, the Choi-Jamio{\l}kowski representation of quantum channels, which describes transformations of quantum states, can be extended to the higher-order setting of transformations of quantum processes. Such higher-order objects are called \textit{quantum combs}. More precisely, deterministic and probabilistic quantum combs are defined as follows.

\begin{definition}[Deterministic comb~\cite{chiribella2009theoretical}]\label{def-681501}
Let $n\geq 1$ be an integer. A deterministic $n$-comb, defined on a sequence of $2n$ Hilbert spaces $(\mathcal{H}_0,\mathcal{H}_1,\ldots,\mathcal{H}_{2n-1})$, is a positive semidefinite operator $X$ on $\bigotimes_{j=0}^{2n-1} \mathcal{H}_{j}$ such that there exists a sequence of operators $X^{(n)}, X^{(n-1)},\ldots, X^{(1)}, X^{(0)}$ satisfying
\begin{equation}\label{eq-681546}
\begin{split}
\tr_{\mathcal{H}_{2j-1}}\!\left(X^{(j)}\right)&=I_{\mathcal{H}_{2j-2}}\otimes X^{(j-1)},\quad 1\leq j \leq n,  
\end{split}
\end{equation}
where $X^{(n)}=X$ and $X^{(0)}=1$.
\end{definition}
\begin{definition}[Probabilistic comb~\cite{chiribella2009theoretical}]
Let $n\geq 1$ be an integer. A probabilistic $n$-comb, defined on a sequence of $2n$ Hilbert spaces $(\mathcal{H}_0,\mathcal{H}_1,\ldots,\mathcal{H}_{2n-1})$, is a positive semidefinite operator $X$ on $\bigotimes_{j=0}^{2n-1}\mathcal{H}_j$ such that $X\leq Y$ for some deterministic $n$-comb $Y$ on $(\mathcal{H}_0,\mathcal{H}_1,\ldots,\mathcal{H}_{2n-1})$.
\end{definition}
We can easily see the following fact.
\begin{proposition}\label{prop-4290313}
Let $n\geq 1$ be an integer. A positive semidefinite operator $X$ is a probabilistic $n$-comb on $(\mathcal{H}_0,\mathcal{H}_1,\ldots,\mathcal{H}_{2n-1})$ if and only if there exists a probabilistic $(n-1)$-comb $Y$ on $(\mathcal{H}_0,\mathcal{H}_1,\ldots,\mathcal{H}_{2n-3})$ such that
$\tr_{\mathcal{H}_{2n-1}}\!\left(X\right)\leq I_{\mathcal{H}_{2n-2}}\otimes Y$.
\end{proposition}
\begin{remark}
In this paper, quantum comb refers to a deterministic comb by default, and the term probabilistic comb will be used explicitly when needed.
\end{remark}
It is easy to see verify following facts:
\begin{itemize}
    \item a quantum $1$-comb is the Choi-Jamio{\l}kowski operator of a quantum channel;
    \item any convex combination of quantum $n$-combs is also a quantum $n$-comb.
\end{itemize}

Now, let us introduce the link product ``$\star$''.
\begin{definition}[Link product ``$\star$''~\cite{chiribella2008quantum,chiribella2009theoretical}]
\label{def-720255}
Let $X$ be a linear operator on $\mathcal{H}_{\bm{i}}=\mathcal{H}_{i_1}\otimes\mathcal{H}_{i_2}\otimes\cdots\otimes\mathcal{H}_{i_{n}}$ and let $Y$ be a linear operator on $\mathcal{H}_{\bm{j}}=\mathcal{H}_{j_1}\otimes\mathcal{H}_{j_2}\otimes\cdots\otimes\mathcal{H}_{j_{m}}$,
where $\bm{i}=(i_1,\ldots,i_n)$ is a sequence of pairwise distinct indices, and likewise for $\bm{j}=(j_1,\ldots,j_m)$.
Let $\bm{a}=\bm{i}\cap\bm{j}$ be the set of indices appearing in both $\bm{i}$ and $\bm{j}$ and $\bm{b}=\bm{i}\cup\bm{j}$ be the set of indices appearing in either $\bm{i}$ or $\bm{j}$.
Then, the combination of $X$ and $Y$ is
\[X\star Y\coloneqq \tr_{\mathcal{H}_{\bm{a}}}\!\left(X^{\textup{T}_{\mathcal{H}_{\bm{a}}}} \cdot Y\right)=\tr_{\mathcal{H}_{\bm{a}}}\!\left(X\cdot Y^{\textup{T}_{\mathcal{H}_{{\bm{a}}}}}\right),\]
where $\mathcal{H}_{\bm{a}}$ denotes the tensor product of subsystems labeled by the indices in $\bm{a}$, $\textup{T}_{\mathcal{H}_{\bm{a}}}$ denotes the partial transpose on $\mathcal{H}_{\bm{a}}$, both $X$ and $Y$ are treated as linear operators on $\mathcal{H}_{\bm{b}}$, extended by tensoring with the identity operator as needed.
\end{definition}

\begin{figure}[t]
    \centering
    \includegraphics[width=0.6\linewidth]{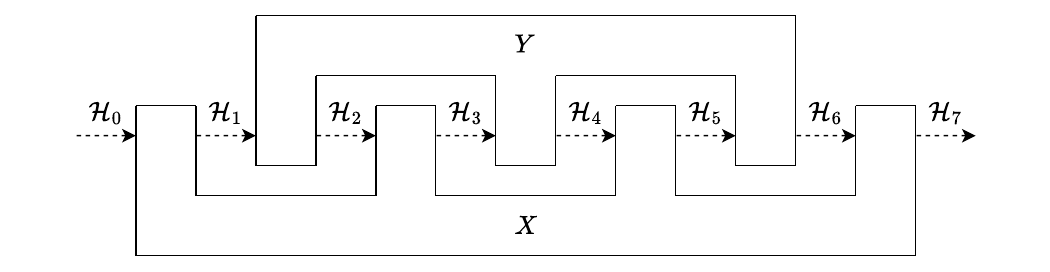}
    \caption{The combination of a $4$-comb $X$ with a $3$-comb $Y$, yielding a $1$-comb $X\star Y$ on $(\mathcal{H}_0,\mathcal{H}_7)$.}
    \label{fig-6210015}
\end{figure}

The link product provides the mathematical description of the combination of quantum combs. For instance, suppose $X$ is an $n$-comb on $(\mathcal{H}_0,\mathcal{H}_1,\ldots,\mathcal{H}_{2n-1})$ and $Y$ is an $(n-1)$-comb on $(\mathcal{H}_1,\mathcal{H}_2\,\ldots,\mathcal{H}_{2n-2})$. Then,
\begin{align}
X\star Y= \tr_{\mathcal{H}_{1:2n-2}}\!\left(X^{\textup{T}_{\mathcal{H}_{1:2n-2}}}\cdot (I_{\mathcal{H}_{2n-1}}\otimes Y \otimes I_{\mathcal{H}_0})\right)=\tr_{\mathcal{H}_{1:2n-2}}\!\left(X\cdot (I_{\mathcal{H}_{2n-1}}\otimes Y^{\textup{T}} \otimes I_{\mathcal{H}_0})\right)\nonumber
\end{align}
turns out to be a $1$-comb on $(\mathcal{H}_0,\mathcal{H}_{2n-1})$, as in the example illustrated in \cref{fig-6210015}.
The link product also has many useful properties:
\begin{itemize}
    \item it preserves the L\"owner order: if $X,Y\ge 0$ then $X\star Y\ge 0$~\cite[Theorem 2]{chiribella2009theoretical};
    \item it is commutative, namely $X\star Y=Y \star X$, and associative, that is, $(X\star Y)\star Z=X\star (Y\star Z)$, whenever $X,Y,Z$ do not share a common subsystem (i.e., there is no subsystem that is a subsystem of all three).
\end{itemize}

Moreover, the link product characterizes the channel concatenation under the Choi representation. Namely, given two quantum channels $\mathcal{E}_1:\mathcal{L}(\mathcal{H}_1)\rightarrow\mathcal{L}(\mathcal{H}_2)$ and $\mathcal{E}_2:\mathcal{L}(\mathcal{H}_2)\rightarrow\mathcal{L}(\mathcal{H}_3)$, we have 
$C_{\mathcal{E}_2\circ \mathcal{E}_1}=C_{\mathcal{E}_2}\star C_{\mathcal{E}_1}$, where $C_\mathcal{E}$ is the Choi operator of $\mathcal{E}$.

\subsection{Formalism of quantum channel testers}
\label{sub:tester}

A \textit{quantum channel tester} is any quantum algorithm that can make multiple queries to an unknown quantum channel in order to produce a classical output.
We are going to use the quantum tester formalism based on Choi-Jamio{\l}kowski representation (see, e.g., \cite{chiribella2009theoretical,bavaresco2021strict,bavaresco2022unitary}), which provides a practical framework in order to study various classes of quantum testers, such as parallel and sequential ones.

Consider a quantum channel tester using $n$ queries to an unknown quantum channel $\mathcal{E}$. 
Let us label the input and output systems of the $i$-th query to $\mathcal{E}$ as $\mathcal{H}_{\mathrm{A},i}$ and $\mathcal{H}_{\mathrm{B},i}$, i.e.\ the $i$-th copy of the unknown channel is a linear map from $\mathcal{L}(\mathcal{H}_{\mathrm{A},i})$ to $\mathcal{L}(\mathcal{H}_{\mathrm{B},i})$.

\paragraph{Parallel testers.}
A parallel tester can only make parallel queries. 
Specifically, it prepares a multipartite input state, possibly including auxiliary systems, and applies the unknown channel in parallel to its subsystems, ensuring that the output of any use never interacts with the inputs of the others. Eventually, after all channel uses, a single joint measurement is performed on the combined output state.

\begin{definition}[Parallel tester]\label{def:parallel_tester}
A parallel tester is defined as any set of linear operators $\{T_i\}_i$, with $T_i\in \mathcal{L}(\bigotimes_{j=1}^n \mathcal{H}_{\mathrm{A},j}\otimes \mathcal{H}_{\mathrm{B},j})$, such that $T_i\geq 0$ and $\sum_{i} T_i=\rho_{\mathrm{A}}\otimes I_{\mathrm{B}}$, where $I_{\mathrm{B}}$ is the identity operator on $\bigotimes_{j=1}^n\mathcal{H}_{\mathrm{B},j}$, and $ \rho_{\mathrm{A}}$ is a positive semidefinite operator on $\bigotimes_{j=1}^n\mathcal{H}_{\mathrm{A},j}$ with $\tr(\rho_\mathrm{A})=1$.
\end{definition}

When applying a parallel tester $\{T_i\}_i$ to a quantum channel $\mathcal{E}$, we get the classical outcome $i$ with probability 
\begin{equation}\label{eq-1211852}
p_i=T_i\star C_\mathcal{E}^{\otimes n}= \tr(T_i (C_{\mathcal{E}}^{\otimes n})^\mathrm{T})=\tr(T_i^\mathrm{T} C_{\mathcal{E}}^{\otimes n}),
\end{equation}
where $C_\mathcal{E}^{\otimes n}\in\mathcal{L}(\bigotimes_{j=1}^n \mathcal{H}_{\mathrm{A},j}\otimes\mathcal{H}_{\mathrm{B},j})$
denotes the Choi operator of all $n$ queries to the channel $\mathcal{E}$ and $(    \,\cdot \,)^{\mathrm{T}}$ is matrix transposition.

To see how the parallel tester $\{T_i\}_i$ can be realized by an algorithm that makes queries in parallel, we consider the following procedure.
\begin{itemize}
    \item Assume $\sum_i T_i=\rho_\mathrm{A}\otimes I_{\mathrm{B}}$; prepare a quantum state $(\sqrt{\rho_\mathrm{A}}^\textup{T}\otimes I_{\mathrm{A}})\kett{I_{\mathrm{A}}}$ in $\bigotimes_{j=1}^n\mathcal{H}_{\mathrm{A},j}\otimes\bigotimes_{j=1}^n\mathcal{H}_{\mathrm{A},j}$. This is indeed a valid quantum state, since
    $\bbra{I_{\mathrm{A}}} (\rho_\mathrm{A}^\mathrm{T}\otimes I_{\mathrm{A}})\kett{I_{\mathrm{A}}}=\tr(\rho_\mathrm{A}^{\textup{T}})=1$.
    
    \item Then, apply the quantum channel $\mathcal{I}_\mathrm{A}\otimes \mathcal{E}^{\otimes n}$ to the prepared state, obtaining the mixed state $(\sqrt{\rho_\mathrm{A}}^\mathrm{T}\otimes I_{\mathrm{B}}) C_\mathcal{E}^{\otimes n} (\sqrt{\rho_\mathrm{A}}^\mathrm{T}\otimes I_{\mathrm{B}})$.

    \item Finally, perform the POVM $\left\{(\sqrt{\rho}_\mathrm{A}^\mathrm{T}\otimes I_{\mathrm{B}})^{-1} \, T_i^\mathrm{T}\, (\sqrt{\rho}_\mathrm{A}^\mathrm{T}\otimes I_{\mathrm{B}})^{-1}\right\}_i$, obtaining the result $i$, where $(\cdot)^{-1}$ is the pseudo-inverse. Then, it is easy to see that the probability of getting result $i$ is exactly that in \cref{eq-1211852}.
\end{itemize}
Conversely, all algorithms that make queries in parallel can be described by a parallel tester. To see this, assume that the algorithm first prepares a state $\rho$ on $(\bigotimes_{j=1}^n \mathcal{H}_{\mathrm{A},j})\otimes \mathcal{H}_\mathrm{anc}$, where $\mathcal{H}_{\mathrm{anc}}$ is an ancilla system, and then apply the channel $\mathcal{E}^{\otimes n}\otimes \mathcal{I}_\mathrm{anc}$ on $\rho$ followed by a POVM $\{E_i\}_i$, where each $E_i \in\mathcal{L}((\bigotimes_{j=1}^n \mathcal{H}_{\mathrm{B},j})\otimes\mathcal{H}_\mathrm{anc})$ is positive semidefinite. 
Then, let $T_i=E_i^\mathrm{T}\star \rho$. We see that $\{T_i\}_i$ is a parallel tester, and the probability of obtain outcome $i$ is given by
\[\tr(E_i\cdot (\mathcal{E}^{\otimes n}\otimes\mathcal{I}_\mathrm{anc})(\rho))=\tr(E_i \cdot (C_\mathcal{E}^{\otimes n}\star \rho))=E_i^\mathrm{T}\star C_\mathcal{E}^{\otimes n} \star \rho=T_i\star C_\mathcal{E}^{\otimes n},\]
which is exactly the same as \cref{eq-1211852}.



\paragraph{Sequential testers.}
A sequential tester can make queries sequentially, adaptively and coherently. 
Specifically, it sends a quantum system through the first use of the channel $\mathcal{E}$ and then it feeds the resulting (quantum) output into subsequent uses, potentially along with auxiliary systems, while allowing arbitrary CPTP maps to act between uses of $\mathcal{E}$. After all $n$ uses of the channel $\mathcal{E}$, a POVM is performed on the final state. In other words, sequential testers represent coherent and adaptive query-access algorithms.

\begin{definition}[Sequential tester]
A sequential tester that uses $n$ queries to an unknown channel is defined as any set of linear operators $\{T_i\}_i$, with $T_i\in\mathcal{L}(\bigotimes_{j=1}^n \mathcal{H}_{\mathrm{A},j}\otimes\mathcal{H}_{\mathrm{B},j})$, such that $T_i\ge 0$ and $\sum_i T_i$ is a quantum $(n+1)$-comb on $(\mathcal{H}_0,\mathcal{H}_{\mathrm{A},1},\mathcal{H}_{\mathrm{B},1},\ldots,\mathcal{H}_{\mathrm{A},n},\mathcal{H}_{\mathrm{B},n},\mathcal{H}_{n+1})$, where $\mathcal{H}_0\cong\mathcal{H}_{n+1}\cong\mathbb{C}$ are one-dimensional.
\end{definition}
It is known that any sequential tester can be obtained by a sequential query-access algorithm and any sequential query-access algorithm can be described by a sequential tester~\cite{chiribella2009theoretical,bavaresco2022unitary}.
When applying a sequential tester $\{T_i\}_i$ to $n$ queries to a quantum channel $\mathcal{E}$, we get the classical outcome $i$ with probability 
\begin{equation*}
p_i=T_i\star C_\mathcal{E}^{\otimes n}= \tr(T_i (C_{\mathcal{E}}^{\otimes n})^\mathrm{T})=\tr(T_i^\mathrm{T} C_{\mathcal{E}}^{\otimes n}),
\end{equation*}
where $C_\mathcal{E}^{\otimes n}$
denotes the Choi operator of all $n$ queries to the channel $\mathcal{E}$ and $(\,\cdot\,)^{\mathrm{T}}$ represents matrix transposition.

\paragraph{Quantum channel discrimination.}
Let $\mathcal{N}$ be a finite set of quantum channels. Then, the discrimination problem for channels in $\mathcal{N}$ can be defined as follows.
\begin{problem}
Suppose the channel $\mathcal{E}$ is uniformly randomly chosen from the set $\mathcal{N}$. The algorithm (or tester) can make $n$ queries to the channel $\mathcal{E}$ in order to identify $\mathcal{E}$.
\end{problem}
Consider a sequential tester $\{T_\mathcal{E}\}_{\mathcal{E}\in\mathcal{N}}$ for this discrimination task, where $T_\mathcal{E}$ corresponds to outputting the label $\mathcal{E}$. Then, the success probability can be computed as
\[\Pr[\textup{success}]=\frac{1}{|\mathcal{N}|} \sum_{\mathcal{E}\in\mathcal{N}} T_\mathcal{E}\star C_{\mathcal{E}}^{\otimes n},\]
where $C_{\mathcal{E}}$ is the (unnormalized) Choi state of $\mathcal{E}$.
We say an algorithm solves the discrimination problem if the success probability is larger than $2/3$.

\subsection{Schur-Weyl duality on bipartite systems}
Let us consider a sequence of Hilbert spaces $\mathcal{H}_1,\mathcal{H}_2,\ldots,\mathcal{H}_n$ such that $\mathcal{H}_i\cong\mathbb{C}^d$ for $1\leq i\leq n$.
The Hilbert space \(\bigotimes_{i=1}^n\mathcal{H}_i\) admits representations of the symmetric group \(\mathfrak{S}_n\) (i.e.\ the group of all permutations on the set $\{1,2,\ldots,n\}$) and unitary group \(\mathbb{U}_d\) (i.e.\ the group of unitaries on $d$-dimensional Hilbert space). The unitary group acts by simultaneous ``rotation'' as \(U^{\otimes n}\) for any \(U\in \mathbb{U}_d\), while the symmetric group acts by permuting tensor factors:
\begin{equation}
\texttt{p}(\pi)\ket{\psi_1}\cdots\ket{\psi_n}=\ket{\psi_{\pi^{-1}(1)}}\cdots\ket{\psi_{\pi^{-1}(n)}},
\end{equation}
where \(\pi\in\mathfrak{S}_n\). 
The two actions \(U^{\otimes n}\) and \(\texttt{p}(\pi)\) commute with each other, hence $\bigotimes_{i=1}^n\mathcal{H}_i$ admits a representation of group \(\mathbb{U}_d\times \mathfrak{S}_n\). In particular, the Schur-Weyl duality (see, e.g., \cite{fulton2013representation}) states that
\begin{equation}\label{eq-1111224}
\bigotimes_{i=1}^n\mathcal{H}_i\stackrel{\mathfrak{S}_n\times\mathbb{U}_d}{\cong}\bigoplus_{\lambda \vdash_d\, n}\mathcal{P}_\lambda\otimes \mathcal{Q}^d_\lambda,
\end{equation}
where \(\mathcal{P}_\lambda\) and \(\mathcal{Q}^d_\lambda\) are irreducible representations of \(\mathfrak{S}_n\) and \(\mathbb{U}_d\) labeled by Young diagram $\lambda$, respectively. We denote by $\texttt{p}_\lambda(\pi)$ and $\texttt{q}_\lambda(U)$ the actions of $\pi\in\mathfrak{S}_n$ and $U\in\mathbb{U}_d$ on $\mathcal{P}_\lambda$ and $\mathcal{Q}_\lambda^d$, respectively.

Now, let us consider two sequences of Hilbert spaces $(\mathcal{H}_{\textup{A},1},\ldots,\mathcal{H}_{\textup{A},n})$ and $(\mathcal{H}_{\textup{B},1},\ldots,\mathcal{H}_{\textup{B},n})$, where $\mathcal{H}_{\textup{A},i}\cong \mathbb{C}^{d_1}$ and $\mathcal{H}_{\textup{B},j} \cong \mathbb{C}^{d_2}$. We define the action of group $\mathfrak{S}_{n}\times\mathfrak{S}_n$ on $\bigotimes_{i=1}^n \mathcal{H}_{\textup{A},i}\otimes \mathcal{H}_{\textup{B},i}$ as $\texttt{p}_\mathrm{A}(\pi_1)\otimes \texttt{p}_{\mathrm{B}}(\pi_2)$ for $(\pi_1,\pi_2)\in\mathfrak{S}_n\times\mathfrak{S}_n$, where $\texttt{p}_\mathrm{A}(\cdot)$ denotes the permutation action on $\bigotimes_{i=1}^n \mathcal{H}_{\mathrm{A},i}$, and similarly for $\texttt{p}_\mathrm{B}(\,\cdot\,)$.
We define the action of \(\mathbb{U}_{d_1}\times \mathbb{U}_{d_2}\) on \(\bigotimes_{i=1}^n \mathcal{H}_{\textup{A},i}\otimes \mathcal{H}_{\textup{B},i}\) as $(U_\mathrm{A}\otimes U_{\textup{B}})^{\otimes n}$ for $(U_{\textup{A}},U_{\textup{B}})\in \mathbb{U}_{d_1}\times\mathbb{U}_{d_2}$. 
Since the tensor product of irreducible representations of two groups $G_1,G_2$ is still irreducible with respect to $G_1\times G_2$, we can see the Schur-Weyl duality on the bipartite system:
\[\bigotimes_{i=1}^n \mathcal{H}_{\textup{A},i}\otimes \mathcal{H}_{\textup{B},i}\,\, \stackrel{\mathfrak{S}_{n}\times\mathfrak{S}_n\times\mathbb{U}_{d_1}\times\mathbb{U}_{d_2}}{\cong}\,\,\bigoplus_{\substack{\lambda\vdash_{d_1} \, n\\ \mu\vdash_{d_2} \, n}} \mathcal{P}_{\lambda}\otimes\mathcal{P}_\mu\otimes \mathcal{Q}^{d_1}_\lambda\otimes\mathcal{Q}^{d_2}_\mu,\]
where $\mathcal{P}_\lambda\otimes\mathcal{P}_\mu\otimes\mathcal{Q}^{d_1}_\lambda\otimes\mathcal{Q}^{d_2}_\mu$ is an irreducible representation of $\mathfrak{S}_n\times\mathfrak{S}_n\times \mathbb{U}_{d_1}\times\mathbb{U}_{d_2}$.

\section{Local test of quantum channels}\label{sec-770126}
Here, we first prove the local test of quantum channels and use it to derive our upper bounds for quantum channel tomography.

\begin{theorem}\label{thm-1221159}
Suppose $d_1,d_2$ and $r$ are positive integers such that $rd_2 \geq d_1$.
If there exists a parallel tester $\{T_i\}_i$ that makes $n$ queries to an unknown isometry $\mathcal{V}\in\isochannel_{d_1,rd_2}$ and produces a classical outcome $i$ with probability $P_i(\mathcal{V})$,
then there exists another parallel tester $\{\widetilde{T}_i\}_i\cup\{\widetilde{T}_\bot\}$ (where $\bot$ is an extra irrelevant label outside the range of $i$) that makes $n$ queries to an unknown channel $\mathcal{E}\in\qchannel_{d_1,d_2}^r$, and produces a classical outcome $i$ with probability $\EE{\mathcal{V}\sim \dilation_r(\mathcal{E})}[P_i(\mathcal{V})]$.
\end{theorem}
The proof of \cref{thm-1221159} is deferred to \cref{sec-4170414}. 
Then, for a more convenient use of \cref{thm-1221159}, we need the following definition.
\begin{definition}[Estimation tasks of quantum channels]\label{not:est_task}
An estimation task of quantum channels in $\qchannel_{d_1,d_2}^r$ is described by a set $\{A_{\mathcal{E}}\}_{\mathcal{E}\in\qchannel_{d_1,d_2}^r}$, where $A_{\mathcal{E}}$ denotes the set of classical outcomes considered to be the ``correct answer'' of the task when the unknown channel is $\mathcal{E}$.
\end{definition}

We can easily obtain the following result, by using \cref{thm-1221159}.
\begin{theorem}\label{coro-12122150}
Suppose $d_1,d_2$ and $r$ are positive integers such that $rd_2 \geq d_1$, and $\{A_{\mathcal{E}}\}_{\mathcal{E}\in\qchannel_{d_1,d_2}^r}$ is an estimation task of quantum channels.
If there exists a parallel tester that makes $n$ queries to an arbitrary dilation $\mathcal{V}\in \dilation_r(\mathcal{E})$ of an unknown channel $\mathcal{E}\in\qchannel_{d_1,d_2}^r$ and produces a classical outcome $i\in A_{\mathcal{E}}$ with probability at least $1-\delta$, then there exists another parallel tester that makes $n$ queries to $\mathcal{E}$ itself and produces a classical outcome $i\in A_{\mathcal{E}}$ with probability at least $1-\delta$.
\end{theorem}
\begin{proof}
Suppose $P_i(\mathcal{V})$ is the probability of the parallel tester outputting $i$ when making queries to $\mathcal{V}$.
Therefore, we can see that $\sum_{i\in A_{\mathcal{E}}}P_i(\mathcal{V})\geq 1-\delta$ for any $\mathcal{V}\in\dilation_r(\mathcal{E})$.
Due to \cref{thm-1221159}, there exists a parallel tester that makes $n$ queries to $\mathcal{E}$ and outputs $i$ with probability
\[\widetilde{P}_i(\mathcal{E})=\EE{\mathcal{V}\sim \dilation_r(\mathcal{E})}[P_i(\mathcal{V})].\]
Therefore, we can see that the probability of this tester outputting $i\in A_{\mathcal{E}}$ is lower bounded as
\[\sum_{i\in A_{\mathcal{E}}}\widetilde{P}_i(\mathcal{E})=\EE{\mathcal{V}\sim \dilation_r(\mathcal{E})}\left[\sum_{i\in A_\mathcal{E}} P_i(\mathcal{V})\right]\geq 1-\delta.\]
\end{proof}

\subsection{Quantum channel tomography and estimation}
We can use \cref{coro-12122150} and the isometry tomography algorithm under diamond norm shown in \cref{sec-3102118} to give the following result.
\begin{theorem}\label{coro-12120134}
There exists a parallel tester that makes $O(rd_1d_2/\varepsilon^2)$ queries to an unknown channel $\mathcal{E}\in\qchannel_{d_1,d_2}^r$ and produces an estimate $\mathcal{F}$ satisfying $\|\mathcal{F}-\mathcal{E}\|_{\diamond}\leq \varepsilon$ with probability at least $2/3$. Here, $\|\cdot\|_{\diamond}$ denotes the diamond norm.
\end{theorem}

\begin{proof}

We consider the estimation task $\{A_{\mathcal{E}}\}_{\mathcal{E}\in\qchannel_{d_1,d_2}^r}$ where
\[A_{\mathcal{E}}=\left\{\mathcal{F}\in\qchannel_{d_1,d_2}^r \,\,\big|\,\, \|\mathcal{F}-\mathcal{E}\|_\diamond \leq \varepsilon\right\}.\]
Note that, for $\mathcal{E}\in\qchannel_{d_1,d_2}^r$, any isometry in $\dilation_r(\mathcal{E})$ is in $\isochannel_{d_1,rd_2}$.
Then, due to \cref{thm-1240018}, we have a parallel tester that makes $n=O(rd_1d_2/\varepsilon^2)$ queries to a dilation $\mathcal{V}\in \dilation_r(\mathcal{E})$ and produces $\mathcal{W}$ satisfying 
\[\|\contract_r(\mathcal{W})-\mathcal{E}\|_\diamond \leq \|\mathcal{W}-\mathcal{V}\|_\diamond\leq \varepsilon,\]
with probability at least $2/3$, where the first inequality is by the contractivity of the diamond norm.
We set the tester to finally output $\contract_r(\mathcal{W})$ when producing $\mathcal{W}$. Then, it solves the task $\{A_{\mathcal{E}}\}_{\mathcal{E}\in\qchannel_{d_1,d_2}^r}$ by making $n$ queries to a dilation of $\mathcal{E}$. 
Due to \cref{coro-12122150}, there exists another parallel tester that solves this task by making $n$ queries to $\mathcal{E}$ it self.
\end{proof}

We also provide upper bounds with Heisenberg scaling $O(1/\varepsilon)$ for (possibly non-unitary) channel tomography in the boundary regime ($\tau=1$), by using \cref{coro-12122150} and the unitary tomography algorithm shown in \cite{yang2020optimal}. 
\begin{theorem}\label{coro-12120135}
Suppose $\tau=rd_2/d_1=1$. There exists a parallel tester that makes $O(rd_1d_2/\varepsilon)$ queries to an unknown channel $\mathcal{E}\in\qchannel_{d_1,d_2}^r$ and produces an estimate $\mathcal{F}$ satisfying $\|\frac{1}{d}C_\mathcal{F}-\frac{1}{d}C_\mathcal{E}\|_{1}\leq \varepsilon$ with probability at least $2/3$. Here, $C_{\mathcal{E}}$ is the (unnormalized) Choi state of $\mathcal{E}$ and $\|\cdot\|_1$ denotes the trace norm.

Moreover, there exists a parallel tester that makes $O(\min\{rd_1^{1.5}d_2/\varepsilon,rd_1d_2/\varepsilon^2\})$ queries to an unknown channel $\mathcal{E}\in\qchannel_{d_1,d_2}^r$ and produces an estimate $\mathcal{F}$ satisfying $\|\mathcal{F}-\mathcal{E}\|_{\diamond}\leq \varepsilon$ with probability at least $2/3$, where $\|\cdot\|_\diamond$ denotes the diamond norm.
\end{theorem}
\begin{proof}
Consider the task $\{A_{\mathcal{E}}\}_{\mathcal{E}\in\qchannel_{d_1,d_2}^r}$ where
\[A_{\mathcal{E}}=\left\{\mathcal{F}\in\qchannel_{d_1,d_2}^r\,\,\bigg|\,\, \left\|\frac{1}{d}C_{\mathcal{F}}-\frac{1}{d}C_{\mathcal{E}}\right\|_1\leq\varepsilon\right\}.\]
Note that for $\mathcal{E}\in\qchannel_{d_1,d_2}^r$, any isometry in $\dilation_r(\mathcal{E})$ is a $d_1$-dimensional unitary.
Using Yang-Renner-Chiribella algorithm~\cite{yang2020optimal}, we have a parallel tester that makes $n=O(d_1^2/\varepsilon)=O(rd_1d_2/\varepsilon)$ queries to a unitary $\mathcal{U}\in\dilation_r(\mathcal{E})$ and produces $\mathcal{W}$ satisfying 
\[\left\|\frac{1}{d}C_{\contract_r(\mathcal{W})}-\frac{1}{d}C_{\mathcal{E}}\right\|_1\leq \left\|\frac{1}{d}C_\mathcal{W}-\frac{1}{d}C_{\mathcal{U}}\right\|_1=\sqrt{1-\mathrm{F}_{\mathrm{ch}}(\mathcal{W},\mathcal{U})}\leq \varepsilon,\]
with probability at least $2/3$,
where the first inequality is due to the contractivity of trace norm, and the last inequality is because the Yang-Renner-Chiribella algorithm produces an estimate $\mathcal{W}$ satisfying $\mathrm{F}_{\mathrm{ch}}(\mathcal{W},\mathcal{U})\geq 1-\varepsilon^2$, with probability at least $2/3$. Here, $\mathrm{F}_{\mathrm{ch}}$ denotes the channel fidelity (or equivalently, entanglement fidelity).
We set the tester to finally output $\contract_r(\mathcal{W})$ upon producing $\mathcal{W}$. Then it solves the task $\{A_{\mathcal{E}}\}_{\mathcal{E}\in\qchannel_{d_1,d_2}^r}$ by making $n$ queries to a dilation of $\mathcal{E}$. 
Due to \cref{coro-12122150}, there exists a parallel tester that solves this task by making $n$ queries to $\mathcal{E}$ itself.

Similarly, consider the task $\{A_{\mathcal{E}}\}_{\mathcal{E}\in\qchannel_{d_1,d_2}^r}$ where
\[A_{\mathcal{E}}=\left\{\mathcal{F}\in\qchannel_{d_1,d_2}^r\,\,\big|\,\, \left\|\mathcal{F}-\mathcal{E}\right\|_\diamond\leq\varepsilon\right\}.\]
Using to Yang-Renner-Chiribella algorithm~\cite{yang2020optimal}, we have a parallel tester that makes $n=O(rd_1d_2/(\varepsilon/\sqrt{d_1}))=O(rd_1^{1.5}d_2/\varepsilon)$ queries to a unitary $\mathcal{U}\in\dilation_r(\mathcal{E})$ and produces $\mathcal{W}$ satisfying 
\[\|\contract_r(\mathcal{W})-\mathcal{E}\|_\diamond\leq \|\mathcal{W}-\mathcal{U}\|_\diamond\leq \sqrt{2d_1}\sqrt{1-\mathrm{F}_{\mathrm{ch}}(\mathcal{W},\mathcal{U})}\leq \varepsilon,\]
with probability at least $2/3$,
where the second inequality is by \cite[Proposition 1.9]{haah2023query}.
We set the tester to finally output $\contract_r(\mathcal{W})$ upon producing $\mathcal{W}$. 
Then it solves the task $\{A_{\mathcal{E}}\}_{\mathcal{E}\in\qchannel_{d_1,d_2}^r}$ by making $n$ queries to a dilation of $\mathcal{E}$. 
Due to \cref{coro-12122150}, there exists a parallel tester that solves this task by making $n$ queries to $\mathcal{E}$ itself. 
Combining this with \cref{coro-12120134}, we can see that $O(\min\{rd_1^{1.5}d_2/\varepsilon,rd_1d_2/\varepsilon^2\})$ queries suffice.
\end{proof}

We also demonstrate (mixed) state tomography with Heisenberg scaling, when state-preparation channels are available. This is obtained by combining our \cref{coro-12122150} with the pure state tomography algorithm given in \cite{chen2025inverse}.
\begin{theorem}\label{coro-12120136}
Suppose $\tau=rd_2/d_1=1$. There exists a parallel tester that makes $O(\min\{d_1^{1.5}/\varepsilon,d_1/\varepsilon^2\})$ queries to a  channel $\mathcal{E}\in\qchannel_{d_1,d_2}^r$ and produces an estimate $\rho$ satisfying $\|\rho-\mathcal{E}(\ketbra{0}{0})\|_1\leq \varepsilon$ with probability at least $2/3$, where $\|\cdot\|_1$ is the trace norm.
\end{theorem}
\begin{proof}
Consider the task $\{A_{\mathcal{E}}\}_{\mathcal{E}\in\qchannel_{d_1,d_2}^r}$ where
\[A_{\mathcal{E}}=\left\{\rho\in \mathcal{L}(\mathbb{C}^{d_2}) \,\,\big|\,\, \left\|\rho-\mathcal{E}(\ketbra{0}{0})\right\|_1\leq\varepsilon\right\}.\]
Note that, for $\mathcal{E}\in\qchannel_{d_1,d_2}^r$, any isometry in $\dilation_r(\mathcal{E})$ is a $d_1$-dimensional unitary.
Using Chen's algorithm~\cite{chen2025inverse}, we have a parallel tester that makes $n=O(\min\{d_1^{1.5}/\varepsilon,d_1/\varepsilon^2\})$ queries to a unitary $\mathcal{U}\in\dilation_r(\mathcal{E})$ and produces $\ket{\psi}$ satisfying
\[\|\tr_{r}(\ketbra{\psi}{\psi})-\mathcal{E}(\ketbra{0}{0})\|_1\leq \|\ketbra{\psi}{\psi}-U\ketbra{0}{0}U^\dag\|_1\leq \varepsilon,\]
with probability at least $2/3$,
where the second inequality is because Chen's algorithm produces an estimate $\ket{\psi}$ for $U\ket{0}$\footnote{In \cite{chen2025inverse}, the author considered estimating $U\ket{d}$ for notation convenience, here we consider estimating $U\ket{0}$.} to within trace norm error $\varepsilon$ with probability at least $2/3$.
We set the tester to finally output $\tr_r(\ketbra{\psi}{\psi})$ upon producing $\ket{\psi}$.
Then, it solves the task $\{A_{\mathcal{E}}\}_{\mathcal{E}\in\qchannel_{d_1,d_2}^r}$ by making $n$ queries to a dilation of $\mathcal{E}$. 
Due to \cref{coro-12122150}, there exists a parallel tester that solves this task by making $n$ queries to $\mathcal{E}$ itself.
\end{proof}

Note that \cref{coro-12120135} and \cref{coro-12120136} provide Heisenberg-scaling upper bounds in the boundary regime $\tau=rd_2/d_1=1$. 
We also provide upper bounds with mixing Heisenberg- and classical- scalings in the non-boundary regime $\tau>1$. 
This is obtained by combining our \cref{coro-12122150} with the isometry tomography algorithm provided in \cite{yoshida2025quantum}.
\begin{theorem}\label{coro-2270246}
Suppose $\tau=rd_2/d_1>1$. There exists a parallel tester that makes $O((rd_2-d_1)d_1/\varepsilon^2+d_1^2/\varepsilon)$ queries to an unknown channel $\mathcal{E}\in\qchannel_{d_1,d_2}^r$ and produces an estimate $\mathcal{F}$ satisfying $\|\frac{1}{d}C_{\mathcal{F}}-\frac{1}{d}C_{\mathcal{E}}\|_1\leq \varepsilon$ with probability at least $2/3$. Here, $C_\mathcal{E}$ is the unnormalized Choi state of $\mathcal{E}$ and $\|\cdot\|_1$ denotes the trace norm.
\end{theorem}
\begin{proof}
Consider the estimation task $\{A_{\mathcal{E}}\}_{\mathcal{E}\in\qchannel_{d_1,d_2}^r}$ where
\[A_{\mathcal{E}}=\left\{\mathcal{F}\in\qchannel_{d_1,d_2}^r \,\,\big|\,\, \left\|\frac{1}{d}C_\mathcal{F}-\frac{1}{d}C_\mathcal{E}\right\|_1 \leq \varepsilon\right\}.\]
Note that, for $\mathcal{E}\in\qchannel_{d_1,d_2}^r$, any isometry in $\dilation_r(\mathcal{E})$ is in $\isochannel_{d_1,rd_2}$.
Using Yoshida-Miyazaki-Murao algorithm~\cite{yoshida2025quantum} (see also \cref{lemma-2232322}), we have a parallel tester that makes $n=O((rd_2-d_1)d_1/\varepsilon^2+d_1^2/\varepsilon)$ queries to an isometry $\mathcal{V}\in \dilation_r(\mathcal{E})$ and produces an estimate $\mathcal{W}$ satisfying 
\[\left\|\frac{1}{d}C_{\contract_r(\mathcal{W})}-\frac{1}{d}C_\mathcal{E}\right\|_1 \leq \left\|\frac{1}{d}C_\mathcal{W}-\frac{1}{d}C_\mathcal{V}\right\|_1\leq \varepsilon,\]
with probability at least $2/3$,
where the first inequality is by the contractivity of the trace norm.
We set the tester to finally output $\contract_r(\mathcal{W})$ upon producing $\mathcal{W}$. Then it solves the task $\{A_{\mathcal{E}}\}_{\mathcal{E}\in\qchannel_{d_1,d_2}^r}$ by making $n$ queries to a dilation of $\mathcal{E}$. 
Due to \cref{coro-12122150}, there exists a parallel tester that solves this task by making $n$ queries to $\mathcal{E}$ itself.
\end{proof}

\subsection{Construction of the local testers}\label{sec-4170414}
Now, we prove \cref{thm-1221159}. Our proof draws on an idea similar to the construction of local testers for quantum states in \cite{chen2024local}, and extends it to local testers for quantum channels.

We introduce some notation that will be used in the proof later.
\begin{notation}
For each $i\in [n]$, we define $\mathcal{H}_{\mathrm{A},i}\cong\mathbb{C}^{d_1}$ and $\mathcal{H}_{\mathrm{B},i}\otimes\mathcal{H}_{\mathrm{anc},i}\cong\mathbb{C}^{d_2}\otimes\mathbb{C}^{r}$ to be the input and output systems of the $i$-th query to an isometry channel $\mathcal{V}\in\isochannel_{d_1,rd_2}$.
The following decompositions are due to Schur-Weyl duality:
\[\bigotimes_{i=1}^n \mathcal{H}_{\mathrm{A},i}\otimes\mathcal{H}_{\mathrm{B},i}\stackrel{\mathfrak{S}_n \times \mathbb{U}_{d_1d_2}}{\cong} \bigoplus_{\lambda\vdash_{d_1d_2} n}\mathcal{P}_{\lambda}\otimes\mathcal{Q}^{d_1d_2}_{\lambda}\qquad\textup{and}\qquad \bigotimes_{i=1}^n\mathcal{H}_{\mathrm{anc},i} \stackrel{\mathfrak{S}_n\times \mathbb{U}_{r}}{\cong}\bigoplus_{\lambda\vdash_{r} n}\mathcal{P}_\lambda\otimes\mathcal{Q}^r_\lambda.\]
Hence, we can see that
\[\bigotimes_{i=1}^n \mathcal{H}_{\mathrm{A},i}\otimes\mathcal{H}_{\mathrm{B},i}\otimes\mathcal{H}_{\mathrm{anc},i}\stackrel{\mathfrak{S}_n\times\mathfrak{S}_n\times\mathbb{U}_{d_1d_2}\times\mathbb{U}_r}{\cong}\bigoplus_{\substack{\lambda\vdash_{d_1d_2}n\\ \mu\vdash_r n}}\mathcal{P}_{\mathrm{AB},\lambda}\otimes\mathcal{P}_{\mathrm{anc},\mu}\otimes\mathcal{Q}^{d_1d_2}_{\mathrm{AB},\lambda}\otimes\mathcal{Q}^r_{\mathrm{anc},\mu},\]
where we use $\mathcal{P}_{\mathrm{AB},\lambda}\otimes \mathcal{Q}^{d_1d_2}_{\mathrm{AB},\lambda}$ to denote the subspace $\mathcal{P}_{\lambda}\otimes \mathcal{Q}^{d_1d_2}_{\lambda}$ in $\bigotimes_{i=1}^n\mathcal{H}_{\mathrm{A},i}\otimes\mathcal{H}_{\mathrm{B},i}$, and use $\mathcal{P}_{\mathrm{anc},\mu}\otimes \mathcal{Q}^r_{\mathrm{anc},\mu}$ to denote the subspace $\mathcal{P}_\mu\otimes\mathcal{Q}^r_\mu$ in $\bigotimes_{i=1}^n \mathcal{H}_{\mathrm{anc},i}$.
\end{notation}

We give the proof of \cref{thm-1221159} as follows.

\begin{proof}[Proof of \cref{thm-1221159}]
Define $s\coloneqq \min\{d_1d_2,r\}$. Note that here, we do not assume $r\leq d_1d_2$.
The construction of the parallel tester $\{\widetilde{T}_i\}_i\cup\{\widetilde{T}_\bot\}$ is as follows.
\begin{itemize}
    \item First, we define a new tester $\{\overline{T}_i\}_i$ such that
    \begin{equation}\label{eq-1242355}
    \overline{T}_i\coloneqq \EE{U\sim\mathbb{U}_{r}}[U^{\otimes n} T_iU^{\dag\otimes n}],
    \end{equation}
    where $U^{\otimes n}$ acts on $\bigotimes_{j=1}^n \mathcal{H}_{\mathrm{anc},j}$.
    
    \item We then define the tester $\{\widetilde{T}_i\}_i\cup\{\widetilde{T}_\bot\}$ as follows.
    For each $i$,
    \begin{equation}\label{eq-1242217}
    \widetilde{T}_i\coloneqq \bigoplus_{\lambda\vdash_{s} n} \frac{1}{\dim(\mathcal{P}_\lambda)\dim(\mathcal{Q}^r_\lambda)} \cdot I_{\mathcal{P}_{\mathrm{AB},\lambda}}\otimes \tr_{\mathcal{Q}^r_{\mathrm{anc},\lambda}}\Big(\bbra{I_{\mathcal{P}_\lambda}}\overline{T}_i\kett{I_{\mathcal{P}_\lambda}}\Big),
    \end{equation}
    where $\kett{I_{\mathcal{P}_\lambda}}\in\mathcal{P}_{\mathrm{AB},\lambda}\otimes\mathcal{P}_{\mathrm{anc},\lambda}$ is the (unnormalized) maximally entangled state defined w.r.t. the Young's orthogonal basis (or Young-Yamanouchi basis) on which $\pi\in\mathfrak{S}_n$ acts as a real matrix~\cite{ceccherini2010representation}.   
    Note that $\widetilde{T}_i$ is a linear operator on $\bigotimes_{j=1}^n \mathcal{H}_{\mathrm{A},j}\otimes \mathcal{H}_{\mathrm{B},j}$.
    As we will not explicitly use $\widetilde{T}_\bot$, its definition is deferred to \cref{eq-4172122} for clarity.   
\end{itemize}
In order to verify our construction, we prove in \cref{lemma-1242352} that $\{\overline{T}_i\}_i$ is a parallel tester which makes $n$ queries to an isometry $\mathcal{V}\in\isochannel_{d_1,rd_2}$ and produces a classical outcome $i$ with probability 
\begin{equation*}
    \EE{\mathcal{W}\sim \dilation_r(\contract_r(\mathcal{V}))} [T_i\star C_\mathcal{W}^{\otimes n}],
\end{equation*}
and we also give an explicit expression of this probability. 
Then, by using \cref{lemma-1242352}, we prove in \cref{lemma-1242353} that $\{\widetilde{T}_i\}_i\cup\{\widetilde{T}_\bot\}$ is indeed a parallel tester that makes $n$ queries to a quantum channel $\mathcal{E}\in\qchannel_{d_1,d_2}^r$ and produces outcome $i$ with probability $\EE{\mathcal{W}\sim \dilation_r(\mathcal{E})} [T_i\star C_\mathcal{W}^{\otimes n}]$, as desired. 
\end{proof}

In the following lemma, we prove some properties of the tester $\{\overline{T_i}\}_i$.
\begin{lemma}\label{lemma-1242352}
The tester $\{\overline{T}_i\}_i$ defined in \cref{eq-1242355} satisfies:
\begin{enumerate}
\item $\{\overline{T}_i\}_i$ is a parallel tester. Moreover, for $\mathcal{V}\in\isochannel_{d_1,rd_2}$, the tester produces $i$ with probability
\begin{equation}\label{eq-4180336}
\overline{T}_i\star C_{\mathcal{V}}^{\otimes n}=\EE{\mathcal{W}\sim\dilation_r(\contract_r(\mathcal{V}))}[T_i\star C_{\mathcal{W}}^{\otimes n}].
\end{equation}

\item The probability in \cref{eq-4180336} can also be expressed as
\[\overline{T}_i\star C_{\mathcal{V}}^{\otimes n}=\sum_{\lambda\vdash_s n}\frac{1}{\dim(\mathcal{Q}^r_\lambda)}\tr\left(\tr_{\mathcal{Q}^r_{\mathrm{anc},\lambda}}\Big(\bbra{I_{\mathcal{P}_\lambda}}\overline{T}_i^{\mathrm{T}}\kett{I_{\mathcal{P}_\lambda}}\Big)\cdot \tr_{\mathcal{Q}^r_{\mathrm{anc},\lambda}}\Big(\ketbra{V_\lambda}{V_\lambda}\Big)\right),\]
where $s=\min\{d_1d_2,r\}$, and $\ket{V_\lambda}\in\mathcal{Q}^{d_1d_2}_{\mathrm{AB},\lambda}\otimes\mathcal{Q}^r_{\mathrm{anc},\lambda}$ is a vector appearing in the decomposition $\kett{V}^{\otimes n}=\bigoplus_{\lambda\vdash_{s}n}\kett{I_{\mathcal{P}_\lambda}}\otimes \ket{V_\lambda}$ according to \cref{lemma-3262205}.
\end{enumerate}
\end{lemma}
\begin{proof}
\textbf{Item 1}. We have
\begin{equation}\label{eq-1260054}
\sum_{i}\overline{T}_i=\EE{U\sim\mathbb{U}_r}\left[U^{\otimes n}\sum_i T_iU^{\dag\otimes n}\right]=\EE{U\sim\mathbb{U}_r}\left[U^{\otimes n}(\rho_\mathrm{A}\otimes I_{\mathrm{B},\mathrm{anc}}) U^{\dag\otimes n}\right]=\rho_\mathrm{A}\otimes I_{\mathrm{B},\mathrm{anc}},
\end{equation}
where we used the fact that $\{T_i\}_i$ is a parallel tester so that $\sum_i T_i=\rho_\mathrm{A}\otimes I_{\mathrm{B},\mathrm{anc}}$ for $\rho_\mathrm{A}$ a density operator on $\bigotimes_{j=1}^n\mathcal{H}_{\mathrm{A},j}$ and $I_{\mathrm{B},\mathrm{anc}}$ the identity operator on $\bigotimes_{j=1}^n\mathcal{H}_{\mathrm{B},j}\otimes\mathcal{H}_{\mathrm{anc},j}$. We further note that $U^{\otimes n}$ acts only non-trivially on $\bigotimes_{j=1}^n \mathcal{H}_{\mathrm{anc},j}$.
Therefore, $\{\overline{T}_i\}_i$ is indeed a parallel tester.

On the other hand, we have
\begin{align}
\overline{T}_i\star C_{\mathcal{V}}^{\otimes n}&=\tr\!\left(\overline{T}_i^{\mathrm{T}} C_{\mathcal{V}}^{\otimes n}\right) \nonumber \\
&= \EE{U\sim \mathbb{U}_r} \left[\tr\!\left(T_i^{\mathrm{T}}U^{\otimes n} C_\mathcal{V}^{\otimes n} U^{\dag\otimes n}\right)\right]\label{eq:use-def-Tbar} \\
&=\EE{U\sim \mathbb{U}_r} \left[\tr\!\left(T_i^{\mathrm{T}}  C_{\mathcal{U}\circ \mathcal{V}}^{\otimes n}\right)\right]\nonumber\\
&=\EE{\mathcal{W}\sim\dilation_r(\contract_r(\mathcal{V}))}\left[\tr\!\left(T_i^\mathrm{T} C_{\mathcal{W}}^{\otimes n}\right)\right]\label{eq-1252216}\\
&=\EE{\mathcal{W}\sim\dilation_r(\contract_r(\mathcal{V}))}\left[T_i \star C_{\mathcal{W}}^{\otimes n}\right] \nonumber
\end{align}
where \Cref{eq:use-def-Tbar} is due to the definition of $\overline{T}_i$ (see \Cref{eq-1242355}),
and \cref{eq-1252216} uses the definition of the Haar distribution on $\dilation_r(\cdot)$ (see \cref{not:Haar}). 

\textbf{Item 2}. We treat $\kett{V}$ as a bipartite vector in $(\mathbb{C}^{d_1}\otimes \mathbb{C}^{d_2})\otimes \mathbb{C}^{r}$. Using \cref{lemma-3262205}, we can write $\kett{V}^{\otimes n}=\bigoplus_{\lambda\vdash_s n}\kett{I_{\mathcal{P}_\lambda}}\otimes\ket{V_\lambda}$, in which $\ket{V_\lambda}\in \mathcal{Q}^{d_1d_2}_{\mathrm{AB},\lambda}\otimes\mathcal{Q}^r_{\mathrm{anc},\lambda}$ and $\kett{I_{\mathcal{P}_\lambda}}\in\mathcal{P}_{\mathrm{AB},\lambda}\otimes\mathcal{P}_{\mathrm{anc},\lambda}$ is an (unnormalized) maximally entangled state. 
Also, we note that $U^{\otimes n}$ commutes with $\overline{T}_i$ for any $U\in\mathbb{U}_r$, in which $U^{\otimes n}$ acts non-trivially on $\bigotimes_{j=1}^n \mathcal{H}_{\mathrm{anc},j}$.
Thus,
\begin{align}
\overline{T}_i\star C_{\mathcal{V}}^{\otimes n}&=\tr\!\left(\overline{T}_i^\mathrm{T} \kettbbra{V}{V}^{\otimes n}\right)\nonumber \\
&=\tr\!\left(\overline{T}_i^\mathrm{T} \EE{U\sim\mathbb{U}_r}[U^{\otimes n}\kettbbra{V}{V}^{\otimes n} U^{\dag\otimes n}]\right) \nonumber  \\
&=\tr\!\left(\overline{T}_i^\mathrm{T} \EE{U\sim\mathbb{U}_r}\left[\bigoplus_{\lambda,\mu\vdash_s n} \kett{I_{\mathcal{P}_\lambda}}\bbra{I_{\mathcal{P}_\mu}}\otimes \texttt{q}_\lambda(U)\ketbra{V_\lambda}{V_\mu}\texttt{q}_\mu(U)^\dag \right]\right)\label{eq-1260033}\\
&=\tr\!\left(\overline{T}_i^{\mathrm{T}} \cdot \left(\bigoplus_{\lambda\vdash_s n}\kett{I_{\mathcal{P}_\lambda}}\bbra{I_{\mathcal{P}_\lambda}}\otimes \tr_{\mathcal{Q}^r_{\mathrm{anc},\lambda}}\Big(\ketbra{V_\lambda}{V_\lambda}\Big)\otimes \frac{1}{\dim(\mathcal{Q}^r_\lambda)}I_{\mathcal{Q}^r_{\mathrm{anc},\lambda}}\right)\right)\label{eq-1260034}\\
&=\sum_{\lambda\vdash_s n} \frac{1}{\dim(\mathcal{Q}^r_\lambda)}\tr\left(\tr_{\mathcal{Q}^r_{\mathrm{anc},\lambda}}\Big(\bbra{I_{\mathcal{P}_\lambda}}\overline{T}_i^{\mathrm{T}} \kett{I_{\mathcal{P}_\lambda}}\Big)\cdot \tr_{\mathcal{Q}^r_{\mathrm{anc},\lambda}}\Big(\ketbra{V_\lambda}{V_\lambda}\Big)\right),\nonumber
\end{align}
where \cref{eq-1260034} uses Schur's lemma~\cite{fulton2013representation}, and in \cref{eq-1260033} $\texttt{q}_\lambda(U)$ acts on $\mathcal{Q}^r_{\mathrm{anc},\lambda}$.
\end{proof}

Next, we prove some properties of $\{\widetilde{T}_i\}_i$ in the following lemma.

\begin{lemma}\label{lemma-1242353}
The operators $\{\widetilde{T}_i\}_i$ defined in \cref{eq-1242217} satisfies:
\begin{enumerate}
\item There exists a positive semidefinite operator $\widetilde{T}_\bot$ such that $\{\widetilde{T}_i\}_i\cup\{\widetilde{T}_\bot\}$ is a parallel tester. In fact, $\widetilde{T}_\bot$ can be explicitly constructed, as shown in \cref{eq-4172122}.
\item For any quantum channel $\mathcal{E}\in\qchannel_{d_1,d_2}^r$, the tester produces outcome $i$ with probability
\[\widetilde{T}_i\star C_{\mathcal{E}}^{\otimes n}=\EE{\mathcal{W}\sim\dilation_r(\mathcal{E})}[T_i\star C_{\mathcal{W}}^{\otimes n}].\]
\end{enumerate}
\end{lemma}
\begin{proof}
\textbf{Item 1}.
According to the definition in \Cref{eq-1242217}, we have
\begin{align}
\sum_i \widetilde{T}_i&=\bigoplus_{\lambda\vdash_s n}\frac{1}{\dim(\mathcal{P}_\lambda)\dim(\mathcal{Q}^r_\lambda)}\cdot I_{\mathcal{P}_{\mathrm{AB},\lambda}}\otimes \tr_{\mathcal{Q}^r_{\mathrm{anc},\lambda}}\left(\bbra{I_{\mathcal{P}_\lambda}}\sum_i \overline{T}_i\kett{I_{\mathcal{P}_\lambda}}\right) \nonumber\\
&=\bigoplus_{\lambda\vdash_s n}\frac{1}{\dim(\mathcal{P}_\lambda)\dim(\mathcal{Q}^r_\lambda)}\cdot I_{\mathcal{P}_{\mathrm{AB},\lambda}}\otimes \tr_{\mathcal{Q}^r_{\mathrm{anc},\lambda}}\Big(\bbra{I_{\mathcal{P}_\lambda}} \rho_\mathrm{A}\otimes I_{\mathrm{B},\mathrm{anc}} \kett{I_{\mathcal{P}_\lambda}}\Big),\label{eq-1260117}
\end{align}
where in \cref{eq-1260117} we use \cref{eq-1260054}.
Then, we write $\rho_\mathrm{A}\otimes I_{\mathrm{B},\mathrm{anc}}=(\rho_\mathrm{A}\otimes I_{\mathrm{B}})\otimes I_{\mathrm{anc}}$ and write $\rho_\mathrm{A}\otimes I_{\mathrm{B}}$ according to the subspace decomposition in the Schur-Weyl duality:
\[\rho_\mathrm{A}\otimes I_{\mathrm{B}}=\bigoplus_{\substack{\lambda,\mu\vdash_{d_1d_2}n}} M_{\lambda\rightarrow \mu},\]
where $M_{\lambda\rightarrow\mu}$ is a linear operator from the subspace $\mathcal{P}_{\mathrm{AB},\lambda}\otimes\mathcal{Q}^{d_1d_2}_{\mathrm{AB},\lambda}$ to the subspace $\mathcal{P}_{\mathrm{AB},\mu}\otimes\mathcal{Q}^{d_1d_2}_{\mathrm{AB},\mu}$, and $M_{\lambda\rightarrow\lambda}$ is positive semidefinite for each $\lambda$ since $\rho_\mathrm{A}\otimes I_{\mathrm{B}}$ is positive semidefinite. 
Moreoever, we can see that
\[(\rho_\mathrm{A}\otimes I_{\mathrm{B}})\otimes I_{\mathrm{anc}}=\bigoplus_{\substack{\lambda,\mu\vdash_{d_1d_2}n \\ \nu\vdash_r n}}M_{\lambda\rightarrow\mu}\otimes I_{\mathcal{P}_{\mathrm{anc},\nu}}\otimes I_{\mathcal{Q}^r_{\mathrm{anc},\nu}}.\]
This means, \cref{eq-1260117} is equal to
\begin{align}
&\bigoplus_{\lambda\vdash_s n}\frac{1}{\dim(\mathcal{P}_\lambda)\dim(\mathcal{Q}^r_\lambda)}\cdot I_{\mathcal{P}_{\mathrm{AB},\lambda}}\otimes \tr_{\mathcal{Q}^r_{\mathrm{anc},\lambda}}\left(\bbra{I_{\mathcal{P}_\lambda}}\bigoplus_{\substack{\kappa,\mu\vdash_{d_1d_2}n \\ \nu\vdash_r n}}M_{\kappa\rightarrow\mu}\otimes I_{\mathcal{P}_{\mathrm{anc},\nu}}\otimes I_{\mathcal{Q}^r_{\mathrm{anc},\nu}}\kett{I_{\mathcal{P}_\lambda}}\right) \nonumber \\
=&\bigoplus_{\lambda\vdash_s n}\frac{1}{\dim(\mathcal{P}_\lambda)\dim(\mathcal{Q}^r_\lambda)}\cdot I_{\mathcal{P}_{\mathrm{AB},\lambda}}\otimes \tr_{\mathcal{Q}^r_{\mathrm{anc},\lambda}}\Big(\tr_{\mathcal{P}_{\mathrm{AB},\lambda}}(M_{\lambda\rightarrow\lambda})\otimes I_{\mathcal{Q}^r_{\mathrm{anc},\lambda}}\Big) \label{eq-1261541}\\
=& \bigoplus_{\lambda\vdash_s n}\frac{1}{\dim(\mathcal{P}_\lambda)}\cdot I_{\mathcal{P}_{\mathrm{AB},\lambda}}\otimes \tr_{\mathcal{P}_{\mathrm{AB},\lambda}}(M_{\lambda\rightarrow\lambda})\nonumber \\
\leq & \bigoplus_{\lambda\vdash_{d_1d_2} n}\frac{1}{\dim(\mathcal{P}_\lambda)}\cdot I_{\mathcal{P}_{\mathrm{AB},\lambda}}\otimes \tr_{\mathcal{P}_{\mathrm{AB},\lambda}}(M_{\lambda\rightarrow\lambda}),\label{eq-1261542}
\end{align}
where \cref{eq-1261541} uses that $\kett{I_{\mathcal{P}_\lambda}}$ is an (unnormalized) maximally entangled state on $\mathcal{P}_{\mathrm{AB},\lambda}\otimes\mathcal{P}_{\mathrm{anc},\lambda}$, and \cref{eq-1261542} uses that $s\leq d_1d_2$ and $M_{\lambda\rightarrow\lambda}\geq 0$ for any $\lambda$.
Next, we have
\begin{align}
\frac{1}{n!}\sum_{\pi\in\mathfrak{S}_n}\texttt{p}_{\mathrm{A}}(\pi)\,\rho_\mathrm{A}\, \texttt{p}_{\mathrm{A}}(\pi)^\dag \otimes I_{\mathrm{B}}&=\frac{1}{n!}\sum_{\pi\in \mathfrak{S}_n}\texttt{p}_{\mathrm{AB}}(\pi)(\rho_\mathrm{A}\otimes I_{\mathrm{B}})\texttt{p}_{\mathrm{AB}}(\pi)^\dag \nonumber \\
&=\bigoplus_{\lambda\vdash_{d_1d_2} n}\frac{1}{\dim(\mathcal{P}_\lambda)} \cdot I_{\mathcal{P}_{\mathrm{AB},\lambda}}\otimes \tr_{\mathcal{P}_{\mathrm{AB},\lambda}}(M_{\lambda\rightarrow\lambda}), \label{eq-1261722}
\end{align}
where $\texttt{p}_\mathrm{A}(\pi)$ and $\texttt{p}_{\mathrm{AB}}(\pi)$ denote the actions of $\pi$ (i.e., permuting tensor factors) on $\bigotimes_{j=1}^n \mathcal{H}_{\mathrm{A},j}$ and $\bigotimes_{j=1}^n \mathcal{H}_{\mathrm{A},j}\otimes\mathcal{H}_{\mathrm{B},j}$, respectively, and \cref{eq-1261722} uses Schur's lemma.
We can see that \cref{eq-1261722} equals exactly \cref{eq-1261542}. This means
\[\sum_{i} \widetilde{T}_i\leq \rho^{\mathrm{sym}}_\mathrm{A} \otimes I_{\mathrm{B}},\]
where $\rho^{\mathrm{sym}}_\mathrm{A}\coloneqq \frac{1}{n!}\sum_{\pi\in\mathfrak{S}_n} \texttt{p}_\mathrm{A}(\pi) \,\rho_\mathrm{A}\, \texttt{p}_\mathrm{A}(\pi)^\dag$ is a mixed quantum state. 
Thus there exists a positive semidefinite linear operator $\widetilde{T}_\bot$ satisfying $\sum_i\widetilde{T}_i+\widetilde{T}_\bot=\rho^{\mathrm{sym}}_\mathrm{A}\otimes I_{\mathrm{B}}$.
In fact, we can explicitly define 
\begin{equation}\label{eq-4172122}
\widetilde{T}_\bot\coloneqq \sum_{\substack{\lambda\vdash_{d_1d_2} n\\ l(\lambda)> s}} P_\lambda(\rho^{\mathrm{sym}}_\mathrm{A} \otimes I_\mathrm{B}) P_\lambda=\frac{1}{r^n}\sum_{i} \sum_{\substack{\lambda\vdash_{d_1d_2} n\\ l(\lambda)> s}} P_\lambda \, \tr_{\mathrm{anc}}(T_i)^{\mathrm{sym}} \, P_\lambda,
\end{equation}
where $l(\lambda)$ denotes the number of rows of $\lambda$, $P_\lambda$ denotes the orghotonal projector onto the subspace $\mathcal{P}_{\mathrm{AB},\lambda}\otimes \mathcal{Q}_{\mathrm{AB},\lambda}^{d_1d_2}$, and $\tr_\mathrm{anc}(T_i)^{\mathrm{sym}}= \frac{1}{n!}\sum_{\pi\in\mathfrak{S}_n} \texttt{p}_\mathrm{AB}(\pi) \,\tr_\mathrm{anc}(T_i)\,\, \texttt{p}_\mathrm{AB}(\pi)^\dag$.
Therefore, $\{\widetilde{T}_i\}_i\cup\{\widetilde{T}_\bot\}$ is a parallel tester.

\textbf{Item 2}. Suppose $\mathcal{E}\in\qchannel_{d_1,d_2}^r$ is a quantum channel. Let $\mathcal{V}\in\dilation_r(\mathcal{E})$ be an arbitrary dilation, where $V: \mathcal{H}_{\mathrm{A}} \rightarrow \mathcal{H}_{\mathrm{B}}\otimes\mathcal{H}_\mathrm{anc}$, $\mathcal{H}_\mathrm{A}\cong \mathbb{C}^{d_1}$, $\mathcal{H}_{\mathrm{B}}\cong\mathbb{C}^{d_2}$, and $\mathcal{H}_{\mathrm{anc}}\cong \mathbb{C}^{r}$. 
We can see that 
\[\tr_{\mathcal{H}_\mathrm{anc}}(C_{\mathcal{V}})=\tr_{\mathcal{H}_\mathrm{anc}}(\kettbbra{V}{V})=C_\mathcal{E}.\]
Treating $\kett{V}$ as a bipartite vector in $(\mathcal{H}_\mathrm{A}\otimes\mathcal{H}_\mathrm{B})\otimes\mathcal{H}_{\mathrm{anc}}$ and according to \cref{lemma-3262205}, we have $\kett{V}^{\otimes n}=\bigoplus_{\lambda\vdash_s n}\kett{I_{\mathcal{P}_\lambda}}\otimes \ket{V_\lambda}$ for $\kett{I_{\mathcal{P}_\lambda}}\in\mathcal{P}_{\mathrm{AB},\lambda}\otimes\mathcal{P}_{\mathrm{anc},\lambda}$ and $\ket{V_\lambda}\in\mathcal{Q}^{d_1d_2}_{\mathrm{AB},\lambda}\otimes\mathcal{Q}^r_{\mathrm{anc},\lambda}$.
Thus,
\begin{align}
\tr_{\mathrm{anc}}(\kettbbra{V}{V}^{\otimes n})&=\tr_{\mathrm{anc}}\left(\bigoplus_{\substack{\lambda,\mu\vdash_{s}n}} \kettbbra{I_{\mathcal{P}_\lambda}}{I_{\mathcal{P}_\mu}}\otimes \ketbra{V_\lambda}{V_\mu}\right) \nonumber\\
&=\bigoplus_{\lambda\vdash_s n} \tr_{\mathcal{P}_{\mathrm{anc},\lambda}}\Big(\kettbbra{I_{\mathcal{P}_\lambda}}{I_{\mathcal{P}_\lambda}}\Big)\otimes \tr_{\mathcal{Q}^r_{\mathrm{anc},\lambda}}\Big(\ketbra{V_\lambda}{V_\lambda}\Big)\nonumber \\
&=\bigoplus_{\lambda\vdash_s n}I_{\mathcal{P}_{\mathrm{AB},\lambda}}\otimes \tr_{\mathcal{Q}^r_{\mathrm{anc},\lambda}}\Big(\ketbra{V_\lambda}{V_\lambda}\Big),\label{eq-12151045}
\end{align}
in which $\tr_{\mathrm{anc}}(\cdot)$ denotes the partial trace on all ancilla systems $\bigotimes_{j=1}^n \mathcal{H}_{\mathrm{anc},j}$.
Also, we can see
\begin{align}
C_\mathcal{E}^{\otimes n}=\bigoplus_{\lambda\vdash_{d_1d_2} n} I_{\mathcal{P}_{\mathrm{AB},\lambda}}\otimes C_{\mathcal{E},\lambda},\label{eq-12151044}
\end{align}
for certain $\mathcal{C}_{\mathcal{E},\lambda}\in\mathcal{L}(\mathcal{Q}^{d_1d_2}_{\mathrm{AB},\lambda})$. 
Comparing \cref{eq-12151045} with \cref{eq-12151044}, we find that $\tr_{\mathcal{Q}^r_{\mathrm{anc},\lambda}}(\ketbra{V_\lambda}{V_\lambda})=C_{\mathcal{E},\lambda}$ for $\lambda\vdash_s n$, and also $C_{\mathcal{E},\lambda}=0$ for those $\lambda$ with more than $s$ rows.
This means,
\begin{align}
\widetilde{T}_i\star C_{\mathcal{E}}^{\otimes n}&=\tr\left(\left(\bigoplus_{\lambda\vdash_{s} n} \frac{1}{\dim(\mathcal{P}_\lambda)\dim(\mathcal{Q}^r_\lambda)} \cdot I_{\mathcal{P}_{\mathrm{AB},\lambda}}\otimes \tr_{\mathcal{Q}^r_{\mathrm{anc},\lambda}}\Big(\bbra{I_{\mathcal{P}_\lambda}}\overline{T}_i\kett{I_{\mathcal{P}_\lambda}}\Big)\right)^{\mathrm{T}} \cdot C^{\otimes n}_\mathcal{E}\right)\nonumber \\
&=\sum_{\lambda\vdash_s n}\frac{1}{\dim(\mathcal{Q}^r_\lambda)}\tr\left(\tr_{\mathcal{Q}^r_{\mathrm{anc},\lambda}}\Big(\bbra{I_{\mathcal{P}_\lambda}}\overline{T}_i\kett{I_{\mathcal{P}_\lambda}}\Big)^{\mathrm{T}}\cdot C_{\mathcal{E},\lambda}\right) \nonumber \\
&=\sum_{\lambda\vdash_s n}\frac{1}{\dim(\mathcal{Q}^r_\lambda)}\tr\left(\tr_{\mathcal{Q}^r_{\mathrm{anc},\lambda}}\Big(\bbra{I_{\mathcal{P}_\lambda}}\overline{T}_i^{\mathrm{T}}\kett{I_{\mathcal{P}_\lambda}}\Big)\cdot \tr_{\mathcal{Q}^r_{\mathrm{anc},\lambda}}\Big(\ketbra{V_\lambda}{V_\lambda}\Big)\right) \nonumber\\
&=\overline{T}_i\star C_{\mathcal{V}}^{\otimes n} \label{eq-1270039}\\
&=\EE{\mathcal{W}\sim \dilation_r(\mathcal{E})}[T_i\star C_{\mathcal{W}}^{\otimes n}], \label{eq-1270040}
\end{align}
where \cref{eq-1270039} is due to item 2 of \cref{lemma-1242352} and \cref{eq-1270040} is due to item 1 of \cref{lemma-1242352}.
\end{proof}

The following result about bipartite pure states is widely used in quantum information theory (see, e.g., \cite{matsumoto2007universal}).
\begin{lemma}\label{lemma-3262205}
Suppose \(\ket{\psi}\in \mathcal{H}_{\mathrm{A}}\otimes\mathcal{H}_{\mathrm{B}}\cong \mathbb{C}^{d_1}\otimes\mathbb{C}^{d_2}\) is a vector and define $s=\min\{d_1, d_2\}$. Then, $\ket{\psi}^{\otimes n}$ has the form
\begin{equation*}
\ket{\psi}^{\otimes n}= \bigoplus_{\lambda\vdash_{s} n}\kett{I_{\mathcal{P}_\lambda}}\otimes\ket{\psi_\lambda},
\end{equation*}
where \(\ket{\psi_\lambda}\in\mathcal{Q}^{d_1}_{\mathrm{A},\lambda}\otimes \mathcal{Q}^{d_2}_{\mathrm{B},\lambda}\), and \(\kett{I_{\mathcal{P}_\lambda}}\) is the (unnormalized) maximally entangled state on the bipartite system \(\mathcal{P}_{\mathrm{A},\lambda}\otimes \mathcal{P}_{\mathrm{B},\lambda}\) w.r.t. the Young's orthogonal basis.
\end{lemma}
\begin{proof}
First, the state $\ket{\psi}^{\otimes n}$ is invariant under the action of $\texttt{p}_{\mathrm{A}}(\pi)\otimes\texttt{p}_{\mathrm{B}}(\pi)$ for any $\pi\in\mathfrak{S}_n$.
Due to the Schur-Weyl duality, we can see that
\begin{align}
\frac{1}{n!}\sum_{\pi\in\mathfrak{S}_n} \texttt{p}_{\mathrm{A}}(\pi)\otimes\texttt{p}_{\mathrm{B}}(\pi)&= \frac{1}{n!}\sum_{\pi\in\mathfrak{S}_n} \bigoplus_{\substack{\lambda\vdash_{d_1} n\\ \mu\vdash_{d_2}n}} \texttt{p}_{\mathrm{A},\lambda}(\pi)\otimes\texttt{p}_{\mathrm{B},\mu}(\pi)\otimes I_{\mathcal{Q}_{\mathrm{A},\lambda}^{d_1}}\otimes I_{\mathcal{Q}_{\mathrm{B},\mu}^{d_2}} \nonumber\\
&=\bigoplus_{\substack{\lambda\vdash_{d_1} n\\ \mu\vdash_{d_2}n}} \frac{1}{n!}\sum_{\pi\in\mathfrak{S}_n} \texttt{p}^*_{\mathrm{A},\lambda}(\pi)\otimes\texttt{p}_{\mathrm{B},\mu}(\pi)\otimes I_{\mathcal{Q}_{\mathrm{A},\lambda}^{d_1}}\otimes I_{\mathcal{Q}_{\mathrm{B},\mu}^{d_2}}\label{eq-12131226}\\
&=\bigoplus_{\lambda\vdash_s n}\frac{1}{\dim(\mathcal{P}_\lambda)} \kettbbra{I_{\mathcal{P}_\lambda}}{I_{\mathcal{P}_\lambda}} \otimes I_{\mathcal{Q}_{\mathrm{A},\lambda}^{d_1}}\otimes I_{\mathcal{Q}_{\mathrm{B},\lambda}^{d_2}}, \label{eq-12131216}
\end{align}
where in \cref{eq-12131226} the $(\cdot)^*$ denotes complex conjugate (w.r.t. the Young's orthogonal basis so that $\texttt{p}_{\lambda}(\pi)$ is a real matrix~\cite{ceccherini2010representation}), \cref{eq-12131216} uses the fact that the only subspace invariant under $\texttt{p}^*_{\mathrm{A},\lambda}(\pi)\otimes\texttt{p}_{\mathrm{B},\mu}(\pi)$ for all $\pi\in\mathfrak{S}_n$ is spanned by $\kett{I_{\mathcal{P}_\lambda}}\in\mathcal{P}_{\mathrm{A},\lambda}\otimes\mathcal{P}_{\mathrm{B},\mu}$ when $\lambda=\mu$, and is the trivial space $\{0\}$ otherwise.\footnote{We can see this by considering the representation isomorphism $\mathcal{P}_{\mathrm{A},\lambda}^*\otimes\mathcal{P}_{\mathrm{B},\mu}\stackrel{\mathfrak{S}_n}{\cong} \mathcal{L}(\mathcal{P}_{\mathrm{A},\lambda},\mathcal{P}_{\mathrm{B},\mu})$. Due to Schur's lemma, any linear operator in $\mathcal{L}(\mathcal{P}_{\mathrm{A},\lambda},\mathcal{P}_{\mathrm{B},\mu})$ that commutes with the action of $\pi$ must be proportional to the identity operator when $\lambda=\mu$ and $0$ otherwise.} 
Thus, the vector $\ket{\psi}^{\otimes n}$ lies in the support of the projector $\frac{1}{n!}\sum_{\pi\in\mathfrak{S}_n}\texttt{p}_\mathrm{A}(\pi)\otimes\texttt{p}_{\mathrm{B}}(\pi)$, which, according to \cref{eq-12131216}, means $\ket{\psi}^{\otimes n}= \bigoplus_{\lambda\vdash_{s} n}\kett{I_{\mathcal{P}_\lambda}}\otimes\ket{\psi_\lambda}$.
\end{proof}

\section{Hardness of isometry channel discrimination}\label{sec-3200450}

In this section, we establish the hardness of distinguishing between isometry channels with specific structures.
In \cref{sec-321632}, 
we present a construction of a family of isometries that are $\varepsilon$-close and prove that 
$\Omega(1/\varepsilon)$ queries are necessary to distinguish between them. Moreover, in  \cref{sec-1110117}, we present a different family of isometries, that are also  $\varepsilon$-close, but $\Omega(1/\varepsilon^2)$ queries are needed to distinguish between them.
For ease of reference, we denote these constructions as Type I and Type II hard instances, respectively.

\subsection{Type I: Heisenberg-scaling hardness}\label{sec-321632}
\subsubsection{Definition}\label{sec:hard-d1=rd2}
Let $D\geq d$ be two positive integers.
We consider isometries with input dimension $d$ and output dimension $D$.

Let $\varepsilon\in (0,1/2)$, $d\ge 2$ and $d'\coloneqq 2\left\lfloor \frac{d}{2}\right\rfloor$. Let $\{\ket{1}_\mathrm{A},\ldots,\ket{d}_\mathrm{A}\}$ and $\{\ket{1}_\mathrm{B},\ldots,\ket{D}_\mathrm{B}\}$ be two orthonormal bases of 
 the Hilbert spaces $\mathcal{H}_{\mathrm{A}}\cong \mathbb{C}^{d}$  and $\mathcal{H}_{\mathrm{B}}\cong \mathbb{C}^{D}$, respectively. 
Define the linear operator $V_{0}:\mathcal{H}_{\mathrm{A}}\rightarrow \mathcal{H}_{\mathrm{B}}$ by
\[V_0\coloneqq
\begin{cases}
\sqrt{1-\varepsilon^2}\sum_{i=1}^{d'}\ket{i}_\mathrm{B}\bra{i}_\mathrm{A}+\ket{d}_\mathrm{B}\bra{d}_\mathrm{A},&\quad \textup{if $d$ is odd}\\
\sqrt{1-\varepsilon^2}\sum_{i=1}^{d'}\ket{i}_\mathrm{B}\bra{i}_\mathrm{A},&\quad \textup{if $d$ is even.}\\
\end{cases}
\]
Define the linear operator $\Delta:\mathcal{H}_{\mathrm{A}}\rightarrow\mathcal{H}_{\mathrm{B}}$  by
\[\Delta\coloneqq \mathrm{i}\left(\sum_{i=1}^{\lfloor d/2\rfloor} \ket{i}_\mathrm{B}\bra{i}_{\mathrm{A}}-\sum_{i= \lfloor d/2\rfloor+1}^{d'} \ket{i}_\mathrm{B}\bra{i}_{\mathrm{A}}\right),\]
where $\mathrm{i}$ denotes the imaginary unit.

Let $U\in\mathbb{U}_{d'}$ be a unitary. 
We define the isometry $V_{\varepsilon,U}:\mathcal{H}_{\mathrm{A}}\rightarrow\mathcal{H}_{\mathrm{B}}$ by
\begin{equation}\label{eq-1111810}
V_{\varepsilon,U} \coloneqq V_0+\varepsilon U\Delta U^\dag.
\end{equation}
Here, we  abuse the notation and assume that $U$ can act either on $\mathcal{H}_\mathrm{A}$ or $\mathcal{H}_\mathrm{B}$ (we identify $\ket{i}_\mathrm{A}$ with $\ket{i}_\mathrm{B}$ for $i\in [d]$ and $U$ acts trivially on $\ket{d'+1}_\mathrm{B},\ldots,\ket{D}_\mathrm{B}$). 

Note that $V_{\varepsilon,U}$ is an isometry as one can verify:
\begin{align}
V_{\varepsilon,U}^\dag V_{\varepsilon,U}&=
V_0^\dag V_0+\varepsilon\left(V_0^\dag U\Delta U^\dag +U\Delta^\dag U^\dag V_0\right)+\varepsilon^2 U \Delta^\dag \Delta U^\dag \nonumber \\
&=V_0^\dag V_0+ \varepsilon^2 U\Delta^\dag \Delta U^\dag \nonumber\\
&=I_\mathrm{A}.\nonumber
\end{align}
Furthermore, we observe that $\bbrakett{V_0}{\Delta}=\tr(V_0^\dag \Delta) = 0$.

\subsubsection{Hardness of the discrimination problem}
We then obtain the following main result of this section.
\begin{theorem}\label{thm-1111809}
 Let $C>0$ be a constant. Suppose $d\geq 8/C$. Let $\mathcal{N}\subseteq\{V_{\varepsilon,U}\,|\, U\in\mathbb{U}_{2\lfloor\frac{d}{2}\rfloor}\}$ with $V_{\varepsilon,U}$ defined in \cref{eq-1111810} of cardinality $|\mathcal{N}|\geq \exp(Cd^2)$.  Then, any algorithm that can distinguish between the isometries in $\mathcal{N}$ with success probability at least $2/3$ must use at least $\min\{\frac{C^{3/2}}{21},\frac{1}{\sqrt{2e^5}}\} \cdot \frac{d^2}{\varepsilon}$ queries. 
\end{theorem}
\begin{proof}
Suppose there is an algorithm that distinguishes between the isometries in $\mathcal{N}$  using $n$ queries. 
We may assume $n\leq \frac{1}{\sqrt{2e^5}}d^2/\varepsilon$ as otherwise the theorem is proven.

Recall that each isometry $V$ in $\mathcal{N}$ can be written as 
\begin{align}
V &=  V_0 + \varepsilon U\Delta U^\dag\nonumber
\end{align}
for some unitary $U\in \mathbb{U}_{d'}$. 

Let $\{T_V\}_{V\in\mathcal{N}}$ be a sequential tester that describes the algorithm for distinguishing between the isometries in the set $\mathcal{N}$.
Using the bound on the cardinality  of the set $\mathcal{N}$, we can bound the success probability as follows
\begin{align}
\Pr[\textup{success}]&=\frac{1}{|\mathcal{N}|}\cdot \sum_{V\in\mathcal{N}} T_V \star \kettbbra{V}{V}^{\otimes n} \nonumber \\
&\leq \exp(-Cd^2)\cdot \sum_{V\in\mathcal{N}} T_V \star \kettbbra{V}{V}^{\otimes n}, \nonumber
\end{align}
where $\kettbbra{V}{V}^{\otimes n}$ is an $n$-comb on $(\mathcal{H}_{\mathrm{A},1},\mathcal{H}_{\mathrm{B},1},\ldots,\mathcal{H}_{\mathrm{A},n},\mathcal{H}_{\mathrm{B},n})$, and $\mathcal{H}_{\mathrm{A},i}$ and $\mathcal{H}_{\mathrm{B},i}$ denote the input and output spaces of the $i$-th query to $V$, respectively.

Observe that for $V\in\mathcal{N}$, we can express $\kett{V}^{\otimes n}$ as follows
\begin{align}
    \kett{V}^{\otimes n}&=\left(  \kett{V_0} +\eps \kett{U\Delta U^\dagger }  \right)^{\otimes n} \nonumber
    \\&=\sum_{i=0}^n \varepsilon^i  \sum_{\substack{S\subseteq[n]\\ |S|=i}} \kett{V_{0}}^{\otimes [n]\setminus S}\otimes \kett{U\Delta U^\dag}^{\otimes S}\nonumber\\
    &= \sum_{i=0}^n\varepsilon^i \left(U\otimes U^*\right)^{\otimes n} \sum_{\substack{S\subseteq[n]\\ |S|=i}}\ket{\gamma_S},\label{eq-1111819}
\end{align}
where $U^*$ acts on on $\mathcal{H}_{\mathrm{A},i}$ and $U$ acts on $\mathcal{H}_{\mathrm{B},i}$ and  for any  $S \subseteq [n]$, we define the vector
\begin{equation}\label{eq-1111818}
\ket{\gamma_S}\coloneqq \kett{V_{0}}^{\otimes [n]\setminus S}\otimes \kett{\Delta}^{\otimes S}.
\end{equation}
Note that, since $\bbrakett{V_0}{\Delta}=\tr(V_0^\dag \Delta)=0$, the vectors $\{\ket{\gamma_S}\}_{S\subseteq[n]}$ are pairwise orthogonal.  

Next, for $S\subseteq[n]$,  we introduce the operator $\Gamma_S$ 
\begin{equation}\label{eq-1111820}
\Gamma_S\coloneqq \EE{\substack{U\sim\mathbb{U}_{d'}}}\left[\Big(U\otimes U^*\Big)^{\otimes n}\ketbra{\gamma_S}{\gamma_S} \left(U^\dag\otimes U^\textup{T}\right)^{\otimes n}\right].
\end{equation}
Similarly, $\{\supp(\Gamma_S)\}_{S\subseteq[n]}$ are  pairwise orthogonal. 
To see this, we observe that $U\otimes U^*$ fixes $\kett{V_0}$ and $U\otimes U^* \kett{\Delta}$ remains orthogonal to $\kett{V_0}$ for any $U\in\mathbb{U}_{d'}$.

By \cref{lemma-1200312}, there exist a set of positive numbers $\{\lambda_S\,|\, S\subseteq[n]\}$ satisfying for any $V\in\mathcal{N}$,
\begin{align}
\sum_{S\subseteq[n]} \lambda_S &\leq 3d^4\exp\!\left(\sqrt[3]{54n^2\varepsilon^2d^2}\right),\nonumber
    \\\kettbbra{V}{V}^{\otimes n}&\leq \sum_{S\subseteq [n]} \lambda_S \Gamma_S.\nonumber
\end{align}
Hence the success probability can be  bounded as
\begin{align}
\Pr[\textup{success}]&\leq \exp(-Cd^2)\cdot \sum_{V\in \mathcal{N}} T_V\star \kettbbra{V}{V}^{\otimes n} \nonumber \\
&\leq \exp(-Cd^2)\cdot \sum_{V\in \mathcal{N}} T_V\star \sum_{S\subseteq [n]} \lambda_S\Gamma_S \nonumber \\
&= \exp(-Cd^2)\cdot \sum_{T\subseteq[n]} \lambda_T \cdot \sum_{V\in \mathcal{N}} T_V\star \sum_{S\subseteq[n]} \frac{\lambda_S}{\sum_{T\subseteq[n]}\lambda_T}\Gamma_S \nonumber \\
&\leq \exp(-Cd^2)\cdot \sum_{S\subseteq[n]} \lambda_S, \label{eq-1191729}
\\&\leq \exp(-Cd^2)\cdot 3d^4\exp\!\left(\sqrt[3]{54n^2\varepsilon^2d^2}\right), \nonumber
\end{align}
where \cref{eq-1191729} uses the fact that
\begin{itemize}
\item $\sum_{V\in\mathcal{N}} T_V$ forms an $(n+1)$-comb with input and output dimensions $1$, and
\item $\sum_{S\subseteq[n]} \frac{\lambda_S}{\sum_{T\subseteq[n]}\lambda_T}\Gamma_S$ forms a probabilistic $n$-comb. Here, we use the fact that each $\Gamma_S$ is a probabilistic comb (see \cref{lemma-1230037}) and their convex combination is also a probabilistic comb,
\end{itemize}
and the contraction of an $(n+1)$-comb with input and output dimensions $1$ and a probabilistic $n$-comb is at most $1$.  

Since the algorithm succeeds with probability $\Pr[\textup{success}]\geq 2/3$, we obtain that
\begin{align}
54n^2\varepsilon^2d^2&\geq \left(Cd^2-4\ln(d)+\ln(2/9) \right)^3\nonumber \\
&\geq \left(Cd^2-4d\right)^3 \geq \frac{C^3}{8}d^6, \nonumber 
\end{align}
where we used $\ln(d)\leq d-1$, $\ln(9/2)\leq 4$ and $d\geq 8/C$.
Finally, we deduce
\[n\geq \frac{C^{3/2}}{12\sqrt{3}}\cdot \frac{d^2}{\varepsilon}> \frac{C^{3/2}}{21}\cdot \frac{d^2}{\varepsilon}.\]
\end{proof}

\paragraph{Technical lemmas.}
For $S\subseteq[n]$, recall the definition of $\ket{\gamma_S}$ (see \cref{eq-1111818}) and $\Gamma_S$ (see \cref{eq-1111820}). 
\begin{lemma}\label{lemma-1230037}
 $\ketbra{\gamma_S}{\gamma_S}$ and $\Gamma_S$  are  probabilistic $n$-combs. 
\end{lemma}
\begin{proof}
Observe that
\[\tr_{\mathcal{H}_\mathrm{B}}(\kettbbra{V_0}{V_0})=
\begin{cases}
(1-\varepsilon^2)\sum_{i=1}^{d'}\ketbra{i}{i}_\mathrm{A}+\ketbra{d}{d}_\mathrm{A},&\quad\textup{if $d$ is odd}\\
(1-\varepsilon^2)\sum_{i=1}^{d'}\ketbra{i}{i}_\mathrm{A},&\quad\textup{if $d$ is even}
\end{cases},\]
so $\tr_{\mathcal{H}_\mathrm{B}}(\kettbbra{V_0}{V_0})\leq I_{\mathrm{A}}$. Moreover, we have that  
\[\tr_{\mathcal{H}_\mathrm{B}}(\ketbra{\Delta}{\Delta})=\sum_{i=1}^{d'}\ketbra{i}{i}_\mathrm{A}\leq I_{\mathrm{A}}.\]
hence $\ket{\gamma_S}$ is a probabilistic $n$-comb.

The similar reasoning shows that $(U\otimes U^*)^{\otimes n}\ketbra{\gamma_S}{\gamma_S} (U^\dag\otimes U^\textup{T})^{\otimes n}$ is also a probabilistic $n$-comb for any $U\in \mathbb{U}_{d'}$. Finally, we deduce that  $\Gamma_S$  is also a probabilistic $n$-comb since it is a convex combination of probabilistic combs.
\end{proof}

\begin{lemma}\label{lemma-1200312}
Suppose $d\geq 2$ and $n\leq \frac{1}{\sqrt{2e^5}}d^2/\varepsilon$. There exists a set of positive numbers $\{\lambda_S\,|\, S\subseteq[n]\}$ such that  for any $V\in\mathcal{N}$,
\begin{align}
\sum_{S\subseteq[n]} \lambda_S &\leq 3d^4\exp\!\left(\sqrt[3]{54n^2\varepsilon^2d^2}\right),\nonumber
    \\\kettbbra{V}{V}^{\otimes n}&\leq \sum_{S\subseteq [n]} \lambda_S \Gamma_S.\label{eq-11917281}
\end{align}
\end{lemma}
\begin{proof}
We remark that, because of \cref{eq-1111819},  $\kett{V}^{\otimes n}$ is contained in $\bigoplus_{S\subseteq [n]} \supp(\Gamma_S)$. 
Moreover, by \cref{fact-5122103}, to establish \cref{eq-11917281}, it is sufficient to show that 
\begin{equation}\label{eq-1191750}
\sum_{S\subseteq [n]} \frac{1}{\lambda_S} \tr\!\left(\Gamma_S^{-1}\kettbbra{V}{V}^{\otimes n}\right)\leq 1.
\end{equation}
To this end, we bound the left hand side of \cref{eq-1191750} as
\begin{align}
\sum_{S\subseteq [n]} \frac{1}{\lambda_S} \tr\!\left(\Gamma_S^{-1}\kettbbra{V}{V}^{\otimes n}\right)&=\sum_{S\subseteq [n]} \frac{1}{\lambda_S} \varepsilon^{2|S|}\tr\!\left(\Gamma_S^{-1}(U\otimes U^*)^{\otimes n}\ketbra{\gamma_S}{\gamma_S}(U^\dag \otimes U^\mathrm{T})^{\otimes n}\right) \label{eq-1192352} \\
&=\sum_{S\subseteq [n]} \frac{1}{\lambda_S} \varepsilon^{2|S|}\tr\!\left(\Gamma_S^{-1}\ketbra{\gamma_S}{\gamma_S}\right) \label{eq-1192351} \\
&\leq \sum_{S\subseteq [n]} \frac{1}{\lambda_S} \varepsilon^{2|S|} \binom{d^2+|S|-2}{|S|}.\label{eq-1192350}
\end{align}
Here,  \cref{eq-1192352} follows from \cref{eq-1111819} and the fact that $(U\otimes U^*)^{\otimes n}\ket{\gamma_S}\in \supp(\Gamma_S)$, \cref{eq-1192351} uses the fact that $\Gamma_S$ commutes with $(U\otimes U^*)^{\otimes n}$, and \cref{eq-1192350} is an application of  \cref{lemma-12271509} where  $\ket{\gamma_S}$ is viewed as a vector in the space
\[\spanspace\left(\left\{\kett{V_0}^{\otimes [n]\setminus S}\otimes \ket{\psi}^{\otimes S} \,\, \bigg|\,\,  \ket{\psi}\in\spanspace\{\ket{1}_\mathrm{A},\ldots,\ket{d}_\mathrm{A}\}\otimes  \spanspace\{\ket{1}_\mathrm{B},\ldots,\ket{d}_\mathrm{B}\},\,\, \ket{\psi}\in \kett{V_0}^{\perp}\right\}\right),\]
which is invariant under the action $(U\otimes U^*)^{\otimes n}$. 
This space has dimension $\binom{d^2+|S|-2}{|S|}$ since it is isomorphic to the symmetric space $\lor^{|S|} \mathbb{C}^{d^2-1}$~\cite{harrow2013church}.

For  $S\subseteq [n]$, we set $\lambda_S=\mu_i/\binom{n}{i}$ where $i=|S|$ for some positive number $\mu_i$.
Then, we observe $\sum_{S\subseteq[n]}\lambda_S=\sum_{i=0}^n \mu_i$ and \cref{eq-1192350} can be written as 
\begin{equation}\label{eq-1200020}
\sum_{i=0}^n \frac{1}{\mu_i} \binom{n}{i}^2 \varepsilon^{2i}\binom{d^2+i-2}{i}.
\end{equation}
Our goal is to look for positive numbers $\mu_0,\ldots,\mu_n$ such that  $\sum_{i=0}^n\mu_i$ is small and \cref{eq-1200020} is upper bounded by $1$. For this, \cref{eq-1200020} can be bounded using \cref{lemma-12272204} as
\begin{align}
\eqref{eq-1200020} &\leq \sum_{i=0}^n\frac{1}{\mu_i} \exp\!\left(2nH\!\left(\frac{i}{n}\right) +2i\ln(\varepsilon) + (d^2+i) H\!\left(\frac{i}{d^2+i}\right)\right) \nonumber\\
&=\sum_{i=0}^n\frac{1}{\mu_i}\exp\!\left(-2i \ln\!\left(\frac{i}{n\varepsilon}\right) -2(n-i)\ln\!\left(\frac{n-i}{n}\right)+ i \ln \!\left(1+\frac{d^2}{i}\right)+d^2 \ln\!\left(1+\frac{i}{d^2}\right)\right)\nonumber \\
&\leq \sum_{i=0}^n \frac{1}{\mu_i} \exp\!\left(-i\ln\!\left(\frac{i^2}{n^2\varepsilon^2}\right)+i\ln\!\left(1+\frac{d^2}{i}\right)+3i\right),\label{eq-1200031}
\end{align}
where the last inequality uses
\begin{align}
2(n-i)\ln\!\left(\frac{n}{n-i}\right)&\leq 2(n-i)\left(\frac{n}{n-i}-1\right)= 2i \;\; \text{ and } \;\; d^2\ln\!\left(1+\frac{i}{d^2}\right)\leq d^2\frac{i}{d^2}=i.\nonumber
\end{align}
To further bound \cref{eq-1200031}, we consider two cases depending on $i$:
\begin{itemize}
\item When $i< d^2$, we have
\begin{align}
-i\ln\!\left(\frac{i^2}{n^2\varepsilon^2}\right)+i\ln\!\left(1+\frac{d^2}{i}\right)+3i &\leq -i\ln\!\left(\frac{i^2}{n^2\varepsilon^2}\right)+i\ln\!\left(\frac{2d^2}{i}\right)+3i \nonumber\\
&=3i\ln\!\left(\frac{\sqrt[3]{2e^3n^2\varepsilon^2d^2}}{i}\right) \nonumber \\
&\leq \sqrt[3]{54n^2\varepsilon^2d^2}, \nonumber
\end{align}
where we used  \cref{lamma-12281043} in the last inequality. 

\item When $i\geq d^2$, since we assumed that $n\leq \frac{1}{\sqrt{2e^5}}d^2/\varepsilon$ we have $i\geq \sqrt{2e^5}n\varepsilon$, thus
\begin{align}
-i\ln\!\left(\frac{i^2}{n^2\varepsilon^2}\right)+i\ln\!\left(1+\frac{d^2}{i}\right)+3i&\leq -i \ln(2e^5)+i\ln(2)+3i=-i\ln\!\left(2e^5/2e^3\right)=-2i.\nonumber
\end{align}
\end{itemize}

Hence, the choice 
\begin{align*}
    \mu_i = \begin{cases}
        2d^2\exp\!\left(\sqrt[3]{54n^2\varepsilon^2d^2}\right), &\quad\textup{if $i<d^2$} \\
        \exp(-i), &\quad\textup{if $i\geq d^2$,}
    \end{cases}
\end{align*}
ensures that  \cref{eq-1200031} can be upper bounded by
\begin{align}
\eqref{eq-1200031} &\leq \sum_{i<d^2} \frac{1}{\mu_i} \exp\!\left(\sqrt[3]{54n^2\varepsilon^2d^2}\right)+\sum_{i\geq d^2} \frac{1}{\mu_i}\exp(-2i) \nonumber \\
&\leq \frac{1}{2}+\sum_{i\geq d^2} \exp(-i) 
\leq\frac{1}{2}+ \exp(-d^2) \frac{e}{e-1} < 1,\nonumber
\end{align}
where we used that $d^2\geq 2$. Finally, we have that 
\begin{align}
\sum_{i=0}^n\mu_i &\leq 2d^4\exp\!\left(\sqrt[3]{54n^2\varepsilon^2d^2}\right)+\exp(-d^2)\frac{e}{e-1} \nonumber \\
& < 2d^4\exp\!\left(\sqrt[3]{54n^2\varepsilon^2d^2}\right)+\frac{1}{2} < 3d^4\exp\!\left(\sqrt[3]{54n^2\varepsilon^2d^2}\right),\nonumber
\end{align}
which concludes the proof.
\end{proof}

\subsection{Type II: classical-scaling hardness}\label{sec-1110117}

\subsubsection{Definition}\label{sec-1110250}
Let $D> d$ be two positive integers.
We consider isometries with input dimension $d$ and output dimension $D$.

Let  $\varepsilon\in (0,1/2)$ and $d'\leq  \min\{d,D-d\}$ be a positive number.
Let $\{\ket{1}_\mathrm{A},\ldots,\ket{d}_\mathrm{A}\}$ and $\{\ket{1}_\mathrm{B},\ldots,\ket{D}_\mathrm{B}\}$ be two (arbitrary) orthonormal bases of 
 the Hilbert spaces $\mathcal{H}_{\mathrm{A}}\cong \mathbb{C}^{d}$  and $\mathcal{H}_{\mathrm{B}}\cong \mathbb{C}^{D}$, respectively. 
Define the linear operator $V_0:\mathcal{H}_{\mathrm{A}}\rightarrow \mathcal{H}_{\mathrm{B}}$ as
\[V_0\coloneqq \sqrt{1-\varepsilon^2}\sum_{i=1}^{d'}\ket{i}_\mathrm{B}\bra{i}_\mathrm{A}+\sum_{i=d'+1}^{d} \ket{i}_\mathrm{B}\bra{i}_{\mathrm{A}},\]
and the linear operator $\Delta:\mathcal{H}_{\mathrm{A}}\rightarrow \mathcal{H}_{\mathrm{B}}$ as
\[\Delta\coloneqq \sum_{i=1}^{d'}\ket{d+i}_\mathrm{B}\bra{i}_\mathrm{A}.\]
Let $U\in\mathbb{U}_{D-d}$ be a unitary. We define the isometry $V_{\varepsilon,U}:\mathcal{H}_{\mathrm{A}}\rightarrow\mathcal{H}_{\mathrm{B}}$ as
\begin{equation}\label{eq-2141418}
V_{\varepsilon,U} \coloneqq  V_0 + \varepsilon U\Delta. 
\end{equation}
Here, $U\in\mathbb{U}_{D-d}$ acts on the subspace spanned by $\{\ket{d+1}_\mathrm{B},\ldots,\ket{D}_\mathrm{B}\}$.

It can be easily checked that $V_{\varepsilon,U}$ is an isometry:
\[V_{\varepsilon,U}^\dag V_{\varepsilon,U}=(1-\varepsilon^2)\sum_{i=1}^{d'}\ket{i}_\mathrm{A}\bra{i}_\mathrm{A}+\sum_{i=d'+1}^{d} \ket{i}_\mathrm{A}\bra{i}_{\mathrm{A}}+\varepsilon^2\sum_{i=1}^{d'}\ket{i}_\mathrm{A}\bra{i}_\mathrm{A}=I_{\mathrm{A}}.\]
Moreover, note that the images of $V_0$ and  $U\Delta$  are orthogonal.

\subsubsection{Hardness of the discrimination problem}
We then obtain the following main result of this section.
\begin{theorem}\label{thm-1110122}
Let $C>0$ be a constant. Suppose $D-d\geq 64/C^2$. Let $\mathcal{N} \subseteq \{V_{\varepsilon,U}\,|\, U\in \mathbb{U}_{D-d}\}$ with $V_{\varepsilon,U}$ defined in \cref{eq-2141418}
of cardinality 
$|\mathcal{N}|\geq \exp(C(D-d) \min\{d,D-d\})$. Then, any algorithm that can distinguish between the  isometries in $\mathcal{N}$ with success probability at least $2/3$ must use at least $\min\{\frac{C^2}{32},\frac{1}{2e^4}\}\cdot \frac{(D-d) \min\{d,D-d\}}{\varepsilon^2}$ queries.
\end{theorem}
\begin{proof}
Let $d_{\textup{min}}\coloneqq \min\{d,D-d\}$. 
Suppose there is an algorithm that distinguishes between the isometries in $\mathcal{N}$  using $n$ queries. 
We may assume $n\leq  \frac{1}{2e^4}(D-d)d_{\textup{min}}/\varepsilon^2$ as otherwise the theorem is proven.

Recall that each isometry $V$ in $\mathcal{N}$ is of the form
\begin{align}
V&= V_0 + \varepsilon U\Delta,\nonumber 
\end{align}
for some unitary $U\in \mathbb{U}_{D-d}$.

Let $\{T_V\}_{V\in\mathcal{N}}$ be a sequential tester that describes the algorithm for distinguishing between the isometries in the set $\mathcal{N}$.
Using the bound on the cardinality  of the set $\mathcal{N}$, we can bound the success probability as follows
\begin{align}
\Pr[\textup{success}]&=\frac{1}{|\mathcal{N}|}\cdot \sum_{V\in\mathcal{N}} T_V \star \kettbbra{V}{V}^{\otimes n} \nonumber \\
&\leq \exp\!\left(-C(D-d)d_{\textup{min}}\right)\cdot \sum_{V\in\mathcal{N}} T_V \star \kettbbra{V}{V}^{\otimes n}. \nonumber
\end{align}
where $\kettbbra{V}{V}^{\otimes n}$ is an $n$-comb on $(\mathcal{H}_{\mathrm{A},1},\mathcal{H}_{\mathrm{B},1},\ldots,\mathcal{H}_{\mathrm{A},n},\mathcal{H}_{\mathrm{B},n})$, and $\mathcal{H}_{\mathrm{A},i}$ and $\mathcal{H}_{\mathrm{B},i}$ denote the input and output spaces of the $i$-th query to $V$, respectively.

Observe that for $V\in\mathcal{N}$, we can express $\kett{V}^{\otimes n}$ as follows
\begin{align}
\kett{V}^{\otimes n}&=\left(\kett{V_0}+\eps \kett{U\Delta}\right)^{\otimes n}\nonumber
\\&=\sum_{i=0}^n \varepsilon^i  \sum_{\substack{S\subseteq[n]\\ |S|=i}} \kett{V_0}^{\otimes [n]\setminus S}\otimes \kett{U\Delta}^{\otimes S} \nonumber 
    \\&=\sum_{i=0}^n\varepsilon^i \sqrt{\binom{n}{i}} U^{\otimes n}\ket{\gamma_i}\label{eq-12271523}
\end{align}
where  $U^{\otimes n}$ acts as $(I_{d}\oplus U)^{\otimes n}$ on $\bigotimes_{i=1}^n \mathcal{H}_{\mathrm{B},i}$ and 
for  $i\in \{0,1,\ldots,n\}$, we define the vector
\begin{equation}\label{eq-12270131}
\ket{\gamma_i}\coloneqq \frac{1}{\sqrt{\binom{n}{i}}} \sum_{\substack{S\subseteq [n]\\ |S|=i}}\kett{V_0}^{\otimes [n]\setminus S}\otimes \kett{\Delta}^{\otimes S}.
\end{equation}
Observe that, since $\bbrakett{V_0}{\Delta}=\tr(V_0^\dag \Delta)=0$, the vectors  $\{\ket{\gamma_i}\}_{i\in [n]}$ are pairwise orthogonal. 
Next, for $i\in [n]$, we introduce the operator $\Gamma_i$ 
\begin{equation}\label{eq-12272104}
\Gamma_i\coloneqq \EE{U\sim\mathbb{U}_{D-d}}\left[U^{\otimes n}\ketbra{\gamma_i}{\gamma_i} U^{\dag\otimes n}\right].
\end{equation}
We observe that $\{\supp(\Gamma_i)\}_{i\in [n]}$ are also pairwise orthogonal.  To see this, we note that the number of $\kett{V_0}$ in  $\ket{\gamma_i}$ is exactly $n-i$, and $\kett{V_0}$ is orthogonal to $U\kett{\Delta}$ for any $U\in\mathbb{U}_{D-d}$. 

By \cref{lemma-12281107}, there exists a set of positive numbers $\{\lambda_i\}_{i=0}^n$ satisfying for any $V\in\mathcal{N}$, 
\begin{align}
\sum_{i=0}^n\lambda_i&\leq 4(D-d)^2d_{\textup{min}}^2\exp\!\left(\sqrt{8n\varepsilon^2(D-d)d_{\textup{min}}}\right), \nonumber \\
    \kettbbra{V}{V}^{\otimes n}&\leq \sum_{i=0}^n \lambda_i \Gamma_i. \nonumber 
\end{align}
Hence, the success probability can be  bounded as
\begin{align}
\Pr[\textup{success}]&\leq \exp\!\left(-C (D-d)d_{\textup{min}}\right)\cdot \sum_{V\in \mathcal{N}} T_V\star \kettbbra{V}{V}^{\otimes n} \nonumber \\
&\leq \exp\!\left(-C (D-d)d_{\textup{min}}\right)\cdot \sum_{V\in \mathcal{N}} T_V\star \sum_{j=0}^n \lambda_j\Gamma_j \nonumber \\
&= \exp\!\left(-C (D-d)d_{\textup{min}}\right)\cdot \sum_{k=0}^n \lambda_k \cdot \sum_{V\in \mathcal{N}} T_V\star \sum_{j=0}^n \frac{\lambda_j}{\sum_{k=0}^n\lambda_k}\Gamma_j \nonumber \\
&\leq \exp\!\left(-C (D-d)d_{\textup{min}}\right)\cdot \sum_{i=0}^n \lambda_i\label{eq-12272057} \\
&\leq  \exp\!\left(-C (D-d)d_{\textup{min}}\right)\cdot 4(D-d)^2d_{\textup{min}}^2\exp\!\left(\sqrt{8n\varepsilon^2(D-d)d_{\textup{min}}}\right),\nonumber
\end{align}
where \cref{eq-12272057} uses the fact that 
\begin{itemize}
\item $\sum_{V\in\mathcal{N}} T_V$ forms an $(n+1)$-comb with input and output dimensions $1$, and
\item $\sum_{j=0}^n \frac{\lambda_j}{\sum_{k=0}^n\lambda_k}\Gamma_j$ forms  a probabilistic $n$-comb. Here, we use the facts that each $\Gamma_i$ are probabilistic combs (see \cref{lemma-12270148}) and their convex combination is also a probabilistic comb,
\end{itemize}
and the contraction of an $(n+1)$-comb with input and output dimensions $1$ and a probabilistic $n$-comb is at most $1$.

For ease of notation, we denote by  $Q=(D-d)d_{\textup{min}}$. Since the algorithm succeeds with probability $\Pr[\textup{success}]\geq 2/3$, we obtain that
\begin{align}
    \sqrt{8n\varepsilon^2Q}&\geq CQ-2\ln(Q)+\ln(1/6) \nonumber \\
& \geq CQ- 4\sqrt{Q}\geq \frac{C}{2}Q, \nonumber
\end{align}
where we used $\ln(\sqrt{Q})\leq \sqrt{Q}-1$, $\ln(6)\leq 2$ and $Q=(D-d)d_{\textup{min}}\geq D-d\geq 64/C^2$.
Finally, we conclude
\[n\geq \frac{C^2}{32}\cdot\frac{Q}{\varepsilon^2}=\frac{C^2}{32}\cdot\frac{(D-d)d_{\textup{min}}}{\varepsilon^2}.\]
\end{proof}

\paragraph{Technical lemmas.}
For $i\in[n]$, recall the definition of  $\ket{\gamma_i}$ (see \cref{eq-12270131}) and $\Gamma_i$ (see \cref{eq-12272104}). 
\begin{lemma}\label{lemma-12270148}
 $\ketbra{\gamma_i}{\gamma_i}$ and $\Gamma_i$ are  probabilistic $n$-combs. 
\end{lemma}
\begin{proof}
Denote by $\ket{\gamma_i^n}$ the state $\ket{\gamma_i}$ defined in \cref{eq-12270131}. 
We prove, by induction on $n$, that for all $i\le n$,  $\ketbra{\gamma_i^n}{\gamma_i^n}$ is a probabilistic $n$-comb.

First, we observe that 
\begin{equation}\label{eq-1120121}
\begin{gathered}
\tr_{\mathcal{H}_{\mathrm{B}}}(\kettbbra{V_0}{V_0})=V_0^\textup{T} V_0\leq I_\mathrm{A}
,\quad\quad \tr_{\mathcal{H}_{\mathrm{B}}}(\kettbbra{\Delta}{V_0})=\Delta^{\textup{T}} V_0=0, \\\tr_{\mathcal{H}_\mathrm{B}}(\kettbbra{V_0}{\Delta})=V_0^{\textup{T}}\Delta=0
,\quad\quad \tr_{\mathcal{H}_{\mathrm{B}}}(\kettbbra{\Delta}{\Delta})=\Delta^\textup{T}\Delta \leq I_\mathrm{A}. 
\end{gathered}
\end{equation}
For $i=n$, we have that $\ketbra{\gamma_n^n}{\gamma_n^n}=\kettbbra{\Delta}{\Delta}^{\otimes n}$ is a probabilistic $n$-comb because  $\tr_{\mathcal{H}_\mathrm{B}}(\kettbbra{\Delta}{\Delta})\leq I_{\mathrm{A}}$. Similarly, for $i=0$, we have that $\ketbra{\gamma_0^n}{\gamma_0^n}=\kettbbra{V_0}{V_0}^{\otimes n}$ is a probabilistic $n$-comb because  $\tr_{\mathcal{H}_\mathrm{B}}(\kettbbra{V_0}{V_0})\leq I_{\mathrm{A}}$.
Thus the hypothesis holds automatically for the case $n=1$.

Let $n\ge 2$ and suppose that $\ketbra{\gamma_i^{n-1}}{\gamma_i^{n-1}}$ is a probabilistic $(n-1)$-comb for all $i\le n-1$.
Now we prove $\ketbra{\gamma_i^{n}}{\gamma_i^n}$ is a probabilistic $n$-comb for any $i\leq n$.

Let $0< i < n$ (since the cases $i=0$ and $i=n$ are already proved).
We can express $\ket{\gamma_i^n}$ as follows
\[\ket{\gamma_i^n}=\sqrt{\frac{\binom{n-1}{i}}{\binom{n}{i}}}\kett{V_0}\otimes \ket{\gamma_i^{n-1}}+\sqrt{\frac{\binom{n-1}{i-1}}{\binom{n}{i}}}\kett{\Delta}\otimes\ket{\gamma_{i-1}^{n-1}}.\]
Hence, 
\begin{align}
\tr_{\mathcal{H}_{\mathrm{B},n}}(\ketbra{\gamma_i^n}{\gamma_i^n})&= \frac{\binom{n-1}{i}}{\binom{n}{i}}\tr_{\mathcal{H}_{\mathrm{B}}}(\kettbbra{V_0}{V_0})\otimes \ketbra{\gamma_{i}^{n-1}}{\gamma_i^{n-1}}+\frac{\binom{n-1}{i-1}}{\binom{n}{i}}\tr_{\mathcal{H}_{\mathrm{B}}}(\kettbbra{\Delta}{\Delta})\otimes \ketbra{\gamma_{i-1}^{n-1}}{\gamma_{i-1}^{n-1}}\nonumber\\
&\,\,\,\,+ \frac{\sqrt{\binom{n-1}{i}\binom{n-1}{i-1}}}{\binom{n}{i}}\Big(\tr_{\mathcal{H}_{\mathrm{B}}}(\kettbbra{V_0}{\Delta})\otimes \ketbra{\gamma_{i}^{n-1}}{\gamma_{i-1}^{n-1}}+\tr_{\mathcal{H}_{\mathrm{B}}}(\kettbbra{\Delta}{V_0})\otimes \ketbra{\gamma_{i-1}^{n-1}}{\gamma_{i}^{n-1}}\Big)\nonumber\\
&\leq\frac{\binom{n-1}{i}}{\binom{n}{i}}I_\mathrm{A}\otimes \ketbra{\gamma_i^{n-1}}{\gamma_i^{n-1}}+\frac{\binom{n-1}{i-1}}{\binom{n}{i}}I_{\mathrm{A}}\otimes \ketbra{\gamma_{i-1}^{n-1}}{\gamma_{i-1}^{n-1}},\nonumber
\end{align}
where we use  \cref{eq-1120121} in the last inequality.
By induction hypothesis, we have that $\ketbra{\gamma_i^{n-1}}{\gamma_i^{n-1}}$ and $\ketbra{\gamma_{i-1}^{n-1}}{\gamma_{i-1}^{n-1}}$ are probabilistic $(n-1)$-combs. 
Since $\binom{n-1}{i}+\binom{n-1}{i-1}=\binom{n}{i}$, we deduce that
\[\frac{\binom{n-1}{i}}{\binom{n}{i}}\ketbra{\gamma_i^{n-1}}{\gamma_i^{n-1}}+\frac{\binom{n-1}{i-1}}{\binom{n}{i}}\ketbra{\gamma_{i-1}^{n-1}}{\gamma_{i-1}^{n-1}}\]
is a probabilistic $(n-1)$-comb.
Therefore, $\ketbra{\gamma_i^n}{\gamma_{i}^n}$ is a probabilistic $n$-comb due to \cref{prop-4290313}.

Since $U$ is applied only on $\mathcal{H}_\mathrm{B}$, it is easy to see that $U^{\otimes n}\ketbra{\gamma_i}{\gamma_i}U^{\dag\otimes n}$ is also a probabilistic $n$-comb. Finally, $\Gamma_i$ is also a probabilistic $n$-comb because it is a convex combination of probabilistic $n$-combs.
\end{proof}

\begin{lemma}\label{lemma-12281107}
Recall that $d_{\textup{min}}=\min\{d,D-d\}$. Suppose $n\leq \frac{1}{2e^4}(D-d)d_{\textup{min}}/\varepsilon^2$. There exists positive numbers $\lambda_0,\ldots,\lambda_n$ such that for any $V\in \mathcal{N}$
\begin{align}
\sum_{i=0}^n\lambda_i&\leq 4(D-d)^2d_{\textup{min}}^2\exp\!\left(\sqrt{8n\varepsilon^2(D-d)d_{\textup{min}}}\right), \nonumber \\
    \kettbbra{V}{V}^{\otimes n}&\leq \sum_{i=0}^n \lambda_i \Gamma_i. \label{eq-12271213}
\end{align}
\end{lemma}
\begin{proof}
We remark that, because of  \cref{eq-12271523},  $\kett{V}^{\otimes n}$ is contained in $\bigoplus_{j=0}^n \supp(\Gamma_j)$. 
Moreover, by \cref{fact-5122103}, to establish \cref{eq-12271213}, it is sufficient to show that
\begin{equation}\label{eq-12271634}
\sum_{i=0}^n \frac{1}{\lambda_i} \tr\!\left(\Gamma_i^{-1}\kettbbra{V}{V}^{\otimes n}\right)\leq 1.
\end{equation}
To this end, we bound the left hand side of \cref{eq-12271634} as
\begin{align}
\sum_{i=0}^n \frac{1}{\lambda_i} \tr\!\left(\Gamma_i^{-1}\kettbbra{V}{V}^{\otimes n}\right)&=\sum_{i=0}^n \frac{1}{\lambda_i}\binom{n}{i} \varepsilon^{2i} \tr\!\left(\Gamma_i^{-1} U^{\otimes n}\ketbra{\gamma_i}{\gamma_i} U^{\dag\otimes n}\right) \label{eq-12271624} \\
&=\sum_{i=0}^n \frac{1}{\lambda_i}\binom{n}{i} \varepsilon^{2i} \tr\!\left(\Gamma_i^{-1}\ketbra{\gamma_i}{\gamma_i} \right) \label{eq-12271632}\\
&\leq \sum_{i=0}^n \frac{1}{\lambda_i}\binom{n}{i} \varepsilon^{2i} \binom{(D-d)d'+i-1}{i}\label{eq-12271633}\\
&\leq \sum_{i=0}^n \frac{1}{\lambda_i}\binom{n}{i} \varepsilon^{2i} \binom{(D-d)d_{\textup{min}}+i-1}{i}. \label{eq-3252109}
\end{align}
Here \cref{eq-12271624} follows from \cref{eq-12271523} and the fact that $U^{\otimes n}\ket{\gamma_i}\in \supp(\Gamma_i)$, \cref{eq-12271632} uses the fact that $\Gamma_i$ commutes with $U^{\otimes n}$, and  
\cref{eq-12271633} is an application of \cref{lemma-12271509} where $\ket{\gamma_i}$ is viewed as a vector in the space
\begin{equation}\label{eq-1301405}
\spanspace\left(\left\{\sum_{\substack{S\subseteq [n]\\ |S|=i }}\ket{\psi}^{\otimes S}\otimes \kett{V_0}^{\otimes [n]\setminus S}\,\, \bigg|\,\, \ket{\psi}\in \Pi\right\}\right),
\end{equation}
where 
\[\Pi=\spanspace\{\ket{d+1}_\mathrm{B},\ldots,\ket{D}_\mathrm{B}\}\otimes\spanspace\{\ket{1}_\mathrm{A},\ldots,\ket{d'}_\mathrm{A}\}.\]
The space defined in \cref{eq-1301405} is invariant  under the action $U^{\otimes n}$ for any $U\in\mathbb{U}_{D-d}$ since $U$ fixes $\kett{V_0}$ and $\Pi$ is an invariant space under $U$ (recall that $U$ acts as $I_d\oplus U$ on $\mathcal{H}_\mathrm{B}$). By \cref{lemma-12271712}, it  has dimension $\binom{(D-d)d'+i-1}{i}$.
Therefore, we obtain  \cref{eq-12271633}.
\cref{eq-3252109} uses  $d'\leq \min\{d,D-d\}=d_{\textup{min}}$.

Our goal is look for  positive numbers $\lambda_0,\ldots,\lambda_n$ such that $\sum_{i=0}^n\lambda_i$ is small and \cref{eq-3252109} is upper bounded by $1$.
For ease of notation, we denote $Q=(D-d)d_{\textup{min}}$. 
 \cref{eq-3252109} can be upper bounded using \cref{lemma-12272204} as follows
\begin{align}
\eqref{eq-3252109} &\leq \sum_{i=0}^n\frac{1}{\lambda_i} \exp\!\left(n H\!\left(\frac{i}{n}\right) + 2i \ln(\varepsilon) + (Q+i) H\!\left(\frac{i}{Q+i}\right)\right) \nonumber\\
&=\sum_{i=0}^n\frac{1}{\lambda_i}\exp\!\left(-i \ln\!\left(\frac{i}{n\varepsilon^2}\right) -(n-i)\ln\!\left(\frac{n-i}{n}\right)+ i \ln \!\left(1+\frac{Q}{i}\right)+Q \ln\!\left(1+\frac{i}{Q}\right)\right)\nonumber \\
&\leq \sum_{i=0}^n \frac{1}{\lambda_i} \exp\!\left(-i\ln\!\left(\frac{i}{n\varepsilon^2}\right)+i\ln\!\left(1+\frac{Q}{i}\right)+2i\right),\label{eq-12280023}
\end{align}
where the last inequality uses 
\[(n-i)\ln\!\left(\frac{n}{n-i}\right)\leq (n-i)\left(\frac{n}{n-i}-1\right)= i \quad \textup{and} \quad Q\ln\!\left(1+\frac{i}{Q}\right)\leq Q\frac{i}{Q}=i. \]
To further bound \cref{eq-12280023}, we consider two cases depending on the value of $i$: 
\begin{itemize}
\item When $i< Q$, we have
\begin{align}
-i\ln\!\left(\frac{i}{n\varepsilon^2}\right)+i\ln\!\left(1+\frac{Q}{i}\right)+2i &\leq -i\ln\!\left(\frac{i}{n\varepsilon^2}\right)+i\ln\!\left(\frac{2Q}{i}\right)+2i \nonumber\\
&=2i\ln\!\left(\frac{\sqrt{2e^2n\varepsilon^2Q}}{i}\right) \nonumber \\
&\leq \sqrt{8n\varepsilon^2Q}, \label{eq-12280237}
\end{align}
where the last inequality follows from \cref{lamma-12281043}. 

\item When $i\geq Q$, since we assumed that $n\leq \frac{1}{2e^4}(D-d)d_{\textup{min}}/\varepsilon^2=\frac{1}{2e^4}Q/\varepsilon^2$  we have $i\geq 2e^4n\varepsilon^2$, thus
\begin{align}
-i\ln\!\left(\frac{i}{n\varepsilon^2}\right)+i\ln\!\left(1+\frac{Q}{i}\right)+2i&\leq -i \ln(2e^4)+i\ln(2)+2i=-2i.\nonumber
\end{align}
\end{itemize}
Therefore, the choice 
\begin{align*}
    \lambda_i = \begin{cases}
        3Q\exp\!\left(\sqrt{8n\varepsilon^2Q}\right), &\quad\textup{if $i< Q$} \\
        \exp(-i), &\quad\textup{if $i \ge Q$,}
    \end{cases}
\end{align*}
ensures that  \cref{eq-12280023} can be  upper bounded by
\begin{align}
\eqref{eq-12280023} &\leq \sum_{i<Q} \frac{1}{\lambda_i} \exp\!\left(\sqrt{8n\varepsilon^2Q}\right)+\sum_{i\geq Q} \frac{1}{\lambda_i}\exp(-2i) \nonumber \\
&\leq \frac{1}{3}+\sum_{i\geq Q} \exp(-i)\nonumber 
\leq\frac{1}{3}+ \exp(-Q) \frac{e}{e-1}\nonumber 
< 1,\nonumber
\end{align}
where  we use that $Q=(D-d)d_{\textup{min}}\geq 1$. 
Finally, we have that
\begin{align}
\sum_{i=0}^n\lambda_i &\leq 3Q^2\exp\!\left(\sqrt{8n\varepsilon^2Q}\right)+\exp(-Q)\frac{e}{e-1} \nonumber \\
& < 3Q^2\exp\!\left(\sqrt{8n\varepsilon^2Q}\right)+\frac{1}{2}\nonumber
 < 4Q^2\exp\!\left(\sqrt{8n\varepsilon^2Q}\right)\nonumber
\\&= 4(D-d)^2d_{\textup{min}}^2\exp\!\left(\sqrt{8n\varepsilon^2(D-d)d_{\textup{min}}}\right), \nonumber
\end{align}
which concludes the proof.
\end{proof}

\section{Packing nets of quantum channels}\label{sec-3290355}

\begin{table}[ht]
\centering
\setlength{\tabcolsep}{1.5mm}{{\renewcommand{\arraystretch}{1.5}
\begin{tabular}{|c|c|c|c|}
\hline
Type of hardness & Assumption  & Distance & Logarithm of cardinality \\ \hline
\shortstack{\\ \\ Type I \\ \\ \cref{sec:hard-instance-d1=rd2}} &  \shortstack{$d_1\leq rd_2\leq \frac{4}{3}d_1$\\ \\ \,} & \shortstack{Choi: $\Omega(\varepsilon)$\\ \\ \,}  &        \shortstack{$\Omega(d_1^2)$\\ \\ \,}  \\ \hline
\multirow{2}{*}{\shortstack{\\ Type II\\ \\ \cref{sec-3230332}}} & \multirow{2}{*}{\shortstack{\\ $rd_2>d_1$\\ \\ $rd_2<d_1+r$}} &  Choi: $\Omega(\sr^{3/2}\varepsilon)$   &  $\Omega(\sr^2 (rd_2-d_1)^2)$  \\ \cline{3-4} 
 &      & Diamond: $\Omega(\varepsilon)$ &  $\Omega((rd_2-d_1)^2)$   \\ \hline
\multirow{2}{*}{\shortstack{\\ Type II \\ \\ \cref{sec-3230333}}} & \multirow{2}{*}{\shortstack{\\ $rd_2>d_1$\\ \\ $d_1+r\leq rd_2$, $r\leq d_1$}} & Choi: $\Omega(\sr^{3/2}\varepsilon)$  &  $\Omega(\sr^3 d_1(rd_2-d_1))$  \\ \cline{3-4} 
  &       & Diamond: $\Omega(\varepsilon)$  &  $\Omega(\sr d_1(rd_2-d_1))$ \\ \hline
\multirow{2}{*}{\shortstack{\\ Type II \\ \\ \cref{sec-3260100}}}  & \multirow{2}{*}{\shortstack{\\ $rd_2>d_1$\\ \\ $d_1+r\leq rd_2$, $d_1<r$}} & \multirow{2}{*}{Choi: $\Omega(\varepsilon)$} &  \multirow{2}{*}{$\Omega(rd_1d_2)$} \\
& & & \\ \hline
\end{tabular}
}}
\caption{Construction of packing nets. Here, $\sr\coloneqq\min\{(rd_2-d_1)/d_1,1\}=\min\{\tau-1,1\}$.
Note that for the non-boundary case ($rd_2>d_1$), $\sr$ can be close to $0$ while $\sr d_1=\min\{rd_2-d_1,d_1\}\geq 1$.}\label{tab-3290336}
\end{table}

In this section, we construct packing nets of quantum channels whose Stinespring dilation isometries correspond to the type I and type II hard instances defined in \cref{sec-3200450}.
The packing nets have different distances and cardinality depending on the parameter regimes. 
We summarize the results in \cref{tab-3290336}.

Let $d_1,d_2$ and $r$ denote the input dimension, output dimension and upper bound of Kraus rank of the quantum channels. 
In this section, we further assume 
\begin{equation}\label{eq-3200605}
r\leq d_1d_2/2,
\end{equation}
instead of the more general $r\leq d_1d_2$. 
Note that adding this constraint here (in the packing net construction) still allows us to derive general lower bounds for quantum channel tomography, since a larger Kraus rank only makes the tomography task more difficult.
Then, the following notation will be used in this section.
\begin{notation}
Let $\mathcal{H}_\mathrm{A}\cong \mathbb{C}^{d_1}$, $\mathcal{H}_{\mathrm{B}}\cong\mathbb{C}^{d_2}$ and $\mathcal{H}_\mathrm{anc}\cong\mathbb{C}^r$ be the input, output and ancilla systems. 
Let $\{\ket{1}_\mathrm{A},\ldots,\ket{d_1}_\mathrm{A}\}$ be an orthonormal basis of $\mathcal{H}_\mathrm{A}$, and similarly let $\{\ket{1}_\mathrm{B},\ldots,\ket{d_2}_\mathrm{B}\}$ and $\{\ket{1}_\mathrm{anc},\ldots,\ket{r}_\mathrm{anc}\}$ be orthonormal bases of $\mathcal{H}_\mathrm{B}$ and $\mathcal{H}_{\mathrm{anc}}$, respectively.
Moreover, for $\mathrm{X}\in\{\mathrm{A},\mathrm{B},\mathrm{anc}\}$, we use $\mathcal{H}_\mathrm{X}[i:j]$ to denote the subspace of a Hilbert space $\mathcal{H}_\mathrm{X}$ spanned by $\{\ket{i}_X,\ldots,\ket{j}_X\}$, and we use $\mathcal{H}[i]$ to denote $\mathcal{H}[i:i]$.
\end{notation}

\subsection{Type I instance: $d_1\leq rd_2 \leq \frac{4}{3} d_1 $}\label{sec:hard-instance-d1=rd2}

\subsubsection{Construction}
In this subsection, we will use the definition in \cref{sec:hard-d1=rd2}, where the parameter $d, D$ in \cref{sec:hard-d1=rd2} correspond to $d_1$ and $rd_2$ here.

We define $d_1'=d_1$ when $d_1$ is even and $d_1'=d_1-1$ when $d_1$ is odd. 
Let $g:[d_1]\rightarrow [d_2]\times [r]$ be the function $g(i)=(\lfloor \frac{i-1}{r}\rfloor+1,i-r\lfloor\frac{i-1}{r}\rfloor)$, i.e., $g$ maps the integers in $[d_1]$ to $[d_2]\times [r]$ in row-major order.
Let $\mathcal{H}_{\mathrm{B}'}$ be a $rd_2$-dimensional space with an orthonormal basis $\{\ket{1}_{\mathrm{B}'},\ldots,\ket{d_2r}_{\mathrm{B}'}\}$.
We identify $\mathcal{H}_\mathrm{B'}$ with $\mathcal{H}_\mathrm{B}\otimes\mathcal{H}_\mathrm{anc}$ by the map $\ket{i}_\mathrm{B'}\mapsto \ket{g(i)_1}_\mathrm{B}\otimes\ket{g(i)_2}_\mathrm{anc}$.
In other words, the basis vectors are identified as follows:
\[\ket{1}_\mathrm{B'},\, \ket{2}_{\mathrm{B'}},\, \ldots,\, \ket{d_2r}_\mathrm{B'} \longleftrightarrow \ket{1}_\mathrm{B}\otimes \ket{1}_\mathrm{anc},\,\ket{1}_\mathrm{B}\otimes \ket{2}_\mathrm{anc},\,\ldots,\ket{d_2}_\mathrm{B}\otimes \ket{r}_\mathrm{anc}.\]
We also identify $\mathcal{H}_\mathrm{A}$ with the subspace $\mathcal{H}_{\mathrm{B'}}[1:d_1]$ of $\mathcal{H}_\mathrm{B'}$ using the map $\ket{i}_\mathrm{A}\mapsto \ket{i}_\mathrm{B'}$ for $i\in[d_1]$. In other words, we have the following correspondence:
\begin{equation}\label{eq-3261637}
\ket{1}_\mathrm{A},\,\ket{2}_\mathrm{A},\,\ldots,\,\ket{d_1}_\mathrm{A}\longleftrightarrow \ket{1}_\mathrm{B'},\,\ket{2}_\mathrm{B'},\,\ldots,\, \ket{d_1}_\mathrm{B'}.
\end{equation}
Then, define the linear operator $V_{0}:\mathcal{H}_{\mathrm{A}}\rightarrow \mathcal{H}_{\mathrm{B}}\otimes\mathcal{H}_\mathrm{anc}$ as
\[V_0\coloneqq
\begin{cases}
\sqrt{1-\varepsilon^2}\sum_{i=1}^{d'}\ket{i}_{\mathrm{B}'}\bra{i}_\mathrm{A},&\quad \textup{if $d_1$ is even}\\
\sqrt{1-\varepsilon^2}\sum_{i=1}^{d'}\ket{i}_{\mathrm{B}'}\bra{i}_\mathrm{A}+\ket{d}_{\mathrm{B}'}\bra{d}_\mathrm{A},&\quad \textup{if $d_1$ is odd}.\\
\end{cases}
\]
Define the linear operator $\Delta:\mathcal{H}_{\mathrm{A}}\rightarrow\mathcal{H}_{\mathrm{B}'}$ as
\begin{equation}\label{eq-3212346}
\Delta\coloneqq \mathrm{i}\left(\sum_{i=1}^{\lfloor d_1/2\rfloor} \ket{i}_{\mathrm{B}'}\bra{i}_{\mathrm{A}}-\sum_{i= \lfloor d_1/2\rfloor+1}^{d_1'} \ket{i}_{\mathrm{B}'}\bra{i}_{\mathrm{A}}\right),
\end{equation}
where $\mathrm{i}$ is the imaginary unit.
Note that under the identification in \cref{eq-3261637}, we can simply discard the subscript $\mathrm{A}$ and $\mathrm{B'}$ in the definition of $V_0$ and $\Delta$.

Suppose $U\in\mathbb{U}_{d_1'}$ is a unitary acting on $\spanspace\{\ket{1},\ldots,\ket{d_1'}\}$, and $U$ fixes $\ket{d_1}$ if $d_1>d_1'$. 
Define the isometry $V_{\varepsilon,U}:\mathcal{H}_{\mathrm{A}}\rightarrow\mathcal{H}_{\mathrm{B}'}$ as
\begin{equation}\label{eq-3211722}
\begin{split}
V_{\varepsilon,U}& \coloneqq U\left(V_0 +\varepsilon \Delta\right) U^\dag = V_0+\varepsilon U\Delta U^\dag,
\end{split}
\end{equation}
where with a little abuse of notation, $U$ can act either on $\mathcal{H}_\mathrm{A}$ or $\mathcal{H}_\mathrm{B'}$ (using the identification in \cref{eq-3261637}, and we assume that $U$ acts trivially on $\spanspace\{\ket{d_1+1}_\mathrm{B'},\ldots,\ket{d_2r}_\mathrm{B'}\}$), which should not cause confusion given the context. 
For clarity, we illustrate our construction in \cref{fig-3290108}.

\begin{figure}[t]
    \centering
    \includegraphics[width=0.75\linewidth]{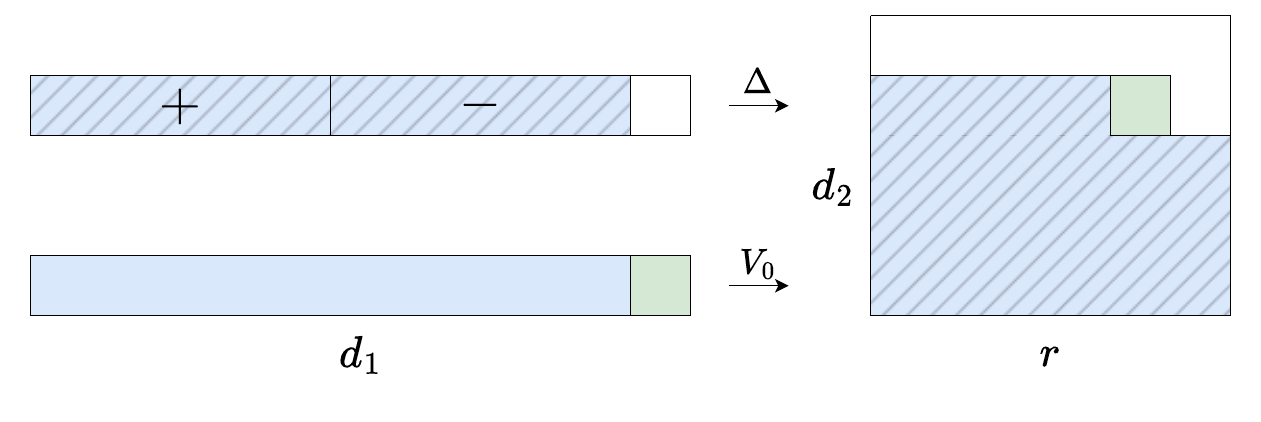}
    \caption{Illustration of our construction. We define linear operators $V_0$ and $\Delta$ from the $d_1$-dimensional space $\mathcal{H}_\mathrm{A}$ to the $d_2r$-dimensional space $\mathcal{H}_\mathrm{B}\otimes\mathcal{H}_\mathrm{anc}$. The Haar randomness is applied on the hatched area.}
    \label{fig-3290108}
\end{figure}

\subsubsection{Existence}
    \begin{theorem}\label{thm:existence-d1=rd2}
    Suppose $\varepsilon\leq 1/39660$.
    There exists a finite subset $\mathcal{N}$ of $\{V_{\varepsilon,U}\,|\, U\in\mathbb{U}_{d_1'}\}$ for $V_{\varepsilon,U}$ defined in \cref{eq-3211722} with cardinality $|\mathcal{N}|\geq \exp(d_1^2/750197760001)$, such that for any $V_x\neq V_y\in \mathcal{N}$, if we set $\mathcal{E}_x=\tr_{\mathcal{H}_\mathrm{anc}}(V_x(\cdot) V_x^\dag)$ and $\mathcal{E}_y=\tr_{\mathcal{H}_\mathrm{anc}}(V_y(\cdot) V_y^\dag)$, then
\begin{align}
     \frac{1}{d_1}\left\|C_{\mathcal{E}_{x}} - C_{\mathcal{E}_y}\right\|_{1}   \ge \frac{\eps}{1983000}.\label{eq-3230249}
\end{align}
\end{theorem}
\begin{proof}
First, we need the following lemma.
\begin{lemma}\label{lem:ppts-Phi-eq}
Suppose $\varepsilon\leq 1/39660$.
Given any unitary $U\in \mathbb{U}_{d_1'}$, let $\mathcal{E}_{U}= \tr_{\mathcal{H}_{\mathrm{anc}}}(V_{\varepsilon,U} (\cdot) V_{\varepsilon,U}^\dagger)$ where $V_{\varepsilon,U}$ is the isometry defined in \cref{eq-3211722}. Then, the function 
 \begin{equation*}
     f(U_x, U_y)= \frac{1}{d_1} \left\|C_{\mathcal{E}_{U_x}} - C_{\mathcal{E}_{U_y}} \right\|_1
 \end{equation*}
is $\varepsilon\sqrt{\frac{32}{d_1}}$-Lipschitz with respect to the $\ell_2$-sum of the $2$-norms (Frobenius norm), namely
\begin{equation*}
    |f(U_x,U_y)-f(U'_x,U'_y)|\leq  \varepsilon\sqrt{\frac{32}{d_1}}  \|(U_x,U_y) -(U_x',U_y')\|_F
\end{equation*}
for all $U_x,U_x',U_y,U_y'\in \mathbb{U}_{d'}$, where $\|(A,B)\|_F\coloneqq\sqrt{\|A\|_F^2+\|B\|_F^2}$ and $\|\cdot\|_F$ denotes the Frobenius norm. 
Furthermore, if we consider independent random unitaries $U_x,U_y \sim \mathbb{U}_{d_1'}$, we have
\begin{equation*}
\E\!\left[f(U_x,U_y)\right]\ge  \frac{\varepsilon}{19830}.
\end{equation*}
\end{lemma}
The proof of \cref{lem:ppts-Phi-eq} is deferred to \cref{app:first-moment-and-lip-eq}.

We now have all the ingredients to prove \cref{thm:existence-d1=rd2}.
Since the function $f(U_x,U_y)$ is 
$\varepsilon\sqrt{\frac{32}{d_1}}$-Lipschitz, when we sample two independent unitaries $U_x,U_u \sim \mathbb{U}_{d_1'}$, by Theorem \ref{lem:corollary17}, we have\footnote{To be precise, we are applying Theorem \ref{lem:corollary17} to $-f$.}
\begin{align}
    \mathbb{P}\!\left[f(U_i,U_j) \leq \frac{\varepsilon}{1983000} \right] \le  \exp\left(-\frac{d_1'\varepsilon^2}{12\cdot 22100^2}\cdot \frac{d_1}{32\varepsilon^2}\right)  \le  \exp\!\left(- \frac{d_1^2}{768\cdot 22100^2}\right).\nonumber
\end{align}
Then, we independently sample $\exp(d_1^2/(1536\cdot 22100^2+1))=\exp(d_1^2/750197760001)$ Haar random unitaries in $\mathbb{U}_{d_1'}$ and the union bound shows that there exists a non-zero probability that for any pair $U_x, U_y$, we have $f(U_x,U_y)\geq \varepsilon/1983000$. 
Therefore, there exists a subset $\mathcal{M}\subseteq\mathbb{U}_{d_1'}$ with cardinality $|\mathcal{M}|\geq \exp(d_1^2/750197760001)$ such that for any pair $U_x,U_y\in\mathcal{M}$, we have $f(U_x,U_y)\geq \varepsilon/1983000$. 
Then, the set $\{V_{\varepsilon,U}| U\in\mathcal{M}\}$ is as desired.
\end{proof}

\subsubsection{Lipschitz continuity and expectation: proof of Lemma \ref{lem:ppts-Phi-eq}} \label{app:first-moment-and-lip-eq}

\paragraph{Proof that $f$ is Lipschitz.}
Recall that   $V_{\varepsilon,U}  = V_0+\varepsilon U\Delta U^\dag$. 
We have that by the triangle inequality:
\begin{align}
    &|f(U_x, U_y) - f(U_x', U_y')| \nonumber
    \\
    = &\frac{1}{d_1}\left|\Big\| \tr_{\mathcal{H}_{\mathrm{anc}}}({\kettbbra{V_{\eps,U_x}}{ V_{\eps,U_x}}} - \kettbbra{V_{\eps,U_y}}{V_{\eps,U_y}})\Big\|_1-\left\| \tr_{\mathcal{H}_{\mathrm{anc}}}({\kettbbra{V_{\eps,U_x'}}{ V_{\eps,U_x'}}} - \kettbbra{V_{\eps,U_y'}}{V_{\eps,U_y'}})\right\|_1\right| \nonumber\\
    \leq & \frac{1}{d_1}\left\| \tr_{\mathcal{H}_{\mathrm{anc}}}({\kettbbra{V_{\eps,U_x}}{ V_{\eps,U_x}}} - \kettbbra{V_{\eps,U_x'}}{V_{\eps,U_x'}})-\tr_{\mathcal{H}_{\mathrm{anc}}}({\kettbbra{V_{\eps,U_y}}{ V_{\eps,U_y}}} - \kettbbra{V_{\eps,U_y'}}{V_{\eps,U_y'}})\right\|_1 \nonumber \\
    \le & \frac{1}{d_1}\Big\| \tr_{\mathcal{H}_{\mathrm{anc}}}({\kettbbra{V_{\eps,U_x}}{ V_{\eps,U_x}}} - \kettbbra{V_{\eps,U_x'}}{V_{\eps,U_x'}}) \Big\|_1 + \frac{1}{d_1}\left\| \tr_{\mathcal{H}_{\mathrm{anc}}}({\kettbbra{V_{\eps,U_y}}{ V_{\eps,U_y}}} - \kettbbra{V_{\eps,U_y'}}{V_{\eps,U_y'}})\right\|_1. \label{eq-3212339}
\end{align}
Note that
\begin{align}
    &\Big\| \tr_{\mathcal{H}_{\mathrm{anc}}}\Big(\big(\kett{V_{\eps,U_x}}-\kett{V_{\eps,U_x'}}\big)\bbra{V_{\eps,U_x}}\Big) \Big\|_{1} \leq \eps \left\|\left(\kett{ U_x \Delta U_x^\dagger} - \kett{U_x' \Delta U_x'^\dagger}\right) \bbra{V_{\varepsilon,U_x}} \right\|_{1} \nonumber \\
    =&\varepsilon\sqrt{d_1} \left\|\kett{ U_x \Delta U_x^\dagger} - \kett{U_x' \Delta U_x'^\dagger}\right\| 
   \le \eps \sqrt{d_1}\left(\left\|\kett{( U_x-U_x') \Delta U_x^\dagger }  \right\| + \left\|\kett{U_x' \Delta( U_x-U_x')^\dagger}  \right\| \right) \nonumber \\
   \leq &\varepsilon \sqrt{d_1} \sqrt{\tr\!\left((U_x-U'_x)^\dag (U_x-U_x')\right)}+\varepsilon \sqrt{d_1} \sqrt{\tr\!\left((U_x-U'_x)^\dag (U_x-U_x')\right)} \label{eq-3212318}  \\
= & 2\varepsilon \sqrt{d_1} \|{U}_x-{U}_x'\|_F,\nonumber
\end{align}
where \cref{eq-3212318} is by using $\Delta^\dag \Delta\leq I$.
Hence by the triangle inequality 
\begin{align}
    &\Big\| \tr_{\mathcal{H}_{\mathrm{anc}}}({\kettbbra{V_{\eps,U_x}}{ V_{\eps,U_x}}} - \kettbbra{V_{\eps,U_x'}}{V_{\eps,U_x'}}) \Big\|_1 \nonumber \\
\le & \Big\| \tr_{\mathcal{H}_{\mathrm{anc}}}\Big(\big(\kett{V_{\eps,U_x}}-\kett{V_{\eps,U_x'}}\big)\bbra{V_{\eps,U_x}}\Big) \Big\|_{1} +  \Big\| \tr_{\mathcal{H}_{\mathrm{anc}}}\Big(\kett{V_{\eps,U_x'}}\big(\bbra{V_{\eps,U_x}}-\bbra{V_{\varepsilon,U'_x}}\big)\Big) \Big\|_{1} \nonumber \\
\le & 4\varepsilon \sqrt{d_1} \|{U}_x-{U}_x'\|_F.\label{eq-3212340}
\end{align}
Therefore, combining \cref{eq-3212339} and \cref{eq-3212340}, we have
\begin{align}
    |f(U_x, U_y) - f(U_x', U_y')|& \leq 4\varepsilon\sqrt{\frac{1}{d_1}}\|U_x - U_x'\|_F + 4\varepsilon\sqrt{\frac{1}{d_1}}\|U_y-U_y'\|_F \nonumber \\
    &\le  4 \varepsilon\sqrt{\frac{2}{d_1}} \sqrt{\|{U}_x-{U}_x'\|_F^2+ \|{U}_y-{U}_y'\|_F^2}. \nonumber
\end{align}

\paragraph{Proof of the lower bound on the the expectation of $f$.}

Recall the definition of the $d_1\times d_1$ matrix $\Delta$ in \cref{eq-3212346}.
This construction ensures the following properties, which will be used later:
\begin{itemize}
    \item $\tr(\Delta)=0$,
    \item $\|\Delta\|_{\rm op}=1$,
    \item $\tr(\Delta \Delta^\dagger \Delta\Delta^\dagger)=\tr(\Delta \Delta^\dagger)= d_1'$. 
\end{itemize}
Note that, for unitaries $U_x,U_y\in\mathbb{U}_{d_1'}$, $f(U_x,U_y)$ equals
\begin{align}
    \frac{1}{d_1}\left\|C_{\mathcal{E}_{U_x}} -C_{\mathcal{E}_{U_y}}\right\|_{1}&\geq \frac{\eps}{d_1} \left\|
    \tr_{\mathcal{H}_{\mathrm{anc}}}\!\left(\kettbbra{U_x\Delta U_x^\dagger}{V_0}  + \kettbbra{V_0}{U_x\Delta U_x^\dagger} -\kettbbra{U_y\Delta U_y^\dagger}{V_0}  - \kettbbra{V_0}{U_y\Delta U_y^\dagger}\right)\right\|_1 \nonumber \\
    &\qquad\qquad -\frac{\eps^2}{d_1}\left\|
    \tr_{\mathcal{H}_{\mathrm{anc}}}\!\left(\kettbbra{U_x\Delta U_x^\dagger}{U_x\Delta U_x^\dagger} -\kettbbra{U_y\Delta U_y^\dagger}{U_y\Delta U_y^\dagger}\right)\right\|_1 \nonumber \\
    &\geq  \frac{\eps}{d_1} \| D(U_x,U_y)\|_1 -2\eps^2, \label{eq-3230212}
\end{align}
where we define
\begin{equation*}
D(U_x,U_y)\coloneqq A_x + A_x^\dagger -A_y -A_y^\dagger,\quad A_x \coloneqq \tr_{\mathcal{H}_{\mathrm{anc}}}\big(\kett{U_x\Delta U_x^\dagger}\bbra{V_0}\big),\quad A_y \coloneqq \tr_{\mathcal{H}_{\mathrm{anc}}}\big(\kett{U_y\Delta U_y^\dagger}\bbra{V_0}\big),
\end{equation*} and used that $ \left\|
    \tr_{\mathcal{H}_{\mathrm{anc}}}\!\left(\kettbbra{U_x\Delta U_x^\dagger}{U_x\Delta U_x^\dagger}\right)\right\|_1 =\tr  \left(
    \kettbbra{U_x\Delta U_x^\dagger}{U_x\Delta U_x^\dagger}\right) = \tr\!\left(\Delta \Delta^\dagger\right)\le d_1$.

Now, we are interested in proving the bound
\begin{align}\label{eq:second-D}
    \E\!\left[\tr\!\left(|D(U_x,U_y)|^2\right)\right] \geq  \frac{d_1^2}{20r}, 
\end{align}
where the expectation is over the Haar random $U_x,U_y\sim \mathbb{U}_{d_1'}$.
    We have that  
    \begin{align}
    \E\!\left[\tr\!\left(|D(U_x,U_y)|^2\right)\right] = 2\E\!\left[\tr(A_x^2)+ 2\tr(A_x^\dagger A_x) + \tr(A_x^{\dagger 2}) \right] -  2\Re \E\!\left[\tr(A_xA_y)+ \tr(A_x^\dagger A_y)\right].\label{eq-3220326}
    \end{align}
    Since $\tr(\Delta) = 0$, we know that $\E\!\left[U_x\Delta U_x^\dagger\right] = \tr(\Delta) \frac{I}{d_1} = 0$ by Schur's lemma. This means 
    \begin{equation}\label{eq-3220324}
         \E\!\left[\tr(A_xA_y)\right] =\E\!\left[\tr(A_x^\dagger A_y)\right] = 0.
    \end{equation}
   Moreover,  we have that 
    \begin{align}
    \E\!\left[\tr(A_xA_x^\dagger)\right] &= \E\!\left[\tr\!\left(\tr_{\mathcal{H}_{\mathrm{anc}}}\!\left(\kettbbra{U_x\Delta U_x^\dagger}{V_0}\right)
    \tr_{\mathcal{H}_{\mathrm{anc}}}\!\left(\kettbbra{V_0}{U_x\Delta U_x^\dagger}\right)\right)\right] \nonumber \\
    &=\sum_{i,j=1}^r  \E\!\left[\tr\!\left(\kettbbra{U_x\Delta U_x^\dagger}{V_0} \cdot \ketbra{i}{j}_{\mathrm{anc}}\cdot  \kettbbra{V_0}{U_x\Delta U_x^\dagger}\cdot \ketbra{j}{i}_{\mathrm{anc}}\right)\right]    \nonumber \\
    &=\sum_{i,j=1}^r  \bbra{V_0} \cdot  \ketbra{i}{j}_{\mathrm{anc}} \cdot \kett{V_0} \cdot \E\!\left[\tr\!\left(\kettbbra{U_x\Delta U_x^\dagger}{ U_x\Delta U_x^\dagger}\cdot \ketbra{j}{i}_{\mathrm{anc}} \right)\right] \nonumber \\
    &= \sum_{i=1}^r  \bbra{V_0} \cdot  \ketbra{i}{i}_{\mathrm{anc}} \cdot \kett{V_0} \cdot \E\!\left[\tr\!\left(\kettbbra{U_x\Delta U_x^\dagger}{U_x\Delta U_x^\dagger}\cdot \ketbra{i}{i}_{\mathrm{anc}} \right)\right] \nonumber\\
    &\geq(1-\eps^2)\left\lfloor\frac{d_1}{r}\right\rfloor \cdot \sum_{i=1}^r \E\!\left[\tr\!\left(
    \kettbbra{U_x\Delta U_x^\dagger}{U_x\Delta^\dagger U_x^\dagger}\cdot \ketbra{i}{i}_{\mathrm{anc}}\right)\right] \label{eq-3261544} \\
    &=(1-\eps^2)\left\lfloor\frac{d_1}{r}\right\rfloor\cdot \E\!\left[\tr\!\left(\kettbbra{U_x\Delta U_x^\dagger}{U_x\Delta U_x^\dagger} \right)\right] \nonumber \\
    &= (1-\eps^2)\left\lfloor\frac{d_1}{r}\right\rfloor\cdot \E\!\left[\tr\!\left(\Delta^\dag \Delta  \right)\right] \nonumber \\
    &=(1-\eps^2)\left\lfloor\frac{d_1}{r}\right\rfloor d_1',\label{eq-3220323}
    \end{align}
where in \cref{eq-3261544} the inequality we used that $\bbra{V_0} \cdot  \ketbra{i}{i}_{\mathrm{anc}} \cdot \kett{V_0} \ge (1-\eps^2) \sum_{j=1}^{d_1}\mathbbm{1}_{g_2(j)=i} \geq (1-\eps^2)\lfloor\frac{d_1}{r}\rfloor$. 
We also have  
\begin{align}
    &\E\!\left[\tr(A_x^2)\right]\nonumber \\ 
     =& \E\!\left[\tr\!\left(\tr_{\mathcal{H}_\mathrm{anc}}\!\left(\kettbbra{U_x\Delta U_x^\dagger}{V_0}\right) \tr_{\mathcal{H}_\mathrm{anc}}\!\left(\kettbbra{U_x\Delta U_x^\dagger}{V_0}\right)\right)\right] \nonumber \\
    =& \sum_{i,j=1}^r  \E\!\left[\tr\!\left(\kettbbra{U_x\Delta U_x^\dagger}{V_0}  \cdot \ketbra{i}{j}_{{\mathrm{anc}}} \cdot \kettbbra{U_x\Delta U_x^\dagger}{V_0}  \cdot \ketbra{j}{i}_{{\mathrm{anc}}} \right)\right] \nonumber \\
    =& \sum_{i,j=1}^r \E\!\left[\tr\!\left(V_0^\dagger \cdot \ketbra{i}{j}_{{\mathrm{anc}}}\cdot   U_x\Delta U_x^\dagger \right)\tr\!\left(V_0^\dagger\cdot \ketbra{j}{i}_{\mathrm{anc}}  \cdot U_x\Delta U_x^\dagger\right)\right] \nonumber \\
    =& \sum_{i,j=1}^r \sum_{a,b=1}^{d_1}\E\!\left[\tr\left(V_0^\dagger\cdot\ketbra{i}{j}_{\mathrm{anc}} \cdot U_x\Delta U_x^\dagger\cdot \ketbra{a}{b}\cdot V_0^\dagger \cdot \ketbra{j}{i}_{\mathrm{anc}} \cdot U_x\Delta U_x^\dagger \cdot \ketbra{b}{a}\right)\right] \nonumber \\
    =&\sum_{i,j=1}^r \sum_{a,b=1}^{d_1}\E\!\left[\tr\left( U_x\Delta U_x^\dagger P\cdot \ketbra{a}{b} \cdot V_0^\dagger  \cdot \ketbra{j}{i}_{\mathrm{anc}}\cdot P  U_x\Delta U_x^\dagger P \cdot \ketbra{b}{a} \cdot V_0^\dagger\cdot\ketbra{i}{j}_{\mathrm{anc}}\cdot P\right)\right] \label{eq-3221728}\\
    =&  \sum_{i,j=1}^r \sum_{a,b=1}^{d_1'} \bigg[\frac{1}{d_1'^2-1}\left( \tr(\Delta^2) \tr(P\ketbra{a}{b}V_0^\dagger\ketbra{j}{i}_\mathrm{anc} P) \tr(P\ketbra{b}{a}V_0^\dagger\ketbra{i}{j}_\mathrm{anc} P) \right) \nonumber \\
    &\qquad\qquad\qquad\qquad - \frac{1}{d_1'(d_1'^2-1)}\left( \tr(\Delta^2) \tr\big(P\ketbra{a}{b}V_0^\dagger\ketbra{j}{i}_{\mathrm{anc}} P \ketbra{b}{a}V_0^\dagger \ketbra{i}{j}_\mathrm{anc}P\big) \right) \bigg]\label{eq-3221729} \\
    =& \sum_{i,j=1}^r \sum_{a,b=1}^{d_1'}\left[ -\frac{d_1'(1-\varepsilon^2)}{d_1'^2-1} \big|\tr( \ketbra{a}{b}\cdot \ketbra{j}{i}_\mathrm{anc})\big|^2+\frac{1-\varepsilon^2}{d_1'^2-1}\tr\big(\ketbra{a}{b}\cdot \ketbra{j}{i}_\mathrm{anc}\cdot \ketbra{b}{a}\cdot \ketbra{i}{j}_\mathrm{anc}\big)\right] \nonumber \\
    \geq &-\frac{d_1'(1-\varepsilon^2)}{d_1'^2-1} \cdot \Big(r^2\cdot \left\lceil\frac{d_1'}{r}\right\rceil\Big)+\frac{1-\varepsilon^2}{d_1'^2-1}\cdot\Big(r\cdot \left\lfloor\frac{d_1'}{r}\right\rfloor\Big)\label{eq-3222206} \\
    \ge & - \frac{(1-\varepsilon^2)r^2d_1'}{d_1'^2-1}\left\lceil\frac{d_1'}{r}\right\rceil, \label{eq-3222222}
    \end{align}
where in \cref{eq-3221728} we introduce the projector  $P\coloneqq \sum_{i=1}^{d_1'}\ketbra{i}{i}$ and note that $P U_x\Delta U_x^\dag P= U_x\Delta U_x^\dag$, and \cref{eq-3221729} is due to \cref{coro-3221651} and the fact that $\tr(\Delta)=0$, \cref{eq-3222206} is because $\sum_{a,b=1}^{d_1'}|\tr(\ketbra{a}{b}\cdot\ketbra{j}{i}_\mathrm{anc})|^2$ equals either $\lfloor \frac{d_1'}{r}\rfloor$ or $\lceil \frac{d_1'}{r}\rceil$ depending on $i,j$ and thus is no more than $\lceil\frac{d_1'}{r}\rceil$;
and $\sum_{a,b=1}^{d_1'}\tr(\ketbra{a}{b}\ketbra{j}{i}_\mathrm{anc}\ketbra{b}{a}\ketbra{i}{j}_\mathrm{anc})$ is non-zero only when $i=j$ and equals either $\lfloor \frac{d_1'}{r}\rfloor$ or $\lceil \frac{d_1'}{r}\rceil$ depending on $i$ (which equals $j$) and thus is no less than $\lfloor\frac{d_1'}{r}\rfloor$, and \cref{eq-3222222} is obtained by simply discarding the positive term.
Therefore, combining \cref{eq-3220326} with \cref{eq-3220324}, \cref{eq-3220323} and \cref{eq-3222222}, we have
\begin{align}
\E\!\left[\tr\!\left(|D(U_x,U_y)|^2\right)\right] &\geq 4(1-\varepsilon^2)\left(d_1'\left\lfloor \frac{d_1}{r}\right\rfloor-\frac{r^2d_1'}{d_1'^2-1}\left\lceil\frac{d_1'}{r}\right\rceil\right)= 4(1-\varepsilon^2)d_1'\left\lfloor \frac{d_1}{r}\right\rfloor \left(1-\frac{r^2}{d_1'^2-1}\cdot\frac{\lceil d_1'/r\rceil}{\lfloor d_1/r\rfloor}\right) \nonumber\\
&\geq4(1-\varepsilon^2)d_1'\left\lfloor \frac{d_1}{r}\right\rfloor \cdot \frac{1}{10} \label{eq-3262003}\\
&\geq (1-\varepsilon^2)\frac{d_1'd_1}{5r}\geq \frac{d_1'd_1}{10r}\geq \frac{d_1^2}{20r}, \nonumber
\end{align}
where in \cref{eq-3262003} we used $\lceil d_1'/r \rceil \leq \lceil d_1/r\rceil \leq 2 \lfloor d_1/r\rfloor$ and
\[2r^2\leq \frac{8}{9}d_1^2\leq \frac{9}{10}(d_1'^2-1),\]
since $2r\leq rd_2\leq \frac{4}{3}\cdot d_1$ and $d_1\geq 162$ by assumption. Therefore, \cref{eq:second-D} is proved.

Finally, let us prove that
    \begin{align}\label{eq:fourth-D}
    \E\!\left[\tr(|D(U_x,U_y)|^4) \right] \le 12288 \cdot \frac{d_1^4}{r^3}.
    \end{align}
    Recall that $D(U_x,U_y) = A_x + A_x^\dagger -A_y -A_y^\dagger$ so by the triangle inequality and Hölder inequality:
    \begin{align}
     \E\!\left[\tr(|D(U_x,U_y)|^4)\right] &= \E\!\left[\|A_x + A_x^\dagger -A_y -A_y^\dagger\|_4^4\right]
     \le \E\!\left[ \left(\|A_x\|_4 + \|A_x^\dagger\|_4 + \|A_y\|_4 +\|A_y^\dagger\|_4\right)^4\right] \nonumber \\
     &\le 4^3 \cdot \E\!\left[ \|A_x\|_4^4 + \|A_x^\dagger\|_4^4 + \|A_y\|_4^4 +\|A_y^\dagger\|_4^4\right]\nonumber  \\
     &= 4^4\cdot \ex{ \|A_x\|_4^4},\label{eq-3230159}
    \end{align}
where $\|\cdot\|_4$ denotes the Schatten $4$-norm.
   Moreover,
\begin{align}
&\E\!\left[\|A_x\|_4^4\right] = \E\!\left[\tr(A_x A_x^\dagger A_x A_x^\dagger )\right] \nonumber \\
= & \E\!\left[\tr\!\left(\tr_{\mathcal{H}_\mathrm{anc}}\!\left(\kettbbra{U_x\Delta U_x^\dagger}{V_0}\right) \tr_{\mathcal{H}_\mathrm{anc}}\!\left(\kettbbra{V_0}{U_x \Delta  U_x^\dagger} \right) \tr_{\mathcal{H}_\mathrm{anc}}\!\left(\kettbbra{U_x\Delta U_x^\dagger}{V_0}\right) \tr_{\mathcal{H}_\mathrm{anc}}\!\left(\kettbbra{V_0}{U_x \Delta  U_x^\dagger} \right)\right)\right] \nonumber \\
=&\sum_{i,j,k,l=1}^r  \mathbb{E}\bigg[\tr\!\Big(\kettbbra{U_x\Delta U_x^\dagger}{V_0} \cdot \ketbra{i}{j}_\mathrm{anc}\cdot \kettbbra{V_0}{U_x\Delta U_x^\dagger } \cdot \ketbra{j}{k}_\mathrm{anc} \nonumber \\[-1em]
&\qquad\qquad\qquad\qquad\qquad \cdot  \kettbbra{U_x\Delta U_x^\dagger}{V_0} \cdot \ketbra{k}{l}_\mathrm{anc} \cdot \kettbbra{ V_0}{U_x\Delta U_x^\dagger }\cdot \ketbra{l}{i}_\mathrm{anc}\Big)\bigg] \nonumber  \\
\le & \sum_{i,k=1}^r  \left\lceil\frac{d_1}{r}\right\rceil^2\cdot  \E\!\left[\tr\!\Big(\kettbbra{U_x\Delta U_x^\dagger}{U_x\Delta U_x^\dagger}  \cdot \ketbra{i}{k}_\mathrm{anc} \cdot \kettbbra{U_x\Delta U_x^\dagger}{U_x\Delta U_x^\dagger} \cdot  \ketbra{k}{i}_\mathrm{anc} \Big)\right] \nonumber \\
=&\sum_{i,k=1}^r  \left\lceil\frac{d_1}{r}\right\rceil^2\cdot  \E\!\left[\tr\!\left(U_x\Delta^\dagger U_x^\dagger \cdot \ketbra{i}{k}_\mathrm{anc}\cdot U_x\Delta U_x^\dagger \right)\tr\!\left(U_x\Delta^\dagger U_x^\dagger  \cdot \ketbra{k}{i}_\mathrm{anc} \cdot U_x\Delta U_x^\dagger\right)\right] \nonumber \\
=&\sum_{i,k=1}^r \left\lceil\frac{d_1}{r}\right\rceil^2\cdot  \sum_{a,b=1}^{rd_2} \E\!\left[\tr\!\left(U_x\Delta\Delta^\dagger U_x^\dagger  \cdot \ketbra{i}{k}_\mathrm{anc}\cdot \ketbra{a}{b} \cdot U_x\Delta\Delta^\dagger U_x^\dagger  \cdot \ketbra{k}{i}_\mathrm{anc}\cdot      \ketbra{b}{a}\right)\right] \nonumber \\
=&\sum_{i,k=1}^r \left\lceil\frac{d_1}{r}\right\rceil^2\cdot  \sum_{a,b=1}^{d_1'} \E\!\left[\tr\!\left(U_x\Delta\Delta^\dagger U_x^\dagger P  \cdot \ketbra{i}{k}_\mathrm{anc}\cdot \ketbra{a}{b} \cdot P U_x\Delta\Delta^\dagger U_x^\dagger P \cdot \ketbra{k}{i}_\mathrm{anc}\cdot      \ketbra{b}{a} \cdot P\right)\right] \label{eq-3230111}\\
=&\!\sum_{i,k=1}^r \left\lceil\frac{d_1}{r}\right\rceil^2\cdot \sum_{a,b=1}^{d_1'}\! \bigg( \frac{1}{d_1'^2\!-\!1}\!\Big[ \tr(\Delta\Delta^\dagger)^2 \!\cdot\! \bra{a}\ketbra{i}{k}_\mathrm{anc}\ket{a} \!\cdot\! \bra{b}\ketbra{k}{i}_\mathrm{anc}\ket{b} + \tr(\Delta\Delta^\dagger\Delta\Delta^\dagger) \!\cdot\! |\bra{b}\ketbra{i}{k}_\mathrm{anc}\ket{a}|^2\Big]  \nonumber \\[-0.5em]
&\,\, -\frac{1}{d_1'(d_1'^2-1)}\Big[ \tr(\Delta\Delta^\dagger)^2 \!\cdot\! |\bra{b}\ketbra{i}{k}_\mathrm{anc}\ket{a}|^2 + \tr(\Delta\Delta^\dagger\Delta\Delta^\dagger)\! \cdot\! \bra{a}\ketbra{i}{k}_\mathrm{anc}\ket{a}\!\cdot\! \bra{b}\ketbra{k}{i}_\mathrm{anc}\ket{b} \Big]\bigg) \label{eq-3230108} \\
\leq &\sum_{i,k=1}^r  \left\lceil\frac{d_1}{r}\right\rceil^2 \cdot \frac{1}{d_1'^2-1} \cdot \bigg(d_1'^2\cdot\mathbbm{1}_{i=k}\cdot  \left\lceil \frac{d_1'}{r}\right\rceil^2 + d_1'\cdot  \left\lceil \frac{d_1'}{r}\right\rceil\bigg) \label{eq-3230113}\\
=& \left\lceil\frac{d_1}{r}\right\rceil^2\cdot \frac{1}{d_1'^2-1}\cdot \left(d_1'^2 r \left\lceil \frac{d_1'}{r}\right\rceil^2 + r^2 d_1'\left\lceil \frac{d_1'}{r}\right\rceil \right)\leq \frac{48d_1^4}{r^3},\label{eq-3230139} 
\end{align}
where in \cref{eq-3230111} we recall the projector $P=\sum_{i=1}^{d'}\ketbra{i}{i}$ and note that $PU_x\Delta\Delta^\dag U_x^\dag P=U_x\Delta\Delta^\dag U_x^\dag$, \cref{eq-3230108} is due to \cref{coro-3221651}, \cref{eq-3230113} is obtained by directly discarding the negative terms, and in \cref{eq-3230139} we used that $d_1'\leq d_1$, $\lceil d_1/r\rceil\leq 2d_1/r$, $d_1'^2-1\geq d_1^2/2$.
Therefore, combining \cref{eq-3230159} with \cref{eq-3230139}, \cref{eq:fourth-D} is proved.

By H\"{o}lder's inequality, we have
\begin{align}
    \E\!\left[\tr\!\big(|D(U_x,U_y)|\big)\right]^{2/3} \E\!\left[\tr\big(|D(U_x,U_y)|^4\big)\right]^{1/3} \ge \E\!\left[\tr\!\big(|D(U_x,U_y)|^2\big)\right].\nonumber
\end{align}
Then, using \cref{eq:second-D} and \cref{eq:fourth-D}, we know that
\[\E\!\left[\tr\!\big(|D(U_x,U_y)|\big)\right]\geq \sqrt{\frac{d_1^6}{8000\cdot r^3}\cdot \frac{r^3}{12288\cdot d_1^4}}\geq  \frac{d_1}{9915}.\]
Finally, using \cref{eq-3230212}, we conclude that
\begin{align}
   \frac{1}{d_1}\E\!\left[\left\|C_{\mathcal{E}_{U_x}} -C_{\mathcal{E}_{U_y}}\right\|_{1}\right]&\ge \frac{\eps}{d_1} \E\!\big[\| D(U_x,U_y)\|_1\big] -2\eps^2 \ge  \frac{\varepsilon}{9915}-2\varepsilon^2\geq \frac{\varepsilon}{19830}, \nonumber
\end{align}
where we used $\varepsilon\leq 1/39660$.

\subsection{Type II instance: $d_1< rd_2$ with $rd_2<d_1+r$}\label{sec-3230332}

\subsubsection{Construction}\label{sec-3232114}
In this subsection, we will use the definition in \cref{sec-1110250}, where the parameter $d, D$ in \cref{sec-1110250} correspond to $d_1$ and $rd_2$ here.

\begin{figure}[ht]
    \centering
    \includegraphics[width=0.7\linewidth]{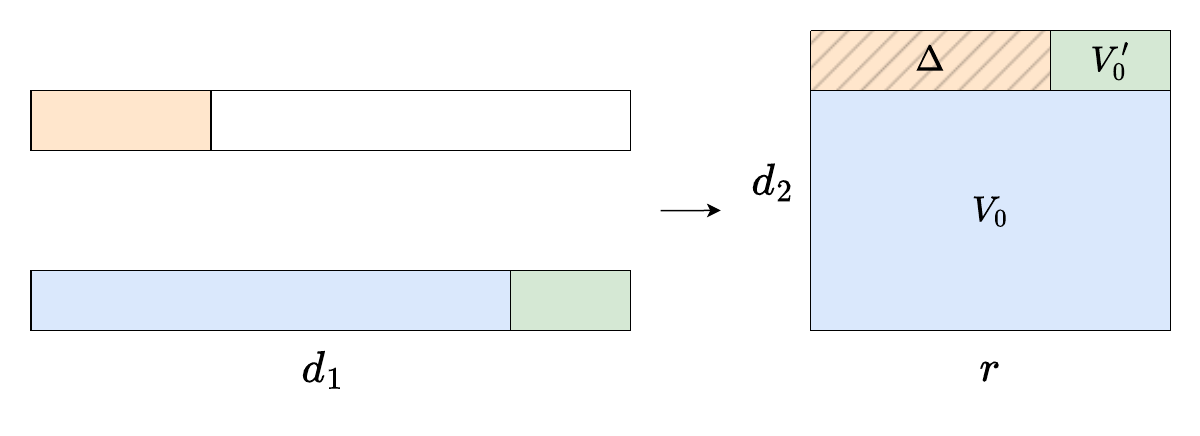}
    \caption{Illustration of our construction. We define linear operators $V_0, V_0'$ and $\Delta$ from the $d_1$-dimensional space $\mathcal{H}_\mathrm{A}$ to the $d_2r$-dimensional space $\mathcal{H}_\mathrm{B}\otimes\mathcal{H}_\mathrm{anc}$. The Haar randomness is applied on the hatched area.}
    \label{fig-3290202}
\end{figure}

Suppose $d_2>1$ and $rd_2< d_1+r$. 
Let $\sr \coloneqq \min\{(rd_2-d_1)/d_1,1\}>0$ and $\eta\coloneqq d_1 -r\lfloor \tfrac{d_1}{r}\rfloor$.
One can see that:
\begin{itemize}
    \item $d_1+r>rd_2\geq 2r$ and thus $d_1>r$, which means $rd_2<d_1+r<2d_1$. Thus $\sr = (rd_2-d_1)/d_1 < 1$.
    \item $d_2-1< d_1/r< d_2$, thus $\lfloor \tfrac{d_1}{r}\rfloor=d_2-1\geq 1$ and $1\leq \eta\leq r-1$. 
    \item $\sr d_1=rd_2-d_1\leq r$.
\end{itemize}
Let $\varepsilon\in(0,1/2)$ and $\Sigma_\varepsilon\coloneqq \sqrt{1-\varepsilon^2}\sum_{i=1}^{r-\eta}\ketbra{i}{i}_\mathrm{A}+\sum_{i=r-\eta+1}^{r(d_2-1)}\ketbra{i}{i}_\mathrm{A}$ be a diagonal matrix.
Let $g:[d_1]\rightarrow [d_2]\times [r]$ be the function $g(i)=(\lfloor \frac{i-1}{r}\rfloor+1,i-r\lfloor\frac{i-1}{r}\rfloor)$, i.e., $g$ maps the integers in $[d_1]$ to $[d_2]\times [r]$ in row-major order. Then, we define the linear operator $V_0:\mathcal{H}_\mathrm{A}[1:r(d_2-1)]\rightarrow\mathcal{H}_\mathrm{B}[1:d_2-1]\otimes \mathcal{H}_\mathrm{anc}$ as 
\begin{equation}\label{eq-3180233}
V_0\coloneqq \sum_{i=1}^{r(d_2-1)} \Big(\ket{g(i)_1}_\mathrm{B}\otimes \ket{g(i)_2}_\mathrm{anc}\bra{i}_\mathrm{A}\Big)\cdot \Sigma_\varepsilon.
\end{equation}
One can easily check that $V_0$ can be written as
\begin{equation}\label{eq-1131347}
V_0= \sum_{i=1}^r \ket{i}_\mathrm{anc}\otimes K_i,
\end{equation}
for $K_i:\mathcal{H}_{\mathrm{A}}[1:r(d_2-1)]\rightarrow\mathcal{H}_\mathrm{B}[1:d_2-1]$ satisfying
\begin{equation}\label{eq-3240522}
\sum_{i=1}^r K_i^\dag K_i=\Sigma_\varepsilon^2 \quad \textup{and}\quad \frac{d_2-1}{2}\cdot\mathbbm{1}_{i=j} \leq \left|\tr\!\left(K_i^\dag K_j\right)\right|\leq   (d_2-1)\cdot\mathbbm{1}_{i=j}.
\end{equation}
Define isometry $V_0':\mathcal{H}_\mathrm{A}[r(d_2-1)+1:d_1]\rightarrow\mathcal{H}_{\mathrm{B}}[d_2]\otimes \mathcal{H}_{\mathrm{anc}}[r-\eta+1:r]$ as
\[V_0'\coloneqq \ket{d_2}_\mathrm{B}\otimes \sum_{i=1}^{\eta}\ket{i+r-\eta}_{\mathrm{anc}}\bra{i+r(d_2-1)}_\mathrm{A}.\]
Define isometry $\Delta:\mathcal{H}_{\mathrm{A}}[1:r-\eta]\rightarrow \mathcal{H}_{\mathrm{B}}[d_2]\otimes \mathcal{H}_\mathrm{anc}[1:r-\eta]$ as
\[\Delta\coloneqq \ket{d_2}_\mathrm{B}\otimes \sum_{i=1}^{r-\eta} \ket{i}_\mathrm{anc}\bra{i}_\mathrm{A}\]
Then, for $U\in\mathbb{U}_{r-\eta}$, we define the isometry $V_{\varepsilon,U}:\mathcal{H}_\mathrm{A}\rightarrow\mathcal{H}_\mathrm{B}\otimes\mathcal{H}_{\mathrm{anc}}$ as
\begin{equation}\label{eq-1140331}
V_{\varepsilon,U}\coloneqq V_0+V_0'+\varepsilon U\Delta,
\end{equation}
where $U$ acts on $\mathcal{H}_{\mathrm{B}}[d_2]\otimes\mathcal{H}_\mathrm{anc}[1:r-\eta]$.
We can verify that $V_{\varepsilon,U}$ is indeed an isometry from \cref{eq-3180233} and the definition of $U\Delta$.
Furthermore, we note that the image of $U\Delta$ is orthogonal to the image of $V_0+V_0'$. 
For clarity, we also illustrate our construction in \cref{fig-3290202}.

\subsubsection{Existence for Choi trace norm}
Then, we prove that there exists a large set of isometries $V_{\varepsilon,U}$ with good separation property.
\begin{theorem}\label{thm-1140332}
Recall that $\sr=(rd_2-d_1)/d_1$ in this regime.
Suppose $\varepsilon\leq \sr^{3/2}/9200$.
There exists a finite subset $\mathcal{N}$ of $\{V_{\varepsilon,U}\,|\, U\in\mathbb{U}_{r-\eta}\}$ for $V_{\varepsilon,U}$ defined in \cref{eq-1140331} with cardinality $|\mathcal{N}|\geq \exp(\sr^2 (rd_2-d_1)^2/480001)$, such that for any $V_x\neq V_y\in\mathcal{N}$, if we set $\mathcal{E}_x=\tr_{\mathcal{H}_\mathrm{anc}}(V_x(\cdot)V_x^\dag)$  and $\mathcal{E}_y=\tr_{\mathcal{H}_\mathrm{anc}}(V_y(\cdot)V_y^\dag)$, then
\[\frac{1}{d_1}\|C_{\mathcal{E}_x}-C_{\mathcal{E}_y}\|_1\geq \frac{\sr^{3/2}}{4600}\varepsilon.\]
\end{theorem}
\begin{proof}
First, we need the following lemma.
\begin{lemma}\label{lemma-1131330}
    There exists a finite subset $\mathcal{M}\subseteq \mathbb{U}_{r-\eta}$ with cardinality $|\mathcal{M}|\geq \exp(\sr^2(rd_2-d_1)^2/480001)$ such that for any $U_x\neq U_y\in\mathcal{M}$, 
\begin{equation}\label{eq-1130328}
\frac{1}{d_1}\left\|\tr_{\mathcal{H}_\mathrm{anc}}\!\left((\kett{V_0}+\kett{V_0'})\Big(\bbra{U_x\Delta}-\bbra{U_y\Delta}\Big)\right)\right\|_1\geq \frac{\sr^{3/2}}{4600}.
\end{equation}
\end{lemma}
The proof of \cref{lemma-1131330} is deferred.

Now, we are able to prove \cref{thm-1140332}.
Let $\mathcal{M}$ be the set given in \cref{lemma-1131330}.
For any $U_x\neq U_y\in\mathcal{M}$, if we set $\mathcal{E}_x=\tr_{\mathcal{H}_\mathrm{anc}}(V_{\varepsilon,U_x}(\cdot)V_{\varepsilon,U_x}^\dag)$  and $\mathcal{E}_y=\tr_{\mathcal{H}_\mathrm{anc}}(V_{\varepsilon,U_y}(\cdot)V_{\varepsilon,U_y}^\dag)$, then we have
\begin{align}
&\|C_{\mathcal{E}_x}-C_{\mathcal{E}_y}\|_1 \nonumber \\
=&\Big\|\tr_{\mathcal{H}_{\mathrm{anc}}}\!\left(\kettbbra{V_{\varepsilon,U_x}}{V_{\varepsilon,U_x}}\right)-\tr_{\mathcal{H}_{\mathrm{anc}}}\!\left(\kettbbra{V_{\varepsilon,U_y}}{V_{\varepsilon,U_y}}\right)\Big\|_1 \nonumber\\
=&\bigg\|\varepsilon^2\Big(\tr_{\mathcal{H}_{\mathrm{anc}}}\!\left(\kettbbra{U_x\Delta}{U_x\Delta}\right)-\tr_{\mathcal{H}_{\mathrm{anc}}}\!\left(\kettbbra{U_y\Delta}{U_y\Delta}\right)\Big) \nonumber\\
&\quad\quad\quad +\varepsilon \tr_{\mathcal{H}_{\mathrm{anc}}}\!\Big(\kettbbra{V_0+V_0'}{(U_x-U_y)\Delta}\Big) + \varepsilon\tr_{\mathcal{H}_{\mathrm{anc}}}\!\Big(\kettbbra{(U_x-U_y)\Delta}{V_0+V_0'}\Big) \bigg\|_1 \nonumber \\
\geq & \varepsilon\bigg\|\tr_{\mathcal{H}_{\mathrm{anc}}}\!\Big(\kettbbra{V_0+V_0'}{(U_x-U_y)\Delta}\Big) + \tr_{\mathcal{H}_{\mathrm{anc}}}\!\Big(\kettbbra{(U_x-U_y)\Delta}{V_0+V_0'}\Big)\bigg\|_1 - 2 \varepsilon^2 d_1 \label{eq-1130259} \\
= & 2\varepsilon \bigg\|\tr_{\mathcal{H}_{\mathrm{anc}}}\!\Big(\kettbbra{V_0+V_0'}{(U_x-U_y)\Delta}\Big)\bigg\|_1-2\varepsilon^2 d_1 \label{eq-1130304}\\
\geq &\frac{2\sr^{3/2}}{4600}\varepsilon d_1 -2\varepsilon^2d_1 \label{eq-1130327} \\
\geq &\frac{\sr^{3/2}}{4600}\varepsilon d_1. \label{eq-1130422}
\end{align}
In \cref{eq-1130259} we used 
\[\|\tr_{\mathcal{H}_{\mathrm{anc}}}(\kettbbra{U\Delta}{U\Delta})\|_1 = \tr(\kettbbra{U\Delta}{U\Delta})=r-\eta\leq r< d_1. \]
In \cref{eq-1130304} we used the fact that $\tr_{\mathcal{H}_{\mathrm{anc}}}\!\Big(\kettbbra{V_0+V_0'}{(U_x-U_y)\Delta}\Big)$ is a linear operator with support in  $\mathcal{\mathcal{H}_\mathrm{A}}[1:r-\eta]\otimes\mathcal{H}_{\mathrm{B}}[d_2]$ and image in 
\[\big(\mathcal{H}_\mathrm{A}[1:r(d_2-1)]\otimes \mathcal{H}_{\mathrm{B}}[1:d_2-1]\big)\oplus \big(\mathcal{H}_\mathrm{A}[r(d_2-1)+1:d_1]\otimes \mathcal{H}_\mathrm{B}[d_2]\big),\] 
which is orthogonal to its support (noting that $r-\eta\leq r(d_2-1)$); and then we used \cref{fact-3191602}.
In \cref{eq-1130327} we used \cref{eq-1130328}.
In \cref{eq-1130422} we used that $\varepsilon\leq \sr^{3/2}/9200$.
Therefore, we can lower bound the Choi trace norm 
\[\frac{1}{d_1}\left\|C_{\mathcal{E}_{x}}-C_{\mathcal{E}_{y}}\right\|_1\geq \frac{\sr^{3/2}}{4600}\varepsilon.\]
Thus, the set $\mathcal{N}=\{V_{\varepsilon,U} \,|\, U\in\mathcal{M}\}$ is the desired set.
\end{proof}

Now, we give the proof of \cref{lemma-1131330}.
\begin{proof}[Proof of \cref{lemma-1131330}]
Let $\hat{V}_0\coloneqq V_0+V_0'$.
Then, we need the following lemma.
\begin{lemma}\label{lemma-1140127}
For $U_x,U_y\in\mathbb{U}_{r-\eta}$, let us define 
\[F(U_x,U_y)=\frac{1}{d_1}\tr_{\mathcal{H}_\mathrm{anc}}\!\left(\kett{\hat{V}_0}\Big(\bbra{U_x\Delta}-\bbra{U_y\Delta}\Big)\right),\]
then the function $f(U_x,U_y)=\|F(U_x,U_y)\|_1=\tr(|F(U_x,U_y)|)$ is $\sqrt{\frac{2}{d_1}}$-Lipschitz with respect to the $\ell_2$-sum of the $2$-norms (Frobenius norm). 
Furthermore, for independent random $U_x,U_y\sim \mathbb{U}_{r-\eta}$, we have $\E\!\left[\tr(|F(U_x,U_y)|^2)\right]\geq\frac{\sr}{2r}$, and $\E\!\left[\tr(|F(U_x,U_y)|^4)\right]\leq \frac{256}{r^3}$.
\end{lemma}
The proof of \cref{lemma-1140127} is deferred to \cref{sec-3211806}.

By the H\"older's inequality we have
\[\E\!\left[\tr\!\left(|F(U_x,U_y)|^2\right)\right]\leq \E\!\left[\tr\!\left(|F(U_x,U_y)|^4\right)\right]^{1/3}\E\!\left[\tr\!\left(|F(U_x,U_y)|\right)\right]^{2/3},\]
which, combined with \cref{lemma-1140127}, implies
\[\E\!\left[\tr\!\left(|F(U_x,U_y)|\right)\right]^2\geq \frac{\sr^3}{2048}.\]
Thus $\E\!\left[\tr\!\left(|F(U_x,U_y)|\right)\right]> \sr^{3/2}/46$.
Then, since the function $f(U_x,U_y)=\tr(|F(U_x,U_y)|)$ is $\sqrt{\frac{2}{d_1}}$-Lipschitz, we can use \cref{lem:corollary17} to prove the concentration result:
\[\Pr\!\left[\tr\!\left(|F(U_x,U_y)|\right)\leq \frac{\sr^{3/2}}{4600}\right]\leq \exp\!\left(-\frac{d_1(r-\eta)}{2}\cdot \frac{\sr^3}{10000\cdot 12}\right)\leq \exp\!\left(-\frac{\sr^2 (rd_2-d_1)^2}{240000}\right),\]
where we used $r-\eta\geq \sr d_1=rd_2-d_1$.
Then, we independently sample $\exp(\sr^2(rd_2-d_1)^2/480001)$ Haar random unitaries in $\mathbb{U}_{r-\eta}$ and the union bound shows that there exists a non-zero probability that for any pair $U_x, U_y$, we have $\tr(|F(U_x,U_y)|)\geq \sr^{3/2}/4600$.
Thus, there exists a set with cardinality $\geq \exp(\sr^2(rd_2-d_1)^2/480001)$ such that \cref{eq-1130328} holds.
\end{proof}

\subsubsection{Second and fourth moments: proof of Lemma \ref{lemma-1140127}}\label{sec-3211806}
The Lipschitz continuity can be seen from \cref{lemma-3192212}.

Note that $\hat{V}_0$ can also be written as
\[\hat{V}_0=\sum_{i=1}^r \ket{i}_{\mathrm{anc}}\otimes \big(K_i\oplus \ket{d_2}_\mathrm{B}\bra{z_i}_\mathrm{A}\big),\]
where $K_i$ are defined in \cref{eq-1131347}, and we set $\ket{z_i}_\mathrm{A}=\ket{i-r+\eta+r(d_2-1)}_\mathrm{A}$ for $r-\eta+1\leq i\leq r$ and $\ket{z_i}_\mathrm{A}=0$ for $1\leq i\leq r-\eta$.
Let us define 
\[K_i'=K_i\oplus \ket{d_2}_\mathrm{B}\bra{z_i}_\mathrm{A},\]
then $K_i'$ satisfy
\begin{equation}\label{eq-1132146}
\left|\tr\!\left(K_i'^\dag K'_j\right)\right|=\left| \tr\!\left(K_i^\dag K_j\right)+\mathbbm{1}_{r-\eta+1\leq i=j\leq r}\right|\leq d_2\cdot\mathbbm{1}_{i=j},
\end{equation}
and
\begin{equation}\label{eq-3182343}
\tr\!\left(K_i'^\dag K_i'\right)=\tr\!\left(K_i^\dag K_i\right)+\mathbbm{1}_{r-\eta+1\leq i\leq r}\geq \tr\!\left(K_i^\dag K_i\right)\geq \frac{d_2-1}{2},
\end{equation}
due to \cref{eq-3240522}.

Define $K_{x,i}=\bra{i}_\mathrm{anc} U_x\Delta$, $K_{y,i}=\bra{i}_\mathrm{anc} U_y\Delta$. This means 
\[F(U_x,U_y)=\frac{1}{d_1}\sum_{i=1}^r\kett{K_i'}\Big(\bbra{K_{x,i}}-\bbra{K_{y,i}}\Big).\]
Then, we note that
\begin{align}
\E\!\left[\tr(|F(U_x,U_y)|^2)\right]&=\frac{1}{d_1^2}\E\!\left[\tr\!\left(\sum_{i,j=1}^r \kett{K_i'}\Big(\bbra{K_{x,i}}-\bbra{K_{y,i}}\Big)\Big(\kett{K_{x,j}}-\kett{K_{y,j}}\Big)\bbra{K_j'}\right)\right] \nonumber\\
&=\frac{1}{d_1^2}\E\!\left[\sum_{i=1}^r \bbrakett{K_i'}{K_i'}\Big(\bbra{K_{x,i}}-\bbra{K_{y,i}}\Big)\Big(\kett{K_{x,i}}-\kett{K_{y,i}}\Big)\right]\nonumber \\
&=\frac{2}{d_1^2}\sum_{i=1}^{r-\eta}\bbrakett{K_i'}{K_i'}  \label{eq-1131627} \\
&\geq \frac{2}{d_1^2} (r-\eta) \frac{d_2-1}{2} \label{eq-3182341}\\
&\geq \frac{\sr(d_2-1)}{d_1}\geq \frac{\sr}{2r} \label{eq-3180253} 
\end{align}
where \cref{eq-1131627} is because for $t_1,t_2\in\{x,y\}$, we have
\begin{align}
\E\!\left[\bbrakett{K_{t_1,i}}{K_{t_2,i}}\right]&=\E\!\left[\tr\!\left(K_{t_1,i}^\dag K_{t_2,i}\right)\right] =\tr\!\left(\Delta^\dag \E\!\left[U_{t_1}^\dag \ketbra{i}{i}_\mathrm{anc}U_{t_2} \right]\Delta\right) \nonumber\\
&=\mathbbm{1}_{t_1=t_2}\mathbbm{1}_{1\leq i\leq r-\eta}\cdot \frac{1}{r-\eta}  \tr(\Delta^\dag \Delta) \label{eq-1131652}\\
&=\mathbbm{1}_{t_1=t_2}\mathbbm{1}_{1\leq i\leq r-\eta}, \nonumber
\end{align}
where \cref{eq-1131652} is due to Schur's lemma (noting that $U_{z_1}$ and $U_{z_2}$ are unitaries acting on the $(r-\eta)$-dimensional space $\mathcal{H}_{\mathrm{anc}}[1:r-\eta]\otimes \mathcal{H}_{\mathrm{B}}[d_2]$) and also the fact that $\bra{i}_{\mathrm{anc}}U_x\Delta=0$ for $i\geq r-\eta+1$, 
\cref{eq-3182341} is by \cref{eq-3182343},
and \cref{eq-3180253} is because $r-\eta\geq \sr d_1$ and $d_2-1\geq d_2/2\geq d_1/(2r)$.

Then, we have
\begin{align}
\E\!\left[\tr(|F(U_x,U_y)|^4)\right]&=\frac{1}{d_1^4}\E\!\Bigg[\sum_{i,j,k,l=1}^r \tr\!\bigg(\kett{K_i'}\Big(\bbra{K_{x,i}}-\bbra{K_{y,i}}\Big)\Big(\kett{K_{x,j}}-\kett{K_{y,j}}\Big)\bbra{K'_j} \nonumber \\
&\qquad\qquad\qquad\qquad\qquad \kett{K'_{k}}\Big(\bbra{K_{x,k}}-\bbra{K_{y,k}}\Big)\Big(\kett{K_{x,l}}-\kett{K_{y,l}}\Big)\bbra{K'_{l}}\bigg)\Bigg]  \nonumber \\ 
&\leq \frac{d_2^2}{d_1^4}\sum_{i,j=1}^r\E\!\bigg[\left|\Big(\bbra{K_{x,i}}-\bbra{K_{y,i}}\Big)\Big(\kett{K_{x,j}}-\kett{K_{y,j}}\Big)\right|^2\bigg]\label{eq-1132143} \\
&\leq \frac{4d_2^2}{d_1^4}\sum_{i,j=1}^r\E\!\left[|\bbrakett{K_{x,i}}{K_{x,j}}|^2+|\bbrakett{K_{y,i}}{K_{y,j}}|^2+|\bbrakett{K_{x,i}}{K_{y,j}}|^2+|\bbrakett{K_{y,i}}{K_{x,j}}|^2\right] \nonumber\\
&\leq \frac{8d_2^2}{d_1^4}\sum_{i,j=1}^r\E\!\left[|\bbrakett{K_{x,i}}{K_{x,j}}|^2+|\bbrakett{K_{y,i}}{K_{y,j}}|^2\right]\label{eq-1132347}\\
&=\frac{16d_2^2}{d_1^4}\sum_{i,j=1}^r \E\!\left[|\bbrakett{K_{x,i}}{K_{x,j}}|^2\right]=\frac{16d_2^2}{d_1^4}\sum_{i,j=1}^r\E\!\left[\left|\tr(K_{x,i}^\dag K_{x,j})\right|^2\right],\nonumber\\
&\leq \frac{16d_2^2}{d_1^4}\cdot\frac{4(r-\eta)^2}{r-\eta}\leq \frac{256}{rd_1^2}\leq \frac{256}{r^3},\label{eq-1140112}
\end{align}
where \cref{eq-1132143} is because \cref{eq-1132146}, \cref{eq-1132347} is because 
\begin{align}
&\sum_{i,j=1}^r |\bbrakett{K_{x,i}}{K_{y,j}}|^2= \|K_x^\dag K_y\|_F^2 =\tr\!\left(K_x^\dag K_y K_y^\dag K_x\right)\leq \|K_xK_x^\dag\|_F \cdot \|K_yK_y^\dag\|_F \nonumber \\
\leq& \frac{1}{2}\left(\|K_x^\dag K_x\|_F^2 + \|K_y^\dag K_y\|_F^2\right)=\frac{1}{2}\left(\sum_{i,j=1}^r |\bbrakett{K_{x,i}}{K_{x,j}}|^2 +\sum_{i,j=1}^r |\bbrakett{K_{y,i}}{K_{y,j}}|^2 \right)\label{eq-3200006},
\end{align}
where $K_x$ denotes the matrix with columns $\kett{K_{x,i}}$, 
and in \cref{eq-1140112}, the first inequality is due to \cref{lemma-3190121} and that $\Delta$ is an isometry from $\mathcal{H}_\mathrm{A}[1:r-\eta]$ to $\mathcal{H}_\mathrm{B}[d_2]\otimes\mathcal{H}_{\mathrm{anc}}[1:r-\eta]$ (further note that $\dim(\mathcal{H}_{\mathrm{anc}}[1:r-\eta])\leq \dim(\mathcal{H}_\mathrm{A}[1:r-\eta])\dim(\mathcal{H}_\mathrm{B}[d_2])$ so we choose $k=1$ in \cref{lemma-3190121}), 
the second and third inequalities is due to $d_2\leq d_1/r+1\leq 2d_1/r$ and $r\leq d_1$.

\subsubsection{Existence for diamond norm}

Note that in \cref{thm-1140332}, the cardinality and distance both depend on $\sr$, which can be arbitrarily close to $0$. 
Here, we give another packing nets w.r.t. diamond norm that does not depend on $\sr$.
We will use the same construction as that given in \cref{sec-3232114}.
Specifically, recall that $\eta\coloneqq d_1-r\lfloor\frac{d_1}{r}\rfloor=d_1-rd_2+r$.
Note that $r-\eta=rd_2-d_1$.

\begin{theorem}\label{thm-3231724}
Suppose $\varepsilon\leq 1/160$.
There exists a finite subset $\mathcal{N}$ of $\{V_{\varepsilon,U}\,|\, U\in\mathbb{U}_{r-\eta}\}$ for $V_{\varepsilon,U}$ defined in \cref{eq-1140331} with cardinality $|\mathcal{N}|\geq \exp((rd_2-d_1)^2/4801)$, such that for any $V_x\neq V_y\in\mathcal{N}$, if we set $\mathcal{E}_x=\tr_{\mathcal{H}_\mathrm{anc}}(V_x(\cdot)V_x^\dag)$  and $\mathcal{E}_y=\tr_{\mathcal{H}_\mathrm{anc}}(V_y(\cdot)V_y^\dag)$, then
\[\|\mathcal{E}_x-\mathcal{E}_y\|_\diamond\geq \frac{1}{80}\varepsilon.\]
\end{theorem}
\begin{proof}
We define 
$\ket{\Psi}\coloneqq \frac{1}{\sqrt{r-\eta}}\sum_{i=1}^{r-\eta}\ket{i}_\mathrm{A}\ket{i}_\mathrm{A}$
be an entangled state on $\mathcal{H}_\mathrm{A}\otimes\mathcal{H}_\mathrm{A}$ and $\Psi=\ketbra{\Psi}{\Psi}$.
We also define $W_0:\mathcal{H}_\mathrm{A}[1:r-\eta]\rightarrow\mathcal{H}_\mathrm{B}[1]\otimes\mathcal{H}_{\mathrm{anc}}[1:r-\eta]$ as:
\[W_0\coloneqq \sqrt{1-\varepsilon^2} \sum_{i=1}^{r-\eta} \ket{i}_\mathrm{anc}\otimes \ket{1}_{\mathrm{B}}\bra{i}_{\mathrm{A}}.\]
We can easily see that
\begin{equation}\label{eq-3240609}
V_0\ket{\Psi}=W_0\ket{\Psi}=\frac{1}{\sqrt{r-\eta}}\kett{W_0}, \quad V_0'\ket{\Psi}=0.
\end{equation}
Then, we need the following lemma.
\begin{lemma}\label{lemma-3231724}
    There exists a finite subset $\mathcal{M}\subseteq \mathbb{U}_{r-\eta}$ with cardinality $|\mathcal{M}|\geq \exp((rd_2-d_1)^2/4801)$ such that for any $U_x\neq U_y\in\mathcal{M}$, 
\begin{equation}\label{eq-3231724}
\frac{1}{r-\eta}\left\|\tr_{\mathcal{H}_\mathrm{anc}}\!\left(\kett{W_0}\Big(\bbra{U_x\Delta}-\bbra{U_y\Delta}\Big)\right)\right\|_1\geq \frac{1}{80}.
\end{equation}
\end{lemma}
The proof of \cref{lemma-3231724} is deferred.

Now, we are able to prove \cref{thm-3231724}.
Let $\mathcal{M}$ be the set given in \cref{lemma-3231724}.
For any $U_x\neq U_y\in\mathcal{M}$, if we set $\mathcal{E}_x=\tr_{\mathcal{H}_\mathrm{anc}}(V_{\varepsilon,U_x}(\cdot)V_{\varepsilon,U_x}^\dag)$  and $\mathcal{E}_y=\tr_{\mathcal{H}_\mathrm{anc}}(V_{\varepsilon,U_y}(\cdot)V_{\varepsilon,U_y}^\dag)$, then we have
\begin{align}
&\|\mathcal{E}_x-\mathcal{E}_y\|_\diamond\geq \|\mathcal{E}_x(\Psi) - \mathcal{E}_y(\Psi)\|_1 \nonumber \\
=&\Big\|\tr_{\mathcal{H}_{\mathrm{anc}}}\!\left((W_0+\varepsilon U_x\Delta)\Psi (W_0+\varepsilon U_x\Delta)^\dag \right)-\tr_{\mathcal{H}_{\mathrm{anc}}}\!\left((W_0+\varepsilon U_y\Delta)\Psi (W_0+\varepsilon U_y\Delta)^\dag\right)\Big\|_1 \label{eq-3232311} \\
=&\frac{1}{r-\eta}\bigg\|\varepsilon^2\Big(\tr_{\mathcal{H}_{\mathrm{anc}}}\!\left(\kettbbra{U_x\Delta}{U_x\Delta}\right)-\tr_{\mathcal{H}_{\mathrm{anc}}}\!\left(\kettbbra{U_y\Delta}{U_y\Delta}\right)\Big) \nonumber\\
&\quad\quad\quad\qquad +\varepsilon \tr_{\mathcal{H}_{\mathrm{anc}}}\!\Big(\kettbbra{W_0}{(U_x-U_y)\Delta}\Big) + \varepsilon\tr_{\mathcal{H}_{\mathrm{anc}}}\!\Big(\kettbbra{(U_x-U_y)\Delta}{W_0}\Big) \bigg\|_1 \label{eq-3232316} \\
\geq & \frac{\varepsilon}{r-\eta}\bigg\|\tr_{\mathcal{H}_{\mathrm{anc}}}\!\Big(\kettbbra{W_0}{(U_x-U_y)\Delta}\Big) + \tr_{\mathcal{H}_{\mathrm{anc}}}\!\Big(\kettbbra{(U_x-U_y)\Delta}{W_0}\Big)\bigg\|_1 - 2 \varepsilon^2 \label{eq-3232317} \\
= & \frac{2\varepsilon}{r-\eta} \bigg\|\tr_{\mathcal{H}_{\mathrm{anc}}}\!\Big(\kettbbra{W_0}{(U_x-U_y)\Delta}\Big)\bigg\|_1-2\varepsilon^2 \label{eq-3232321}\\
\geq &\frac{1}{40}\varepsilon -2\varepsilon^2 \label{eq-3232322} \\
\geq &\frac{1}{80}\varepsilon. \label{eq-3232323}
\end{align}
In \cref{eq-3232311}, we used \cref{eq-3240609}.
In \cref{eq-3232316}, we used $U\Delta\ket{\Psi}=\frac{1}{\sqrt{r-\eta}}\kett{U\Delta}$, $W_0\ket{\Psi}=\frac{1}{\sqrt{r-\eta}}\kett{W_0}$.
In \cref{eq-3232317}, we used 
\[\|\tr_{\mathcal{H}_{\mathrm{anc}}}(\kettbbra{U\Delta}{U\Delta})\|_1 = \tr(\kettbbra{U\Delta}{U\Delta})=r-\eta. \]
In \cref{eq-3232321} we used the fact that $\tr_{\mathcal{H}_{\mathrm{anc}}}\!\Big(\kettbbra{W_0}{(U_x-U_y)\Delta}\Big)$ is a linear operator with support in  $\mathcal{\mathcal{H}_\mathrm{A}}[1:r-\eta]\otimes\mathcal{H}_{\mathrm{B}}[d_2]$ and image in 
$\mathcal{H}_\mathrm{A}[1:r-\eta]\otimes \mathcal{H}_{\mathrm{B}}[1]$,
which is orthogonal to its support; and then we used \cref{fact-3191602}.
In \cref{eq-3232322} we used \cref{eq-3231724}.
In \cref{eq-3232323} we used that $\varepsilon\leq 1/160$.
Therefore, we can lower bound the diamond norm
\[\|\mathcal{E}_x-\mathcal{E}_y\|_\diamond\geq \frac{1}{80}\varepsilon.\]
Thus, the set $\mathcal{N}=\{V_{\varepsilon,U} \,|\, U\in\mathcal{M}\}$ is the desired set.
\end{proof}

Now, we give the proof of \cref{lemma-3231724}.
\begin{proof}[Proof of \cref{lemma-3231724}]
We need the following lemma.
\begin{lemma}\label{lemma-3240011}
For $U_x,U_y\in\mathbb{U}_{r-\eta}$, let us define 
\[F(U_x,U_y)=\frac{1}{r-\eta}\tr_{\mathcal{H}_\mathrm{anc}}\!\left(\kett{W_0}\Big(\bbra{U_x\Delta}-\bbra{U_y\Delta}\Big)\right),\]
then the function $f(U_x,U_y)=\|F(U_x,U_y)\|_1=\tr(|F(U_x,U_y)|)$ is $\sqrt{\frac{2}{r-\eta}}$-Lipschitz with respect to the $\ell_2$-sum of the $2$-norms (Frobenius norm). 
Furthermore, for independent random $U_x,U_y\sim \mathbb{U}_{r-\eta}$, we have $\E\!\left[\tr(|F(U_x,U_y)|^2)\right]\geq\frac{1}{r-\eta}$, and $\E\!\left[\tr(|F(U_x,U_y)|^4)\right]\leq \frac{64}{(r-\eta)^3}$.
\end{lemma}
The proof of \cref{lemma-3240011} is deferred to \cref{sec-3240011}.

By the H\"older's inequality we have
\[\E\!\left[\tr\!\left(|F(U_x,U_y)|^2\right)\right]\leq \E\!\left[\tr\!\left(|F(U_x,U_y)|^4\right)\right]^{1/3}\E\!\left[\tr\!\left(|F(U_x,U_y)|\right)\right]^{2/3},\]
which, combined with \cref{lemma-3240011}, implies
\[\E\!\left[\tr\!\left(|F(U_x,U_y)|\right)\right]^2\geq \frac{1}{64}.\]
Thus $\E\!\left[\tr\!\left(|F(U_x,U_y)|\right)\right]> 1/8$.
Then, since the function $f(U_x,U_y)=\tr(|F(U_x,U_y)|)$ is $\sqrt{\frac{2}{r-\eta}}$-Lipschitz, we can use \cref{lem:corollary17} to prove the concentration result:
\[\Pr\!\left[\tr\!\left(|F(U_x,U_y)|\right)\leq \frac{1}{80}\right]\leq \exp\!\left(-\frac{(r-\eta)^2}{2}\cdot \frac{1}{100\cdot 12}\right)\leq \exp\!\left(-\frac{(r-\eta)^2}{2400}\right).\]
Then, we independently sample $\exp((r-\eta)^2/4801)$ Haar random unitaries in $\mathbb{U}_{r-\eta}$ and the union bound shows that there exists a non-zero probability that for any pair $U_x, U_y$, we have $\tr(|F(U_x,U_y)|)\geq 1/80$.
Thus, there exists a set with cardinality $\geq \exp((r-\eta)^2/4801)=\exp((rd_2-d_1)^2/4801)$ such that \cref{eq-3231724} holds.
\end{proof}

\subsubsection{Second and fourth moments: proof of Lemma \ref{lemma-3240011}}\label{sec-3240011}
The Lipschitz continuity can be seen from \cref{lemma-3192212}, where we treat $W_0$ and $U\Delta$ as linear operators acting on input space with dimension $r-\eta$.

Define $K_{x,i}=\bra{i}_\mathrm{anc} U_x\Delta$, $K_{y,i}=\bra{i}_\mathrm{anc} U_y\Delta$, and $K'_i=\ket{1}_\mathrm{B}\bra{i}_\mathrm{A}=\frac{1}{\sqrt{1-\varepsilon^2}}\bra{i}_\mathrm{anc}W_0$. This means 
\[F(U_x,U_y)=\frac{\sqrt{1-\varepsilon^2}}{r-\eta}\sum_{i=1}^{r-\eta} \kett{K_i'}\Big(\bbra{K_{x,i}}-\bbra{K_{y,i}}\Big).\]
Then, we note that
\begin{align}
\E\!\left[\tr(|F(U_x,U_y)|^2)\right]&=\frac{1-\varepsilon^2}{(r-\eta)^2}\E\!\left[\tr\!\left(\sum_{i,j=1}^{r-\eta} \kett{K_i'}\Big(\bbra{K_{x,i}}-\bbra{K_{y,i}}\Big)\Big(\kett{K_{x,j}}-\kett{K_{y,j}}\Big)\bbra{K_j'}\right)\right] \nonumber\\
&=\frac{1-\varepsilon^2}{(r-\eta)^2}\E\!\left[\sum_{i=1}^{r-\eta} \Big(\bbra{K_{x,i}}-\bbra{K_{y,i}}\Big)\Big(\kett{K_{x,i}}-\kett{K_{y,i}}\Big)\right]\nonumber \\
&=\frac{2(1-\varepsilon^2)}{(r-\eta)^2} (r-\eta)=\frac{2(1-\varepsilon^2)}{r-\eta}\geq \frac{1}{r-\eta} \label{eq-3240226}
\end{align}
where \cref{eq-3240226} is because for $z_1,z_2\in\{x,y\}$, we have
\begin{align}
\E\!\left[\bbrakett{K_{z_1,i}}{K_{z_2,i}}\right]&=\E\!\left[\tr\!\left(K_{z_1,i}^\dag K_{z_2,i}\right)\right] =\tr\!\left(\Delta^\dag \E\!\left[U_{z_1}^\dag \ketbra{i}{i}_\mathrm{anc}U_{z_2} \right]\Delta\right) \nonumber\\
&=\mathbbm{1}_{z_1=z_2}\mathbbm{1}_{1\leq i\leq r-\eta}\cdot \frac{1}{r-\eta}  \tr(\Delta^\dag \Delta) \label{eq-3240239}\\
&=\mathbbm{1}_{z_1=z_2}\mathbbm{1}_{1\leq i\leq r-\eta}, \nonumber
\end{align}
where \cref{eq-3240239} is due to Schur's lemma and also the fact that $\bra{i}_{\mathrm{anc}}U_x\Delta=0$ for $i> r-\eta$.

Then, we also have
\begin{align}
\E\!\left[\tr(|F(U_x,U_y)|^4)\right]&=\frac{1}{(r-\eta)^4}\E\!\Bigg[\sum_{i,j,k,l=1}^{r-\eta} \tr\!\bigg(\kett{K_i'}\Big(\bbra{K_{x,i}}-\bbra{K_{y,i}}\Big)\Big(\kett{K_{x,j}}-\kett{K_{y,j}}\Big)\bbra{K_j'} \nonumber \\
&\qquad\qquad\qquad\qquad\kett{K_{k}'}\Big(\bbra{K_{x,k}}-\bbra{K_{y,k}}\Big)\Big(\kett{K_{x,l}}-\kett{K_{y,l}}\Big)\bbra{K_{l}'}\bigg)\Bigg]  \nonumber \\ 
&\leq \frac{1}{(r-\eta)^4}\sum_{i,j=1}^{r-\eta}\E\!\bigg[\left|\Big(\bbra{K_{x,i}}-\bbra{K_{y,i}}\Big)\Big(\kett{K_{x,j}}-\kett{K_{y,j}}\Big)\right|^2\bigg] \nonumber \\
&\leq \frac{4}{(r-\eta)^4}\sum_{i,j=1}^{r-\eta}\E\!\left[|\bbrakett{K_{x,i}}{K_{x,j}}|^2+|\bbrakett{K_{y,i}}{K_{y,j}}|^2+|\bbrakett{K_{x,i}}{K_{y,j}}|^2+|\bbrakett{K_{y,i}}{K_{x,j}}|^2\right] \nonumber\\
&\leq \frac{8}{(r-\eta)^4}\sum_{i,j=1}^{r-\eta}\E\!\left[|\bbrakett{K_{x,i}}{K_{x,j}}|^2+|\bbrakett{K_{y,i}}{K_{y,j}}|^2\right]\label{eq-3240259}\\
&=\frac{16}{(r-\eta)^4}\sum_{i,j=1}^{r-\eta} \E\!\left[|\bbrakett{K_{x,i}}{K_{x,j}}|^2\right]=\frac{16}{(r-\eta)^4}\sum_{i,j=1}^r\E\!\left[\left|\tr(K_{x,i}^\dag K_{x,j})\right|^2\right],\nonumber\\
&\leq \frac{16}{(r-\eta)^4}\cdot\frac{4(r-\eta)^2}{r-\eta}\leq \frac{64}{(r-\eta)^3},\label{eq-3240303}
\end{align}
where \cref{eq-3240259} is due to a similar argument as that in \cref{eq-3200006}, and in \cref{eq-3240303}, the first inequality is due to \cref{lemma-3190121} and that $\Delta$ is an isometry from $\mathcal{H}_\mathrm{A}[1:r-\eta]$ to $\mathcal{H}_\mathrm{B}[d_2]\otimes\mathcal{H}_{\mathrm{anc}}[1:r-\eta]$ (further note that $\dim(\mathcal{H}_{\mathrm{anc}}[1:r-\eta])\leq \dim(\mathcal{H}_\mathrm{A}[1:r-\eta])\dim(\mathcal{H}_\mathrm{B}[d_2])$ so we choose $k=1$ in \cref{lemma-3190121}).

\subsection{Type II instance: $d_1< rd_2$ with $d_1+r\leq rd_2$, $r\leq d_1$}\label{sec-3230333}

\subsubsection{Construction}\label{sec-3250054}
In this subsection, we will use the definition in \cref{sec-1110250}, where the parameter $d, D$ in \cref{sec-1110250} correspond to $d_1$ and $rd_2$ here.

Suppose $d_2>1$, $d_1+r\leq rd_2$ and $r\leq d_1$. Let $\sr\coloneqq \min\{(rd_2-d_1)/d_1,1\}\in(0,1]$ and let 
\[\underline{\chi}\coloneqq \lfloor \tfrac{d_1}{r}\rfloor,\quad \overline{\chi}\coloneqq \lceil\tfrac{d_1}{r}\rceil, \quad \zeta\coloneqq \min\{\underline{\chi},d_2-\overline{\chi}\}.\]
We note the following relations can be easily verified:
\begin{itemize}
\item $\underline{\chi}\geq 1$, and $d_2-\overline{\chi}\geq 1$ (since $d_1/r+1\leq d_2$),
\item $r(d_2-\overline{\chi})\geq rd_2-d_1-r\geq \tfrac{\sr}{1+\sr}rd_2-r$, thus $r(d_2-\overline{\chi})\geq \max\{r,\tfrac{\sr}{1+\sr}rd_2-r\}\geq \frac{\sr}{2+2\sr}rd_2\geq \frac{\sr}{4}rd_2$,
\item $r\underline{\chi}\geq d_1/2$ and $r(d_2-\overline{\chi})\geq \frac{\sr}{4}rd_2>  \frac{\sr}{4}d_1$, and thus $r\zeta\geq \frac{\sr}{4}d_1$.
\end{itemize}
Let $\varepsilon\in(0,1/2)$ and $\Sigma_\varepsilon\coloneqq \sqrt{1-\varepsilon^2}\sum_{i=1}^{r\zeta}\ketbra{i}{i}_\mathrm{A}+\sum_{i=r\zeta+1}^{d_1}\ketbra{i}{i}_\mathrm{A}$ be a diagonal matrix.
Let $g:[d_1]\rightarrow [d_2]\times [r]$ be a bijective function $g(i)=(\lfloor \frac{i-1}{r}\rfloor+1,i-r\lfloor\frac{i-1}{r}\rfloor)$, i.e., $g$ maps the integers in $[d_1]$ to $[d_2]\times [r]$ in row-major order. Then, we define the linear operator $V_0:\mathcal{H}_\mathrm{A}\rightarrow\mathcal{H}_\mathrm{B}[1:\overline{\chi}]\otimes \mathcal{H}_\mathrm{anc}$ as 
\begin{equation}\label{eq-3242139}
V_0\coloneqq \sum_{i=1}^{d_1} \Big(\ket{g(i)_1}_\mathrm{B}\otimes \ket{g(i)_2}_\mathrm{anc}\bra{i}_\mathrm{A}\Big)\cdot \Sigma_\varepsilon.
\end{equation}
One can easily check that $V_0$ can be written as
\begin{equation}\label{eq-3242146}
V_0= \sum_{i=1}^r \ket{i}_\mathrm{anc}\otimes K_i,
\end{equation}
for $K_i:\mathcal{H}_{\mathrm{A}}\rightarrow\mathcal{H}_\mathrm{B}[1:\overline{\chi}]$ satisfying
\begin{equation}\label{eq-3242152}
\sum_{i=1}^r K_i^\dag K_i=\Sigma_\varepsilon^2 \quad \textup{and}\quad \frac{1}{2}\underline{\chi}\cdot\mathbbm{1}_{i=j} \leq \left|\tr\!\left(K_i^\dag K_j\right)\right|\leq   \overline{\chi}\cdot\mathbbm{1}_{i=j}.
\end{equation}
Define $\Delta:\mathcal{H}_{\mathrm{A}}[1:r\zeta]\rightarrow \mathcal{H}_{\mathrm{B}}[\overline{\chi}+1:d_2]\otimes \mathcal{H}_\mathrm{anc}$ be an arbitrary (but fixed) isometry.
Then, for $U\in\mathbb{U}_{r(d_2-\overline{\chi})}$, we define the isometry $V_{\varepsilon,U}:\mathcal{H}_\mathrm{A}\rightarrow\mathcal{H}_\mathrm{B}\otimes\mathcal{H}_{\mathrm{anc}}$ as
\begin{equation}\label{eq-3191424}
V_{\varepsilon,U}\coloneqq V_0+\varepsilon U\Delta,
\end{equation}
where $U$ acts on $\mathcal{H}_{\mathrm{B}}[\overline{\chi}+1:d_2]\otimes\mathcal{H}_\mathrm{anc}$.
We can verify that $V_{\varepsilon,U}$ is indeed an isometry from \cref{eq-3242139} and the definition of $U\Delta$. 
Note that the image of $U\Delta$ is orthogonal to the image of $V_0$. 
For clarity, we also illustrate our construction in \cref{fig-3290256}.

\begin{figure}[t]
    \centering
    \begin{subfigure}[b]{0.7\linewidth}
    \centering
    \includegraphics[width=1.0\linewidth]{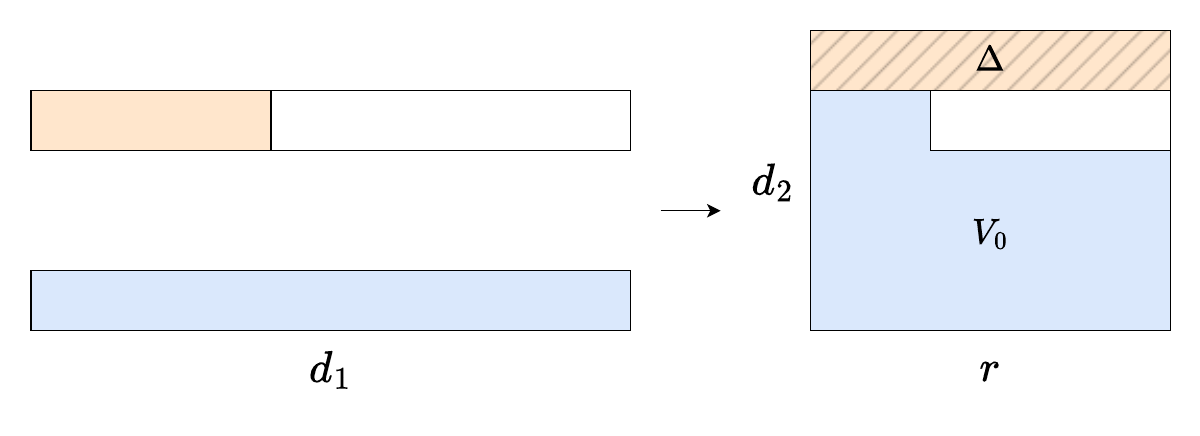}
    \vspace{-10mm}
    \caption{Case $1$}
    \end{subfigure}\\
    \vspace{3mm}
    \begin{subfigure}[b]{0.7\linewidth}
    \centering
    \includegraphics[width=1.0\linewidth]{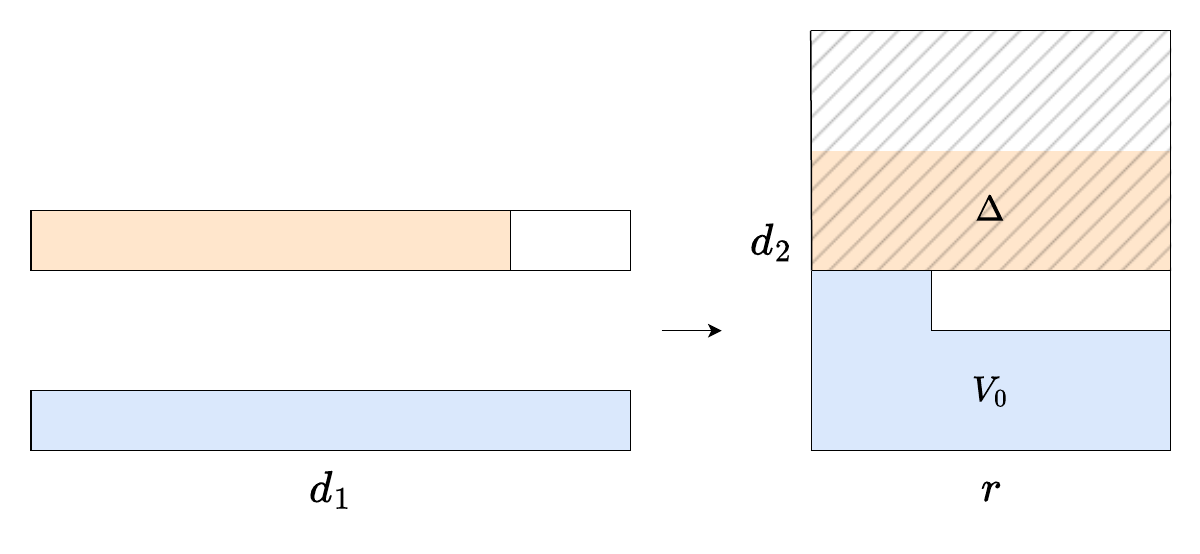}
    \vspace{-10mm}
    \caption{Case $2$}
    \end{subfigure}
    \caption{Illustration of our construction. There are two cases depending on whether Case $1$ : $\floor{\frac{d_1}{r}}\ge d_2- \ceil{\frac{d_1}{r}}$ or  Case $2$ : $\floor{\frac{d_1}{r}} < d_2- \ceil{\frac{d_1}{r}}$. We define linear operators $V_0$ and $\Delta$ from the $d_1$-dimensional space $\mathcal{H}_\mathrm{A}$ to the $d_2r$-dimensional space $\mathcal{H}_\mathrm{B}\otimes\mathcal{H}_\mathrm{anc}$. The Haar randomness is applied on the hatched area.}
    \label{fig-3290256}
\end{figure}

\subsubsection{Existence for Choi trace norm}

Then, we prove that there exists a large set of isometries $V_{\varepsilon,U}$ with good separation property.
\begin{theorem}\label{thm-3191542}
Suppose $\varepsilon\leq \sr^{3/2}/25600$. There exists a finite subset $\mathcal{N}$ of $\{V_{\varepsilon,U}\,|\, U\in\mathbb{U}_{r(d_2-\overline{\chi})}\}$ for $V_{\varepsilon,U}$ defined in \cref{eq-3191424} with cardinality 
\[|\mathcal{N}|\geq \exp\!\left(\frac{\sr^2(rd_2-d_1)\min\{d_1,rd_2-d_1\}}{3840001}\right),\]
such that for any $V_x\neq V_y\in\mathcal{N}$, if we set $\mathcal{E}_x=\tr_{\mathcal{H}_\mathrm{anc}}(V_x(\cdot)V_x^\dag)$  and $\mathcal{E}_y=\tr_{\mathcal{H}_\mathrm{anc}}(V_y(\cdot)V_y^\dag)$, then
\[\frac{1}{d_1}\|C_{\mathcal{E}_x}-C_{\mathcal{E}_y}\|_1\geq \frac{\sr^{3/2}}{12800}\varepsilon.\]
\end{theorem}
\begin{proof}
First, we need the following lemma.
\begin{lemma}\label{lemma-3191542}
    There exists a finite subset $\mathcal{M}\subseteq \mathbb{U}_{r(d_2-\overline{\chi})}$ with cardinality $|\mathcal{M}|\geq \exp(\sr^3d_1(rd_2-d_1)/3840001)$ such that for any $U_x\neq U_y\in\mathcal{M}$, 
\begin{equation}\label{eq-3191435}
\frac{1}{d_1}\left\|\tr_{\mathcal{H}_\mathrm{anc}}\!\left(\kett{V_0}\Big(\bbra{U_x\Delta}-\bbra{U_y\Delta}\Big)\right)\right\|_1\geq \frac{\sr^{3/2}}{12800}.
\end{equation}
\end{lemma}
The proof of \cref{lemma-3191542} is deferred.

Now, we are able to prove \cref{thm-3191542}. Let $\mathcal{M}$ be the set given in \cref{lemma-3191542}.
Then, for any $U_x\neq U_y\in\mathcal{M}$, if we set $\mathcal{E}_x=\tr_{\mathcal{H}_\mathrm{anc}}(V_{\varepsilon,U_x}(\cdot)V_{\varepsilon,U_{x}}^\dag)$  and $\mathcal{E}_y=\tr_{\mathcal{H}_\mathrm{anc}}(V_{\varepsilon,U_y}(\cdot)V_{\varepsilon,U_y}^\dag)$, we have
\begin{align}
&\|C_{\mathcal{E}_{x}}-C_{\mathcal{E}_{y}}\|_1 \nonumber \\
=&\Big\|\tr_{\mathcal{H}_{\mathrm{anc}}}\!\left(\kettbbra{V_{\varepsilon,U_x}}{V_{\varepsilon,U_x}}\right)-\tr_{\mathcal{H}_{\mathrm{anc}}}\!\left(\kettbbra{V_{\varepsilon,U_y}}{V_{\varepsilon,U_y}}\right)\Big\|_1 \nonumber\\
=&\bigg\|\varepsilon^2\Big(\tr_{\mathcal{H}_{\mathrm{anc}}}\!\left(\kettbbra{U_x\Delta}{U_x\Delta}\right)-\tr_{\mathcal{H}_{\mathrm{anc}}}\!\left(\kettbbra{U_y\Delta}{U_y\Delta}\right)\Big) \nonumber\\
&\quad\quad\quad +\varepsilon\tr_{\mathcal{H}_{\mathrm{anc}}}\!\Big(\kettbbra{V_0}{(U_x-U_y)\Delta}\Big) + \varepsilon\tr_{\mathcal{H}_{\mathrm{anc}}}\!\Big(\kettbbra{(U_x-U_y)\Delta}{V_0}\Big) \bigg\|_1 \nonumber \\
\geq & \varepsilon\bigg\|\tr_{\mathcal{H}_{\mathrm{anc}}}\!\Big(\kettbbra{V_0}{(U_x-U_y)\Delta}\Big) + \tr_{\mathcal{H}_{\mathrm{anc}}}\!\Big(\kettbbra{(U_x-U_y)\Delta}{V_0}\Big)\bigg\|_1 - 2 \varepsilon^2 d_1 \label{eq-3191437} \\
=& 2\varepsilon\bigg\|\tr_{\mathcal{H}_{\mathrm{anc}}}\!\Big(\kettbbra{V_0}{(U_x-U_y)\Delta}\Big)\bigg\|_1-2\varepsilon^2 d_1 \label{eq-3191544}\\
\geq & \frac{2\sr^{3/2}}{12800}\varepsilon d_1 -2\varepsilon^2d_1 \label{eq-3191545} \\
\geq & \frac{\sr^{3/2}}{12800}\varepsilon d_1. \label{eq-3191546}
\end{align}
In \cref{eq-3191437}, we used $\|\tr_{\mathcal{H}_{\mathrm{anc}}}(\kettbbra{U\Delta}{U\Delta})\|_1 = \tr(\kettbbra{U\Delta}{U\Delta})=r\zeta\leq r\underline{\chi}\leq d_1$. 
In \cref{eq-3191544} we used the fact that $\tr_{\mathcal{H}_{\mathrm{anc}}}\!\Big(\kettbbra{V_0}{(U_x-U_y)\Delta}\Big)$ is a linear operator with support in  $\mathcal{\mathcal{H}_\mathrm{A}}[1:r\zeta]\otimes\mathcal{H}_{\mathrm{B}}[\overline{\chi}+1:d_2]$ and image in $\mathcal{H}_\mathrm{A}\otimes\mathcal{H}_{\mathrm{B}}[1:\overline{\chi}]$ 
which is orthogonal to its support; and then we used \cref{fact-3191602}.
In \cref{eq-3191545} we used \cref{eq-3191435}.
In \cref{eq-3191546} we used that $\varepsilon\leq \sr^{3/2}/25600$.
Therefore, we can lower bound the Choi trace norm 
\[\frac{1}{d_1}\left\|C_{\mathcal{E}_{x}}-C_{\mathcal{E}_{y}}\right\|_1\geq \frac{\sr^{3/2}}{12800}\varepsilon.\]
Thus, the set $\mathcal{N}=\{V_{\varepsilon,U} \,|\, U\in\mathcal{M}\}$ is the desired set, and the cardinality is $\exp(\sr^3 d_1(rd_2-d_1)/3840001)=\exp(\sr^2 (rd_2-d_1)\min\{d_1,rd_2-d_1\}/3840001)$.
\end{proof}

Now, we give the proof of \cref{lemma-3191542}.
\begin{proof}[Proof of \cref{lemma-3191542}]
First, we need the following lemma:
\begin{lemma}\label{lemma-3191720}
For $U_x,U_y\in\mathbb{U}_{r(d_2-\overline{\chi})}$, let us define 
\[F(U_x,U_y)=\frac{1}{d_1}\tr_{\mathcal{H}_\mathrm{anc}}\!\left(\kett{V_0}\Big(\bbra{U_x\Delta}-\bbra{U_y\Delta}\Big)\right),\]
then the function $f(U_x,U_y)=\|F(U_x,U_y)\|_1=\tr(|F(U_x,U_y)|)$ is $\sqrt{\frac{2}{d_1}}$-Lipschitz with respect to the $\ell_2$-sum of the $2$-norms (Frobenius norm). 
Furthermore, for independent random $U_x,U_y\sim \mathbb{U}_{r(d_2-\overline{\chi})}$, we have $\E\!\left[\tr(|F(U_x,U_y)|^2)\right]\geq\frac{\sr}{4r}$, and $\E\!\left[\tr(|F(U_x,U_y)|^4)\right]\leq \frac{256}{r^3}$.
\end{lemma}
The proof of \cref{lemma-3191720} is deferred to \cref{sec-3211819}.

By the H\"older's inequality we have
\[\E\!\left[\tr\!\left(|F(U_x,U_y)|^2\right)\right]\leq \E\!\left[\tr\!\left(|F(U_x,U_y)|^4\right)\right]^{1/3}\E\!\left[\tr\!\left(|F(U_x,U_y)|\right)\right]^{2/3},\]
which, combined with \cref{lemma-3191720}, implies
\[\E\!\left[\tr\!\left(|F(U_x,U_y)|\right)\right]^2\geq \frac{\sr^3}{16384}.\]
Thus $\E\!\left[\tr\!\left(|F(U_x,U_y)|\right)\right]> \sr^{3/2}/128$.
Then, since the fucntion $f(U_x,U_y)=\tr(|F(U_x,U_y)|)$ is $\sqrt{\frac{2}{d_1}}$-Lipschitz, we can use \cref{lem:corollary17} to prove the concentration result:
\[\Pr\!\left[\tr\!\left(|F(U_x,U_y)|\right)\leq \frac{\sr^{3/2}}{12800}\right]\leq \exp\!\left(-\frac{d_1r(d_2-\overline{\chi})}{2}\cdot \frac{\sr^3}{40000\cdot 12}\right)\leq \exp\!\left(-\frac{\sr^3 d_1(rd_2-d_1)}{1920000}\right),\]
where we used $r(d_2-\overline{\chi})=r\lfloor \frac{rd_2-d_1}{r} \rfloor \geq \frac{rd_2-d_1}{2}$ (since $rd_2-d_1\geq r$).
Then, we independently sample $\exp(\sr^3d_1(rd_2-d_1)/3840001)$ Haar random unitaries in $\mathbb{U}_{r(d_2-\overline{\chi})}$ and the union bound shows that there exists a non-zero probability that for any pair $U_x, U_y$, we have $\tr(|F(U_x,U_y)|)\geq \sr^{3/2}/12800$.
Thus, there exists a set with cardinality $\geq \exp(\sr^3d_1(rd_2-d_1)/3840001)$ such that \cref{eq-3191435} holds.
\end{proof}

\subsubsection{Second and fourth moments: proof of Lemma \ref{lemma-3191720}}\label{sec-3211819}
The Lipschitz continuity can be seen from \cref{lemma-3192212}.

Define $K_{x,i}=\bra{i}_\mathrm{anc} U_x\Delta$ and $K_{y,i}=\bra{i}_\mathrm{anc} U_y\Delta$. Thus
$F(U_x,U_y)=\frac{1}{d_1}\sum_{i=1}^r\kett{K_i}\Big(\bbra{K_{x,i}}-\bbra{K_{y,i}}\Big)$.
Then, we note that
\begin{align}
\E\!\left[\tr(|F(U_x,U_y)|^2)\right]&=\frac{1}{d_1^2}\E\!\left[\tr\!\left(\sum_{i,j=1}^r \kett{K_i}\Big(\bbra{K_{x,i}}-\bbra{K_{y,i}}\Big)\Big(\kett{K_{x,j}}-\kett{K_{y,j}}\Big)\bbra{K_j}\right)\right] \nonumber\\
&=\frac{1}{d_1^2}\E\!\left[\sum_{i=1}^r \bbrakett{K_i}{K_i}\Big(\bbra{K_{x,i}}-\bbra{K_{y,i}}\Big)\Big(\kett{K_{x,i}}-\kett{K_{y,i}}\Big)\right]\nonumber \\
&=\frac{2\zeta}{d_1^2}\sum_{i=1}^r\bbrakett{K_i}{K_i}  \label{eq-3192313} \\
&\geq\frac{\zeta}{d_1}\geq \frac{\sr}{4r},\label{eq-3192315}
\end{align}
where \cref{eq-3192313} is because for $t_1,t_2\in\{x,y\}$, we have
\begin{align}
\E\!\left[\bbrakett{K_{t_1,i}}{K_{t_2,i}}\right]&=\E\!\left[\tr\!\left(K_{t_1,i}^\dag K_{t_2,i}\right)\right] =\tr\!\left(\Delta^\dag \E\!\left[U_{t_1}^\dag \ketbra{i}{i}_\mathrm{anc}U_{t_2} \right]\Delta\right) \nonumber\\
&=\mathbbm{1}_{t_1=t_2}\cdot \frac{1}{r}  \tr(\Delta^\dag \Delta) \label{eq-3192320}\\
&=\mathbbm{1}_{t_1=t_2}\cdot \zeta, \nonumber
\end{align}
where \cref{eq-3192320} is due to Schur's lemma, 
\cref{eq-3192315} is because $\sum_{i=1}^r\tr(K_i^\dag K_i)=\tr(V_0^\dag V_0)\geq \frac{1}{2}\tr(I_{\mathcal{H}_\mathrm{A}})=d_1/2$.

Then, we also have
\begin{align}
\E\!\left[\tr(|F(U_x,U_y)|^4)\right]&=\frac{1}{d_1^4}\E\!\Bigg[\sum_{i,j,k,l=1}^r \tr\!\bigg(\kett{K_i}\Big(\bbra{K_{x,i}}-\bbra{K_{y,i}}\Big)\Big(\kett{K_{x,j}}-\kett{K_{y,j}}\Big)\bbra{K_j} \nonumber\\
&\qquad\qquad\qquad\qquad\kett{K_{k}}\Big(\bbra{K_{x,k}}-\bbra{K_{y,k}}\Big)\Big(\kett{K_{x,l}}-\kett{K_{y,l}}\Big)\bbra{K_{l}}\bigg)\Bigg]  \nonumber \\ 
&\leq \frac{\overline{\chi}^2}{d_1^4}\sum_{i,j=1}^r\E\!\bigg[\left|\Big(\bbra{K_{x,i}}-\bbra{K_{y,i}}\Big)\Big(\kett{K_{x,j}}-\kett{K_{y,j}}\Big)\right|^2\bigg]\label{eq-3192355} \\
&\leq \frac{4\overline{\chi}^2}{d_1^4}\sum_{i,j=1}^r\E\!\left[|\bbrakett{K_{x,i}}{K_{x,j}}|^2+|\bbrakett{K_{y,i}}{K_{y,j}}|^2+|\bbrakett{K_{x,i}}{K_{y,j}}|^2+|\bbrakett{K_{y,i}}{K_{x,j}}|^2\right] \nonumber\\
&\leq \frac{8\overline{\chi}^2}{d_1^4}\sum_{i,j=1}^r\E\!\left[|\bbrakett{K_{x,i}}{K_{x,j}}|^2+|\bbrakett{K_{y,i}}{K_{y,j}}|^2\right]\label{eq-3192359}\\
&=\frac{16\overline{\chi}^2}{d_1^4}\sum_{i,j=1}^r \E\!\left[|\bbrakett{K_{x,i}}{K_{x,j}}|^2\right]=\frac{16\overline{\chi}^2}{d_1^4}\sum_{i,j=1}^r\E\!\left[\left|\tr(K_{x,i}^\dag K_{x,j})\right|^2\right]\nonumber\\
&\leq \frac{16 \overline{\chi}^2}{d_1^4}\cdot\frac{4r^2\zeta^2}{r}= \frac{64\cdot r\overline{\chi}^2\zeta^2}{d_1^4}\leq\frac{64\cdot r\overline{\chi}^2\underline{\chi}^2}{d_1^4} \leq \frac{256}{r^3},\label{eq-3200000}
\end{align}
where \cref{eq-3192355} is due to \cref{eq-3242152}, \cref{eq-3192359} due to the similar argument as that in \cref{eq-3200006}, 
and in \cref{eq-3200000}, the first inequality is due to \cref{lemma-3190121} and that $\Delta$ is an isometry from $\mathcal{H}_\mathrm{A}[1:r\zeta]$ to $\mathcal{H}_\mathrm{B}[\overline{\chi}+1:d_2]\otimes\mathcal{H}_{\mathrm{anc}}$ (further note that $\dim(\mathcal{H}_{\mathrm{anc}})\leq \dim(\mathcal{H}_\mathrm{A}[1:r\zeta])\dim(\mathcal{H}_\mathrm{B}[\overline{\chi}+1:d_2])$ so we choose $k=1$ in \cref{lemma-3190121}).

\subsubsection{Existence for diamond norm}

Note that in \cref{thm-3191542}, the cardinality and distance both depend on $\sr$, which can be arbitrarily close to $0$. 
Here, we give another packing nets w.r.t. diamond norm that does not depend on $\sr$.
We will use the same construction as that given in \cref{sec-3250054}.
Specifically, recall that $\underline{\chi}\coloneqq \lfloor\frac{d_1}{r}\rfloor$, $\overline{\chi}\coloneqq \lceil\frac{d_1}{r}\rceil$, and $\zeta\coloneqq\min\{\underline{\chi},d_2-\overline{\chi}\}$.

\begin{theorem}\label{thm-3250101}
Suppose $\varepsilon\leq 1/160$.
There exists a finite subset $\mathcal{N}$ of $\{V_{\varepsilon,U}\,|\, U\in\mathbb{U}_{r(d_2-\overline{\chi})}\}$ for $V_{\varepsilon,U}$ defined in \cref{eq-3191424} with cardinality 
\[|\mathcal{N}|\geq \exp\!\left(\frac{(rd_2-d_1)\min\{d_1,rd_2-d_1\}}{38404}\right),\]
such that for any $V_x\neq V_y\in\mathcal{N}$, if we set $\mathcal{E}_x=\tr_{\mathcal{H}_\mathrm{anc}}(V_x(\cdot)V_x^\dag)$  and $\mathcal{E}_y=\tr_{\mathcal{H}_\mathrm{anc}}(V_y(\cdot)V_y^\dag)$, then
\[\|\mathcal{E}_x-\mathcal{E}_y\|_\diamond\geq \frac{1}{80}\varepsilon.\]
\end{theorem}
\begin{proof}
We define 
$\ket{\Psi}\coloneqq \frac{1}{\sqrt{r\zeta}}\sum_{i=1}^{r\zeta}\ket{i}_\mathrm{A}\ket{i}_\mathrm{A}$
be an entangled state on $\mathcal{H}_\mathrm{A}\otimes\mathcal{H}_\mathrm{A}$ and $\Psi=\ketbra{\Psi}{\Psi}$.
We also define $W_0:\mathcal{H}_\mathrm{A}[1:r\zeta]\rightarrow\mathcal{H}_\mathrm{B}[1:\zeta]\otimes\mathcal{H}_{\mathrm{anc}}$ as:
\[W_0\coloneqq \sqrt{1-\varepsilon^2} \sum_{i=1}^{r\zeta} \ket{g(i)_1}_{\mathrm{B}}\otimes \ket{g(i)_2}_\mathrm{anc}\bra{i}_{\mathrm{A}}.\]
We can easily see that (c.f. \cref{eq-3242139}):
\begin{equation}\label{eq-3250107}
V_0\ket{\Psi}=W_0\ket{\Psi}=\frac{1}{\sqrt{r\zeta}}\kett{W_0}.
\end{equation}
Then, we need the following lemma.
\begin{lemma}\label{lemma-3250110}
    There exists a finite subset $\mathcal{M}\subseteq \mathbb{U}_{r(d_2-\overline{\chi})}$ with cardinality $|\mathcal{M}|\geq \exp(r\zeta(rd_2-d_1)/9601)$ such that for any $U_x\neq U_y\in\mathcal{M}$, 
\begin{equation}\label{eq-3250110}
\frac{1}{r\zeta}\left\|\tr_{\mathcal{H}_\mathrm{anc}}\!\left(\kett{W_0}\Big(\bbra{U_x\Delta}-\bbra{U_y\Delta}\Big)\right)\right\|_1\geq \frac{1}{80}.
\end{equation}
\end{lemma}
The proof of \cref{lemma-3250110} is deferred.

Now, we are able to prove \cref{thm-3231724}.
Let $\mathcal{M}$ be the set given in \cref{lemma-3231724}.
For any $U_x\neq U_y\in\mathcal{M}$, if we set $\mathcal{E}_x=\tr_{\mathcal{H}_\mathrm{anc}}(V_{\varepsilon,U_x}(\cdot)V_{\varepsilon,U_x}^\dag)$  and $\mathcal{E}_y=\tr_{\mathcal{H}_\mathrm{anc}}(V_{\varepsilon,U_y}(\cdot)V_{\varepsilon,U_y}^\dag)$, then we have
\begin{align}
&\|\mathcal{E}_x-\mathcal{E}_y\|_\diamond\geq \|\mathcal{E}_x(\Psi) - \mathcal{E}_y(\Psi)\|_1 \nonumber \\
=&\Big\|\tr_{\mathcal{H}_{\mathrm{anc}}}\!\left((W_0+\varepsilon U_x\Delta)\Psi (W_0+\varepsilon U_x\Delta)^\dag \right)-\tr_{\mathcal{H}_{\mathrm{anc}}}\!\left((W_0+\varepsilon U_y\Delta)\Psi (W_0+\varepsilon U_y\Delta)^\dag\right)\Big\|_1 \label{eq-3250117} \\
=&\frac{1}{r\zeta}\bigg\|\varepsilon^2\Big(\tr_{\mathcal{H}_{\mathrm{anc}}}\!\left(\kettbbra{U_x\Delta}{U_x\Delta}\right)-\tr_{\mathcal{H}_{\mathrm{anc}}}\!\left(\kettbbra{U_y\Delta}{U_y\Delta}\right)\Big) \nonumber\\
&\quad\quad\quad\qquad +\varepsilon \tr_{\mathcal{H}_{\mathrm{anc}}}\!\Big(\kettbbra{W_0}{(U_x-U_y)\Delta}\Big) + \varepsilon\tr_{\mathcal{H}_{\mathrm{anc}}}\!\Big(\kettbbra{(U_x-U_y)\Delta}{W_0}\Big) \bigg\|_1 \label{eq-3250118} \\
\geq & \frac{\varepsilon}{r\zeta}\bigg\|\tr_{\mathcal{H}_{\mathrm{anc}}}\!\Big(\kettbbra{W_0}{(U_x-U_y)\Delta}\Big) + \tr_{\mathcal{H}_{\mathrm{anc}}}\!\Big(\kettbbra{(U_x-U_y)\Delta}{W_0}\Big)\bigg\|_1 - 2 \varepsilon^2 \label{eq-3250119} \\
= & \frac{2\varepsilon}{r\zeta} \bigg\|\tr_{\mathcal{H}_{\mathrm{anc}}}\!\Big(\kettbbra{W_0}{(U_x-U_y)\Delta}\Big)\bigg\|_1-2\varepsilon^2 \label{eq-3250120}\\
\geq &\frac{1}{40}\varepsilon -2\varepsilon^2 \label{eq-3250121} \\
\geq &\frac{1}{80}\varepsilon. \label{eq-3250122}
\end{align}
In \cref{eq-3250117}, we used \cref{eq-3250107}.
In \cref{eq-3250118}, we used $U\Delta\ket{\Psi}=\frac{1}{\sqrt{r\zeta}}\kett{U\Delta}$, $W_0\ket{\Psi}=\frac{1}{\sqrt{r\zeta}}\kett{W_0}$.
In \cref{eq-3250119}, we used $\|\tr_{\mathcal{H}_{\mathrm{anc}}}(\kettbbra{U\Delta}{U\Delta})\|_1 = \tr(\kettbbra{U\Delta}{U\Delta})=r\zeta$.
In \cref{eq-3250120} we used the fact that $\tr_{\mathcal{H}_{\mathrm{anc}}}\!\Big(\kettbbra{W_0}{(U_x-U_y)\Delta}\Big)$ is a linear operator with support in  $\mathcal{\mathcal{H}_\mathrm{A}}[1:r\zeta]\otimes\mathcal{H}_{\mathrm{B}}[\overline{\chi}+1:d_2]$ and image in 
$\mathcal{H}_\mathrm{A}[1:r\zeta]\otimes \mathcal{H}_{\mathrm{B}}[1:\zeta]$,
which is orthogonal to its support; and then we used \cref{fact-3191602}.
In \cref{eq-3250121} we used \cref{eq-3250110}.
In \cref{eq-3250122} we used that $\varepsilon\leq 1/160$.
Therefore, we can lower bound the diamond norm
\[\|\mathcal{E}_x-\mathcal{E}_y\|_\diamond\geq \frac{1}{80}\varepsilon.\]
Thus, the set $\mathcal{N}=\{V_{\varepsilon,U} \,|\, U\in\mathcal{M}\}$ is the desired set, and the cardinality is at least $\exp\!\left(\frac{r\zeta(rd_2-d_1)}{9601}\right)\geq \exp\!\left(\frac{(rd_2-d_1)\min\{d_1,rd_2-d_1\}}{38404}\right)$,
since $r\zeta\geq \frac{\sr}{4}d_1=\frac{1}{4}\min\{d_1,rd_2-d_1\}$.
\end{proof}

Now, we give the proof of \cref{lemma-3250110}.
\begin{proof}[Proof of \cref{lemma-3250110}]
We need the following lemma:
\begin{lemma}\label{lemma-3250133}
For $U_x,U_y\in\mathbb{U}_{r(d_2-\overline{\chi})}$, let us define 
\[F(U_x,U_y)=\frac{1}{r\zeta}\tr_{\mathcal{H}_\mathrm{anc}}\!\left(\kett{W_0}\Big(\bbra{U_x\Delta}-\bbra{U_y\Delta}\Big)\right),\]
then the function $f(U_x,U_y)=\|F(U_x,U_y)\|_1=\tr(|F(U_x,U_y)|)$ is $\sqrt{\frac{2}{r\zeta}}$-Lipschitz with respect to the $\ell_2$-sum of the $2$-norms (Frobenius norm). 
Furthermore, for independent random $U_x,U_y\sim \mathbb{U}_{r(d_2-\overline{\chi})}$, we have $\E\!\left[\tr(|F(U_x,U_y)|^2)\right]\geq\frac{1}{r}$, and $\E\!\left[\tr(|F(U_x,U_y)|^4)\right]\leq \frac{64}{r^3}$.
\end{lemma}
The proof of \cref{lemma-3250133} is deferred to \cref{sec-3250111}.

By the H\"older's inequality we have
\[\E\!\left[\tr\!\left(|F(U_x,U_y)|^2\right)\right]\leq \E\!\left[\tr\!\left(|F(U_x,U_y)|^4\right)\right]^{1/3}\E\!\left[\tr\!\left(|F(U_x,U_y)|\right)\right]^{2/3},\]
which, combined with \cref{lemma-3250133}, implies
\[\E\!\left[\tr\!\left(|F(U_x,U_y)|\right)\right]^2\geq \frac{1}{64}.\]
Thus $\E\!\left[\tr\!\left(|F(U_x,U_y)|\right)\right]> 1/8$.
Then, since the function $f(U_x,U_y)=\tr(|F(U_x,U_y)|)$ is $\sqrt{\frac{2}{r\zeta}}$-Lipschitz, we can use \cref{lem:corollary17} to prove the concentration result:
\[\Pr\!\left[\tr\!\left(|F(U_x,U_y)|\right)\leq \frac{1}{80}\right]\leq \exp\!\left(-\frac{r\zeta\cdot r(d_2-\overline{\chi})}{2}\cdot \frac{1}{100\cdot 12}\right)\leq \exp\!\left(-\frac{r\zeta(rd_2-d_1)}{4800}\right),\]
where in the last inequality we used $r(d_2-\overline{\chi})=r\lfloor\frac{rd_2-d_1}{r}\rfloor\geq \frac{rd_2-d_1}{2}$ (since $rd_2-d_1\geq r$).
Then, we independently sample $\exp(r\zeta(rd_2-d_1)/9601)$ Haar random unitaries in $\mathbb{U}_{r(d_2-\overline{\chi})}$ and the union bound shows that there exists a non-zero probability that for any pair $U_x, U_y$, we have $\tr(|F(U_x,U_y)|)\geq 1/80$.
Thus, there exists a set with cardinality $\geq \exp(r\zeta(rd_2-d_1)/9601)$ such that \cref{eq-3250110} holds.
\end{proof}

\subsubsection{Second and fourth moments: proof of Lemma \ref{lemma-3250133}}\label{sec-3250111}
The Lipschitz continuity can be seen from \cref{lemma-3192212}, where we treat $W_0$ and $U\Delta$ as linear operators acting on the input space with dimension $r\zeta$.

Define $K_{x,i}=\bra{i}_\mathrm{anc} U_x\Delta$, $K_{y,i}=\bra{i}_\mathrm{anc} U_y\Delta$, and $K_i'=\bra{i}_\mathrm{anc} W_0$.
We can easily see that 
\begin{equation}\label{eq-3250158}
\frac{1}{2}\zeta \cdot \mathbbm{1}_{i=j}\leq \left|\tr(K_i'^\dag K'_j)\right|\leq \zeta\cdot\mathbbm{1}_{i=j},
\end{equation}
and also
\[F(U_x,U_y)=\frac{1}{r\zeta}\sum_{i=1}^{r} \kett{K'_i}\Big(\bbra{K_{x,i}}-\bbra{K_{y,i}}\Big).\]
Then, we note that
\begin{align}
\E\!\left[\tr(|F(U_x,U_y)|^2)\right]&=\frac{1}{r^2\zeta^2}\E\!\left[\tr\!\left(\sum_{i,j=1}^{r} \kett{K'_i}\Big(\bbra{K_{x,i}}-\bbra{K_{y,i}}\Big)\Big(\kett{K_{x,j}}-\kett{K_{y,j}}\Big)\bbra{K'_j}\right)\right] \nonumber\\
&\geq \frac{1}{2r^2\zeta}\E\!\left[\sum_{i=1}^{r} \Big(\bbra{K_{x,i}}-\bbra{K_{y,i}}\Big)\Big(\kett{K_{x,i}}-\kett{K_{y,i}}\Big)\right]\label{eq-3250219} \\
&=\frac{1}{2r^2\zeta} \cdot 2r\zeta= \frac{1}{r} \label{eq-3250224}
\end{align}
where \cref{eq-3250219} is due to \cref{eq-3250158}, and \cref{eq-3250224} is because for $t_1,t_2\in\{x,y\}$, we have
\begin{align}
\E\!\left[\bbrakett{K_{t_1,i}}{K_{t_2,i}}\right]&=\E\!\left[\tr\!\left(K_{t_1,i}^\dag K_{t_2,i}\right)\right] =\tr\!\left(\Delta^\dag \E\!\left[U_{t_1}^\dag \ketbra{i}{i}_\mathrm{anc}U_{t_2} \right]\Delta\right) \nonumber\\
&=\mathbbm{1}_{t_1=t_2}\cdot \frac{1}{r}  \tr(\Delta^\dag \Delta) \label{eq-3250225}\\
&=\mathbbm{1}_{t_1=t_2}\zeta, \nonumber
\end{align}
where \cref{eq-3250225} is due to Schur's lemma..

Then, we also have
\begin{align}
\E\!\left[\tr(|F(U_x,U_y)|^4)\right]&=\frac{1}{r^4\zeta^4}\E\!\Bigg[\sum_{i,j,k,l=1}^{r} \tr\!\bigg(\kett{K_i'}\Big(\bbra{K_{x,i}}-\bbra{K_{y,i}}\Big)\Big(\kett{K_{x,j}}-\kett{K_{y,j}}\Big)\bbra{K_j'} \nonumber \\
&\qquad\qquad\qquad\qquad\kett{K_{k}'}\Big(\bbra{K_{x,k}}-\bbra{K_{y,k}}\Big)\Big(\kett{K_{x,l}}-\kett{K_{y,l}}\Big)\bbra{K_{l}'}\bigg)\Bigg]  \nonumber \\ 
&\leq \frac{1}{r^4\zeta^2}\sum_{i,j=1}^{r}\E\!\bigg[\left|\Big(\bbra{K_{x,i}}-\bbra{K_{y,i}}\Big)\Big(\kett{K_{x,j}}-\kett{K_{y,j}}\Big)\right|^2\bigg] \label{eq-3250230} \\
&\leq \frac{4}{r^4\zeta^2}\sum_{i,j=1}^{r}\E\!\left[|\bbrakett{K_{x,i}}{K_{x,j}}|^2+|\bbrakett{K_{y,i}}{K_{y,j}}|^2+|\bbrakett{K_{x,i}}{K_{y,j}}|^2+|\bbrakett{K_{y,i}}{K_{x,j}}|^2\right] \nonumber\\
&\leq \frac{8}{r^4\zeta^2}\sum_{i,j=1}^{r}\E\!\left[|\bbrakett{K_{x,i}}{K_{x,j}}|^2+|\bbrakett{K_{y,i}}{K_{y,j}}|^2\right]\label{eq-3250233}\\
&=\frac{16}{r^4\zeta^2}\sum_{i,j=1}^{r} \E\!\left[|\bbrakett{K_{x,i}}{K_{x,j}}|^2\right]=\frac{16}{r^4\zeta^2}\sum_{i,j=1}^r\E\!\left[\left|\tr(K_{x,i}^\dag K_{x,j})\right|^2\right],\nonumber\\
&\leq \frac{16}{r^4\zeta^2}\cdot\frac{4 (r\zeta)^2}{r}= \frac{64}{r^3},\label{eq-3250234}
\end{align}
where \cref{eq-3250230} is due to \cref{eq-3250158},
\cref{eq-3250233} is due to a similar argument as that in \cref{eq-3200006}, and in \cref{eq-3250234}, the first inequality is due to \cref{lemma-3190121} and that $\Delta$ is an isometry from $\mathcal{H}_\mathrm{A}[1:r\zeta]$ to $\mathcal{H}_\mathrm{B}[\overline{\chi}+1:d_2]\otimes\mathcal{H}_{\mathrm{anc}}$ (further note that $\dim(\mathcal{H}_{\mathrm{anc}})\leq \dim(\mathcal{H}_\mathrm{A}[1:r\zeta])\dim(\mathcal{H}_\mathrm{B}[\overline{\chi}+1:d_2])$ so we choose $k=1$ in \cref{lemma-3190121}).

\subsection{Type II instance: $d_1< rd_2$ with $d_1+r\leq rd_2$, $d_1<r$}\label{sec-3260100}

\subsubsection{Construction}
In this subsection, we will use the definition in \cref{sec-1110250}, where the parameter $d, D$ in \cref{sec-1110250} correspond to $d_1$ and $rd_2$ here.

Suppose $d_2>1$.
We further assume $d_1+r\leq rd_2$ and $d_1<r$. 
Let $\chi=\lceil \frac{r}{d_1}\rceil$. Thus $\chi\geq 1$ and $\chi d_1\geq r$. Moreover, since $d_1d_2\geq 2r$ (see \cref{eq-3200605}), we have $\chi\leq \lceil\frac{d_2}{2}\rceil$ and thus $d_2-\chi\geq \lfloor \frac{d_2}{2}\rfloor$.

Note that $\chi d_1\geq r$. Thus by \cref{lemma-3180158}, there exists a set of linear operators $\{K_i\}_{i=1}^r$ for $K_i:\mathcal{H}_{\mathrm{A}}\rightarrow\mathcal{H}_{\mathrm{B}}[1:\chi]$ such that
\begin{equation}\label{eq-3200358}
\sum_{i=1}^r K_i^\dag K_i=(1-\varepsilon^2)I_{\mathcal{H}_\mathrm{A}} \quad \textup{and}\quad \left|\tr\!\left(K_i^\dag K_j\right)\right|\leq  \frac{2\dim(\mathcal{H}_\mathrm{A})}{r}\cdot\mathbbm{1}_{i=j}= \frac{2d_1}{r}\cdot\mathbbm{1}_{i=j}.
\end{equation}
Then we define $V_0:\mathcal{H}_{\mathrm{A}}\rightarrow  \mathcal{H}_\mathrm{B}[1:\chi]\otimes \mathcal{H}_\mathrm{anc}$ as 
\begin{equation*}
V_0= \sum_{i=1}^{r} \ket{i}_\mathrm{anc}\otimes K_i.
\end{equation*}
Note that $V_0$ can also be written as $V_0=\sqrt{1-\varepsilon^2}V$ for some isometry $V:\mathcal{H}_{\mathrm{A}}\rightarrow  \mathcal{H}_\mathrm{B}[1:\chi]\otimes \mathcal{H}_\mathrm{anc}$.
Define $\Delta:\mathcal{H}_{\mathrm{A}}\rightarrow \mathcal{H}_{\mathrm{B}}[\chi+1:d_2]\otimes \mathcal{H}_\mathrm{anc}$ be an arbitrary (but fixed) isometry.
Then, for $\varepsilon\in(0,1/2)$ and $U\in\mathbb{U}_{r(d_2-\chi)}$, we define the isometry $V_{\varepsilon,U}:\mathcal{H}_\mathrm{A}\rightarrow\mathcal{H}_\mathrm{B}\otimes\mathcal{H}_{\mathrm{anc}}$ as
\begin{equation}\label{eq-3200308}
V_{\varepsilon,U}\coloneqq V_0+\varepsilon U\Delta,
\end{equation}
where $U$ acts on $\mathcal{H}_{\mathrm{B}}[\chi+1:d_2]\otimes\mathcal{H}_\mathrm{anc}$.
We can easily verify that $V_{\varepsilon,U}$ is indeed an isometry.
We also note that the image of $U\Delta$ is orthogonal to the image of $V_0$. 
For clarity, we illustrate our construction in \cref{fig-3290258}.

\begin{figure}[t]
    \centering
    \includegraphics[width=0.55\linewidth]{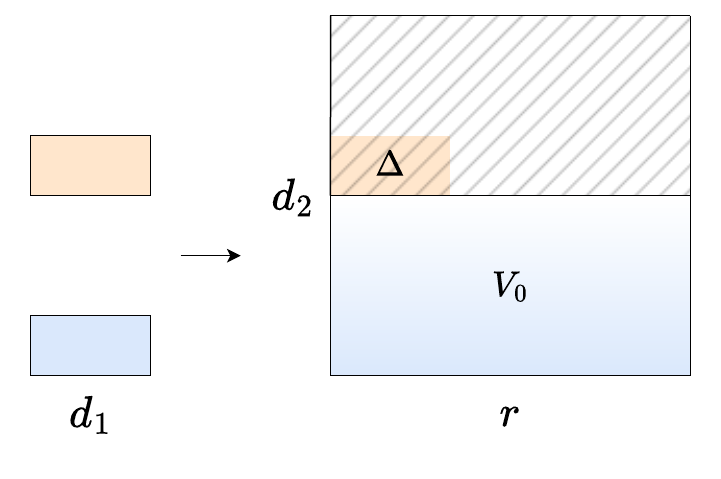}
    \caption{Illustration of our construction. We define linear operators $V_0$ and $\Delta$ from the $d_1$-dimensional space $\mathcal{H}_\mathrm{A}$ to the $d_2r$-dimensional space $\mathcal{H}_\mathrm{B}\otimes\mathcal{H}_\mathrm{anc}$. The Haar randomness is applied on the hatched area. The gradient-colored area means the image of $V_0$ is approximately ``uniformly'' distributed along the $r$-axis, using \cref{lemma-3180158}.}
    \label{fig-3290258}
\end{figure}

\subsubsection{Existence}

Then, we prove that there exists a large set of isometries $V_{\varepsilon,U}$ with good separation property.
\begin{theorem}\label{thm-3200309}
Suppose $\varepsilon\leq 1/4000$. There exists a finite subset $\mathcal{N}$ of $\{V_{\varepsilon,U}\,|\, U\in\mathbb{U}_{r(d_2-\chi)}\}$ for $V_{\varepsilon,U}$ defined in \cref{eq-3200308} with cardinality $|\mathcal{N}|\geq \exp(rd_1d_2/307201)$, such that for any $V_x\neq V_y\in\mathcal{N}$, if we set $\mathcal{E}_x=\tr_{\mathcal{H}_\mathrm{anc}}(V_x(\cdot)V_x^\dag)$  and $\mathcal{E}_y=\tr_{\mathcal{H}_\mathrm{anc}}(V_y(\cdot)V_y^\dag)$, then
\[\frac{1}{d_1}\|C_{\mathcal{E}_x}-C_{\mathcal{E}_y}\|_1\geq \frac{1}{2000}\varepsilon.\]
\end{theorem}
\begin{proof}
First, we need the following lemma.
\begin{lemma}\label{lemma-3200315}
    There exists a finite subset $\mathcal{M}\subseteq \mathbb{U}_{r(d_2-\chi)}$ with cardinality $|\mathcal{M}|\geq \exp(rd_1d_2/307201)$ such that for any $U_x\neq U_y\in\mathcal{M}$, 
\begin{equation}\label{eq-3200315}
\frac{1}{d_1}\left\|\tr_{\mathcal{H}_\mathrm{anc}}\!\left(\kett{V_0}\Big(\bbra{U_x\Delta}-\bbra{U_y\Delta}\Big)\right)\right\|_1\geq \frac{1}{2000}.
\end{equation}
\end{lemma}
The proof of \cref{lemma-3200315} is deferred. 

Now we are able to prove \cref{thm-3200309}.
Let $\mathcal{M}$ be the set given in \cref{lemma-3200315}.
Then, for any $U_x\neq U_y\in\mathcal{M}$, if set $\mathcal{E}_x=\tr_{\mathcal{H}_{\mathrm{anc}}}(V_{\varepsilon,U_x}(\cdot) V_{\varepsilon,U_x}^\dag)$ and $\mathcal{E}_y=\tr_{\mathcal{H}_{\mathrm{anc}}}(V_{\varepsilon,U_y}(\cdot) V_{\varepsilon,U_y}^\dag)$, we have
\begin{align}
&\|C_{\mathcal{E}_{x}}-C_{\mathcal{E}_{y}}\|_1 \nonumber \\
=&\Big\|\tr_{\mathcal{H}_{\mathrm{anc}}}\!\left(\kettbbra{V_{\varepsilon,U_x}}{V_{\varepsilon,U_x}}\right)-\tr_{\mathcal{H}_{\mathrm{anc}}}\!\left(\kettbbra{V_{\varepsilon,U_y}}{V_{\varepsilon,U_y}}\right)\Big\|_1 \nonumber\\
=&\bigg\|\varepsilon^2\Big(\tr_{\mathcal{H}_{\mathrm{anc}}}\!\left(\kettbbra{U_x\Delta}{U_x\Delta}\right)-\tr_{\mathcal{H}_{\mathrm{anc}}}\!\left(\kettbbra{U_y\Delta}{U_y\Delta}\right)\Big) \nonumber \\
&\qquad\qquad\qquad\qquad+\varepsilon\tr_{\mathcal{H}_{\mathrm{anc}}}\!\Big(\kettbbra{V_0}{(U_x-U_y)\Delta}\Big) + \varepsilon\tr_{\mathcal{H}_{\mathrm{anc}}}\!\Big(\kettbbra{(U_x-U_y)\Delta}{V_0}\Big) \bigg\|_1 \nonumber \\
\geq & \varepsilon\bigg\|\tr_{\mathcal{H}_{\mathrm{anc}}}\!\Big(\kettbbra{V_0}{(U_x-U_y)\Delta}\Big) + \tr_{\mathcal{H}_{\mathrm{anc}}}\!\Big(\kettbbra{(U_x-U_y)\Delta}{V_0}\Big)\bigg\|_1 - 2 \varepsilon^2 d_1 \label{eq-3200317} \\
= & 2\varepsilon\bigg\|\tr_{\mathcal{H}_{\mathrm{anc}}}\!\Big(\kettbbra{V_0}{(U_x-U_y)\Delta}\Big)\bigg\|_1-2\varepsilon^2 d_1 \label{eq-3200318}\\
\geq & \frac{1}{1000}\varepsilon d_1 -2\varepsilon^2d_1 \label{eq-3200319} \\
\geq & \frac{1}{2000}\varepsilon d_1. \label{eq-3200320}
\end{align}
In \cref{eq-3200317}, we used $\|\tr_{\mathcal{H}_{\mathrm{anc}}}(\kettbbra{U\Delta}{U\Delta})\|_1 = \tr(\kettbbra{U\Delta}{U\Delta})=d_1$. 
In \cref{eq-3200318} we used the fact that $\tr_{\mathcal{H}_{\mathrm{anc}}}\!\Big(\kettbbra{V_0}{(U_x-U_y)\Delta}\Big)$ is a linear operator with support in  $\mathcal{\mathcal{H}_\mathrm{A}}\otimes\mathcal{H}_{\mathrm{B}}[\chi+1:d_2]$ and image in $\mathcal{H}_\mathrm{A}\otimes\mathcal{H}_{\mathrm{B}}[1:\chi]$ 
which is orthogonal to its support; and then we used \cref{fact-3191602}.
In \cref{eq-3200319} we used \cref{eq-3200315}.
In \cref{eq-3200320} we used that $\varepsilon\leq 1/4000$.
Therefore, we can lower bound the Choi trace norm 
\[\frac{1}{d_1}\left\|C_{\mathcal{E}_{x}}-C_{\mathcal{E}_{y}}\right\|_1\geq \frac{1}{2000}\varepsilon.\]
Thus, the set $\mathcal{N}=\{V_{\varepsilon,U} \,|\, U\in\mathcal{M}\}$ is the desired set.
\end{proof}

Now, we give the proof of \cref{lemma-3200315}.
\begin{proof}[Proof of \cref{lemma-3200315}]
First, we need the following lemma:
\begin{lemma}\label{lemma-3200327}
For $U_x,U_y\in\mathbb{U}_{r(d_2-\chi)}$, let us define 
\[F(U_x,U_y)=\frac{1}{d_1}\tr_{\mathcal{H}_\mathrm{anc}}\!\left(\kett{V}\Big(\bbra{U_x\Delta}-\bbra{U_y\Delta}\Big)\right),\]
then the function $f(U_x,U_y)=\|F(U_x,U_y)\|_1=\tr(|F(U_x,U_y)|)$ is $\sqrt{\frac{2}{d_1}}$-Lipschitz with respect to the $\ell_2$-sum of the $2$-norms (Frobenius norm). 
Furthermore, for independent random $U_x,U_y\sim \mathbb{U}_{r(d_2-\chi)}$, we have $\E\!\left[\tr(|F(U_x,U_y)|^2)\right]\geq\frac{1}{r}$, and $\E\!\left[\tr(|F(U_x,U_y)|^4)\right]\leq \frac{384}{r^3}$.
\end{lemma}
The proof of \cref{lemma-3200327} is deferred to \cref{sec-3211843}.

By the H\"older's inequality we have
\[\E\!\left[\tr\!\left(|F(U_x,U_y)|^2\right)\right]\leq \E\!\left[\tr\!\left(|F(U_x,U_y)|^4\right)\right]^{1/3}\E\!\left[\tr\!\left(|F(U_x,U_y)|\right)\right]^{2/3},\]
which, combined with \cref{lemma-3200327}, implies
\[\E\!\left[\tr\!\left(|F(U_x,U_y)|\right)\right]^2\geq \frac{1}{384}.\]
Thus $\E\!\left[\tr\!\left(|F(U_x,U_y)|\right)\right]> 1/20$.
Then, since the function $f(U_x,U_y)=\tr(|F(U_x,U_y)|)$ is $\sqrt{\frac{2}{d_1}}$-Lipschitz, we can use \cref{lem:corollary17} to prove the concentration result:
\[\Pr\!\left[\tr\!\left(|F(U_x,U_y)|\right)\leq \frac{1}{2000}\right]\leq \exp\!\left(-\frac{d_1r(d_2-\chi)}{2}\cdot \frac{1}{1600\cdot 12}\right)\leq \exp\!\left(-\frac{rd_1d_2}{153600}\right),\]
where we used $d_2-\chi\geq \lfloor d_2/2 \rfloor\geq d_2/4$.
Then, we independently sample $\exp(rd_1d_2/307201)$ Haar random unitaries in $\mathbb{U}_{r(d_2-\chi)}$ and the union bound shows that there exists a non-zero probability that for any pair $U_x, U_y$, we have $\tr(|F(U_x,U_y)|)\geq 1/2000$.
Thus, there exists a set with cardinality $\geq \exp(rd_1d_2/307201)$ such that \cref{eq-3200315} holds.
\end{proof}

\subsubsection{Second and fourth moments: proof of Lemma \ref{lemma-3200327}}\label{sec-3211843}
The Lipschitz continuity can be seen from \cref{lemma-3192212}.

Define $K_{x,i}=\bra{i}_\mathrm{anc} U_x\Delta$ and $K_{y,i}=\bra{i}_\mathrm{anc} U_y\Delta$. Thus
$F(U_x,U_y)=\frac{1}{d_1}\sum_{i=1}^r\kett{K_i}\Big(\bbra{K_{x,i}}-\bbra{K_{y,i}}\Big)$.
Then, we note that
\begin{align}
\E\!\left[\tr(|F(U_x,U_y)|^2)\right]&=\frac{1}{d_1^2}\E\!\left[\tr\!\left(\sum_{i,j=1}^r \kett{K_i}\Big(\bbra{K_{x,i}}-\bbra{K_{y,i}}\Big)\Big(\kett{K_{x,j}}-\kett{K_{y,j}}\Big)\bbra{K_j}\right)\right] \nonumber\\
&=\frac{1}{d_1^2}\E\!\left[\sum_{i=1}^r \bbrakett{K_i}{K_i}\Big(\bbra{K_{x,i}}-\bbra{K_{y,i}}\Big)\Big(\kett{K_{x,i}}-\kett{K_{y,i}}\Big)\right]\nonumber \\
&=\frac{2}{rd_1}\sum_{i=1}^r\bbrakett{K_i}{K_i}  \label{eq-3200338} \\
&\geq\frac{1}{r},\label{eq-3200339}
\end{align}
where \cref{eq-3200338} is because for $z_1,z_2\in\{x,y\}$, we have
\begin{align}
\E\!\left[\bbrakett{K_{z_1,i}}{K_{z_2,i}}\right]&=\E\!\left[\tr\!\left(K_{z_1,i}^\dag K_{z_2,i}\right)\right] =\tr\!\left(\Delta^\dag \E\!\left[U_{z_1}^\dag \ketbra{i}{i}_\mathrm{anc}U_{z_2} \right]\Delta\right) \nonumber\\
&=\mathbbm{1}_{z_1=z_2}\cdot \frac{1}{r}  \tr(\Delta^\dag \Delta) \label{eq-3200341}\\
&=\mathbbm{1}_{z_1=z_2}\cdot \frac{d_1}{r}, \nonumber
\end{align}
where \cref{eq-3200341} is due to Schur's lemma, 
\cref{eq-3200339} is because $\sum_{i=1}^r\tr(K_i^\dag K_i)=\tr(V_0^\dag V_0)\geq \frac{1}{2}\tr(I_{\mathcal{H}_\mathrm{A}})=d_1/2$.

Then, we also have
\begin{align}
\E\!\left[\tr(|F(U_x,U_y)|^4)\right]&=\frac{1}{d_1^4}\E\!\Bigg[\sum_{i,j,k,l=1}^r \tr\!\bigg(\kett{K_i}\Big(\bbra{K_{x,i}}-\bbra{K_{y,i}}\Big)\Big(\kett{K_{x,j}}-\kett{K_{y,j}}\Big)\bbra{K_j} \nonumber\\
&\qquad\qquad\qquad\qquad\kett{K_{k}}\Big(\bbra{K_{x,k}}-\bbra{K_{y,k}}\Big)\Big(\kett{K_{x,l}}-\kett{K_{y,l}}\Big)\bbra{K_{l}}\bigg)\Bigg]  \nonumber \\ 
&\leq \frac{4}{r^2d_1^2}\sum_{i,j=1}^r\E\!\bigg[\left|\Big(\bbra{K_{x,i}}-\bbra{K_{y,i}}\Big)\Big(\kett{K_{x,j}}-\kett{K_{y,j}}\Big)\right|^2\bigg]\label{eq-3200346} \\
&\leq \frac{16}{r^2 d_1^2}\sum_{i,j=1}^r\E\!\left[|\bbrakett{K_{x,i}}{K_{x,j}}|^2+|\bbrakett{K_{y,i}}{K_{y,j}}|^2+|\bbrakett{K_{x,i}}{K_{y,j}}|^2+|\bbrakett{K_{y,i}}{K_{x,j}}|^2\right] \nonumber\\
&\leq \frac{32}{r^2 d_1^2}\sum_{i,j=1}^r\E\!\left[|\bbrakett{K_{x,i}}{K_{x,j}}|^2+|\bbrakett{K_{y,i}}{K_{y,j}}|^2\right]\label{eq-3200347}\\
&=\frac{64}{r^2 d_1^2}\sum_{i,j=1}^r \E\!\left[|\bbrakett{K_{x,i}}{K_{x,j}}|^2\right]=\frac{64}{ r^2 d_1^2}\sum_{i,j=1}^r\E\!\left[\left|\tr(K_{x,i}^\dag K_{x,j})\right|^2\right],\nonumber\\
&\leq \frac{384}{r^3},\label{eq-3200348}
\end{align}
where \cref{eq-3200346} is due to \cref{eq-3200358}, \cref{eq-3200347} due to the similar argument as that in \cref{eq-3200006}, 
and in \cref{eq-3200348}, the first inequality is due to \cref{lemma-3190121} and that $\Delta$ is an isometry from $\mathcal{H}_\mathrm{A}$ to $\mathcal{H}_\mathrm{B}[\chi+1:d_2]\otimes\mathcal{H}_{\mathrm{anc}}$ (further note that $\dim(\mathcal{H}_{\mathrm{anc}})=r\leq d_1d_2/2\leq 2d_1\lfloor \frac{d_2}{2}\rfloor\leq 2d_1(d_2-\chi)=2\dim(\mathcal{H}_\mathrm{A}) \dim(\mathcal{H}_{\mathrm{B}}[\chi+1:d_2])$ so we choose $k=2$ in \cref{lemma-3190121}).

\section{Query lower bounds}\label{sec-3200555}
Now, we combine the results from \cref{sec-3200450} and \cref{sec-3290355} to give the lower bounds for quantum channel tomography.

\begin{theorem}\label{thm-390442}
Let $d_1,d_2,r$ be positive integers such that $d_1\leq rd_2\leq \frac{4}{3}d_1$. 
Suppose $\varepsilon\leq \frac{1}{39660}$ and $rd_1d_2\geq 7\cdot 10^{25}$. Then, tomography of an unknown channel $\mathcal{E}\in\qchannel_{d_1,d_2}^r$ to within Choi trace norm or diamond norm error $\frac{\varepsilon}{3966000}$ with probability at least $2/3$ must use at least $\frac{1}{2\cdot 10^{19}}\cdot \frac{d_1^2}{\varepsilon}$ 
queries to $\mathcal{E}$.
\end{theorem}
\begin{proof}
This can be obtained by combining \cref{thm-1111809} with \cref{thm:existence-d1=rd2}. 
Specifically, a tomography algorithm to within Choi trace norm error $\frac{\varepsilon}{2\cdot 1983000}=\frac{\varepsilon}{3966000}$ can discriminate between the quantum channels constructed by \cref{thm:existence-d1=rd2}.
Therefore, this algorithm can be used to discriminate between the corresponding Stinespring dilation isometries (simply discard the ancilla system and then apply this algorithm).
Note that these isometries coincide with those defined in \cref{sec:hard-d1=rd2} (up to a change of basis).
Furthermore, note that $\frac{4}{3}d_1^2\geq rd_1d_2\geq 7\cdot 10^{25}$ implies $d_1\geq 7\cdot 10^{12}$. 
Therefore, \cref{thm-1111809} applies and shows that any such algorithm must use at least $\frac{1}{2\cdot 10^{19}}\cdot \frac{d_1^2}{\varepsilon}$ queries.
Since diamond norm is always larger than or equal to the Choi trace norm, this lower bound also applies to diamond norm tomography algorithms.
\end{proof}

\begin{theorem}\label{thm-390512}
Let $c\in(0,1]$ be constant, and $d_1,d_2,r$ be positive integers such that $(1+c)d_1\le rd_2$.
Suppose $\varepsilon\leq \frac{c^{3/2}}{25600}$ and $rd_1d_2\geq \frac{2\cdot 10^{30}}{c^{10}}$.
Then, tomography of an unknown channel $\mathcal{E}\in\qchannel_{d_1,d_2}^r$ to within Choi trace norm or diamond norm error $\frac{c^{3/2}}{25600}\varepsilon$ with probability at least $2/3$ must use at least $\frac{c^6}{10^{16}}\cdot \frac{rd_1d_2}{\varepsilon^2}$ queries to $\mathcal{E}$.
\end{theorem}
\begin{proof}
This can be obtained by combining \cref{thm-1110122} with \cref{thm-1140332}, \cref{thm-3191542} and \cref{thm-3200309}.
Specifically, the regime $d_1<rd_2$ can be divided into three regimes: 
\begin{enumerate}
  \renewcommand{\labelenumi}{(\roman{enumi})}
    \item $rd_2< d_1+r$,
    \item $d_1+r\leq rd_2$ and $r\leq d_1$,
    \item $d_1+r\leq rd_2$ and $d_1<r$,
\end{enumerate}
which are covered by \cref{thm-1140332}, \cref{thm-3191542} and \cref{thm-3200309}, respectively.
We define $\sr=\min\{\frac{rd_2-d_1}{d_1},1\}$. By assumption, we know that $\sr\geq c$. Note that $rd_2\leq \frac{2}{c}(rd_2-d_1)$, $d_1\leq \frac{1}{c} (rd_2-d_1)$ and $rd_1d_2\leq \frac{2}{c^2} (rd_2-d_1)^2$.
Since $rd_1d_2\geq \frac{2\cdot 10^{30}}{c^{10}}$, we know $rd_2-d_1\geq \frac{10^{15}}{c^4}$.

In regime (i), first note that $2r \leq rd_2< d_1+r$ (since $d_2=1$ is trivial), which means $d_1> r$. Thus $rd_2-d_1 \leq r <d_1$.
A tomography algorithm to within Choi trace norm error $\frac{c^{3/2}}{9200}\varepsilon$ can discriminate between $\exp(\frac{c^2}{480001}(rd_2-d_1)^2)$ quantum channels constructed by \cref{thm-1140332}.
Therefore, this algorithm can be used to discriminate between the corresponding Stinespring dilation isometries (simply discard the ancilla system and then apply this algorithm).
Note that these isometries coincide with those defined in \cref{sec-1110250} (up to a change of basis).
Furthermore, note that $rd_2-d_1\geq \frac{2\cdot 10^{13}}{c^4}$.
Therefore, \cref{thm-1110122} applies and shows that any such algorithm must use at least
\[\frac{c^4}{8\cdot 10^{12}}\cdot \frac{(rd_2-d_1)^2}{\varepsilon^2}\geq \frac{c^6}{16\cdot 10^{12}}\cdot \frac{rd_2d_1}{\varepsilon^2}\]
queries.

In regime (ii), we define $d'=\min\{d_1,rd_2-d_1\}$. 
A tomography algorithm to within Choi trace norm error $\frac{c^{3/2}}{25600}\varepsilon$ can discriminate between $\exp(\frac{c^2}{3840001}(rd_2-d_1)d')$ quantum channels constructed by \cref{thm-3191542}.
Therefore, this algorithm can be used to discriminate between the corresponding Stinespring dilation isometries (simply discard the ancilla system and then apply this algorithm).
Note that these isometries coincide with those defined in \cref{sec-1110250} (up to a change of basis). 
Furthermore, note that $rd_2-d_1\geq \frac{10^{15}}{c^4}$.
Therefore, \cref{thm-1110122} applies and shows that any such algorithm must use at least
\[\frac{c^4}{5\cdot 10^{14}}\cdot \frac{(rd_2-d_1)d'}{\varepsilon^2}\geq \frac{c^6}{10^{16}}\cdot \frac{rd_1d_2}{\varepsilon^2}\]
queries.

In regime (iii), note that $rd_2-d_1\geq r>d_1$. Furthermore, since $rd_2\geq 2r>2d_1$ (since $d_2=1$ is trivial), we have $rd_2-d_1\geq rd_2/2$.
A tomography algorithm to within Choi trace norm error $\frac{1}{4000}\varepsilon$ can discriminate between $\exp(\frac{1}{307201}rd_1d_2)\geq \exp(\frac{1}{307201}(rd_2-d_1)d_1)$ quantum channels constructed by \cref{thm-3200309}.
Therefore, this algorithm can be used to discriminate between the corresponding Stinespring dilation isometries (simply discard the ancilla system and then apply this algorithm).
Note that these isometries coincide with those defined in \cref{sec-1110250} (up to a change of basis).
Furthermore, note that $rd_2-d_1\geq 7\cdot 10^{12}$.
Therefore, \cref{thm-1110122} applies and shows that any such algorithm must use at least
\[\frac{1}{4\cdot 10^{12}}\cdot \frac{(rd_2-d_1)d_1}{\varepsilon^2}\geq \frac{1}{8\cdot 10^{12}}\cdot \frac{rd_1d_2}{\varepsilon^2}\]
queries.

Then, combining the results in these three regimes, we obtain the desired lower bound.
Moreover, the Choi trace norm lower bound directly implies the same lower bound for diamond norm.
\end{proof}

\begin{theorem}\label{thm-390511}
Let $d_1,d_2,r$ be positive integers such that $d_1 < rd_2 \le  2d_1$.
Suppose $\varepsilon\leq \frac{1}{160}$ and $rd_2-d_1\geq 64\cdot 38404^2$.
Then, tomography of an unknown channel $\mathcal{E}\in\qchannel_{d_1,d_2}^r$ to within diamond norm error $\frac{\varepsilon}{160}$ with probability at least $2/3$ must use at least $\frac{1}{32\cdot 38404^2}\cdot \frac{(rd_2-d_1)^2}{\varepsilon^2}$ queries to $\mathcal{E}$.
\end{theorem}
\begin{proof}
This can be easily obtained by combining \cref{thm-1110122} with \cref{thm-3231724} and \cref{thm-3250101}.
Specifically, the regime $d_1<rd_2$ can be divided into three regimes: (i) $rd_2\leq d_1+r$, (ii) $rd_2\geq d_1+r$ and $d_1\geq r$, (iii) $rd_2\geq d_1+r$ and $d_1<r$. Since the case (iii) implies $rd_2> 2d_1$, so we know that the regime $d_1<rd_2\leq 2d_1$ is in the union of regimes (i) and (ii), which are covered by \cref{thm-3231724} and \cref{thm-3250101}, respectively.

In regime (i), a tomography algorithm to within diamond norm error $\varepsilon/160$ can discriminate between $\exp((rd_2-d_1)^2/4801)$ quantum channels constructed by \cref{thm-3231724}.
Therefore, this algorithm can be used to discriminate between the corresponding Stinespring dilation isometries (simply discard the ancilla system and then apply this algorithm).
Note that these isometries coincide with those defined in \cref{sec-1110250} (up to a change of basis).
Therefore, \cref{thm-1110122} applies and shows that any such algorithm must use at least $\frac{1}{32\cdot 4801^2}\cdot \frac{(rd_2-d_1)^2}{\varepsilon^2}$ queries (note that $rd_2-d_1\leq d_1$).

In regime (ii), we use a similar argument. Note that $rd_2-d_1\leq d_1$. Thus \cref{thm-3250101} provides $\exp((rd_2-d_1)^2/38404)$ quantum channels, which also have Stinespring dilation isometries that coincide with those defined in \cref{sec-1110250} (up to a change of basis).
Therefore, \cref{thm-1110122} applies and shows that any such algorithm must use at least $\frac{1}{32\cdot 38404^2}\cdot \frac{(rd_2-d_1)^2}{\varepsilon^2}$ queries.

Then, combining the results in these two regimes, we obtain the desired lower bound.
\end{proof}

\begin{corollary}\label{coro-3270443}
Let $d_1,d_2,r$ be positive integers such that $d_1 < rd_2 \le  2d_1$.
Suppose $\varepsilon\leq \frac{1}{160d_1}$ and $rd_2-d_1\geq 64\cdot 38404^2$.
Then, tomography of an unknown channel $\mathcal{E}\in\qchannel_{d_1,d_2}^r$ to within Choi trace norm error $\frac{\varepsilon}{160}$ with probability at least $2/3$ must use at least $\frac{1}{32\cdot 38404^2}\cdot \frac{(rd_2/d_1-1)^2}{\varepsilon^2}$ queries to $\mathcal{E}$.
\end{corollary}
\begin{proof}
Let $\varepsilon'=d_1 \cdot \varepsilon$.
By \cref{eq-3302201}, a tomography algorithm to within Choi trace norm error $\frac{\varepsilon}{160}$ is also a tomography algorithm to within diamond norm error $\frac{1}{160}\varepsilon'$. Note that $\varepsilon'\leq \frac{1}{160}$. Thus \cref{thm-390511} applies and shows that this algorithm must use at least $\frac{1}{32\cdot 38404^2}\cdot\frac{(rd_2-d_1)^2}{\varepsilon'^2}=\frac{1}{32\cdot 38404^2}\cdot\frac{(rd_2/d_1-1)^2}{\varepsilon^2}$ queries.
\end{proof}

\section{Deferred lemmas}\label{sec-770141}

\subsection{Technical lemmas}

\begin{lemma}\label{lemma-12271509}
Suppose $G$ is a compact Lie group and $\rho(\cdot)$ is an action of $G$ on a finite-dimensional Hilbert space $\mathcal{H}$.
Suppose $X\in \mathcal{L}(\mathcal{H})$ is positive semidefinite. Then,
\[\tr\left(\left(\EE{g\sim G}\left[\rho(g) X\rho(g)^{-1}\right]\right)^{-1} X\right)\leq \dim(\mathcal{H}),\]
where $\EE{g\sim G}$ denotes the expectation over the Haar measure of $G$ and $(\cdot)^{-1}$ denotes pseudo-inverse.
\end{lemma}
\begin{proof}
Define $\overline{X}=\EE{g\sim G}\left[\rho(g) X\rho(g)^{-1}\right]$. Then
\begin{align}
\tr\left(\overline{X}^{-1} X\right)
&= \E_{g\sim G}\!\left[\tr\!\left(\rho(g) \overline{X}^{-1}  X\rho(g)^{-1}\right)\right] \nonumber \\
&= \E_{g\sim G}\!\left[\tr\!\left( \overline{X}^{-1} \rho(g) X\rho(g)^{-1}\right)\right] \nonumber\\
&= \tr\left(\overline{X}^{-1} \overline{X}\right) \nonumber \\
& \leq \dim(\mathcal{H}),\nonumber 
\end{align}
where we used that $\overline{X}$ commutes with $\rho(g)$.
\end{proof}

\begin{lemma}\label{lemma-12271712}
Suppose $n, m$ are positive integers such that $n\geq m$. For $i\in[n]$, let $\mathcal{H}_i\cong \mathbb{C}^{d}$ be a Hilbert space. 
We use $\{\ket{0},\ldots,\ket{d-1}\}$ to denote an orthonormal basis of $\mathbb{C}^d$.
We define the following subspace of $\bigotimes_{i=1}^n \mathcal{H}_i$:
\[A=\spanspace\left(\left\{\sum_{\substack{S\subseteq [n]\\ |S|=m }}\ket{\psi}^{\otimes S}\otimes \ket{0}^{\otimes [n]\setminus S}\,\, \bigg|\,\, \ket{\psi}\in \Pi \right\}\right),\]
where $\Pi\subseteq \spanspace\{\ket{1},\ldots,\ket{d-1}\}$ is a subspace orthogonal to $\ket{0}$.
Then, 
\[\dim(A)=\binom{\dim(\Pi)+m-1}{m}.\]
\end{lemma}
\begin{proof}
We define
\[P=\sum_{\pi\in\mathfrak{S}_{n}} \texttt{p}(\pi),\]
where $\texttt{p}(\cdot)$ denotes the action of $\mathfrak{S}_n$ on $\bigotimes_{i=1}^n\mathcal{H}_i$ (that is, for $\pi\in\mathfrak{S}_n$, $\texttt{p}(\pi)\ket{\psi_1}\otimes\cdots\otimes\ket{\psi_{n}}=\ket{\psi_{\pi^{-1}(1)}}\otimes\cdots\otimes\ket{\psi_{\pi^{-1}(n)}}$).
We can see that, when restricting to the subspace $\spanspace(\{\ket{\psi}^{\otimes m}\otimes \ket{0}^{\otimes n-m} \,\, | \,\, \ket{\psi}\in \Pi\})$, $P$ is injective and $A$ is exactly the image of $P$ on this subspace. 
Therefore, $A$ has the same dimension as $\spanspace(\{\ket{\psi}^{\otimes m}\,\,|\,\, \ket{\psi}\in\Pi\})\cong \lor^m \Pi$. The latter is the symmetric subspace of $\Pi^{\otimes m}$, which is of dimension $\binom{\dim(\Pi)+m-1}{m}$~\cite{harrow2013church}.
\end{proof}

\begin{lemma}\label{lemma-3192212}
Let $\mathcal{H}_1$, $\mathcal{H}_2$ and $\mathcal{H}_3$ be Hilbert spaces with dimension $d_1,d_2$ and $d_3$, respectively.
Let $U_x,U_y\in\mathbb{U}_{d_2d_3}$ be unitaries acting on $\mathcal{H}_2\otimes\mathcal{H}_3$. Let $V,W:\mathcal{H}_1\rightarrow\mathcal{H}_2\otimes\mathcal{H}_3$ be sub-isometries (i.e., $V^\dag V,\, W^\dag W \leq I_{\mathcal{H}_1}$). Then, the following function is $\sqrt{\frac{2}{d_1}}$-Lipschitz
\[f(U_x,U_y)=\frac{1}{d_1}\Big\|\tr_{\mathcal{H}_3}\!\left(\kett{V}\big(\bbra{U_x W}-\bbra{U_y W}\big)\right)\Big\|_1,\]
i.e., $|f(U_x,U_y)-f(U'_x,U'_y)|\leq\sqrt{\frac{2}{d_1}}\sqrt{\|U_x-U_x'\|_F^2+\|U_y-U_y'\|_F^2}$. Here $\|\cdot\|_1$ denotes the trace norm and $\|\cdot\|_F$ denotes the Frobenius norm.
\end{lemma}
\begin{proof}
First, note that
\begin{align}
\Big\|\tr_{\mathcal{H}_3}\!\left(\kett{V}\big(\bbra{U_x W}-\bbra{U_y W}\big)\right)\Big\|_1&\leq \Big\|\kett{V}(\bbra{U_xW}-\bbra{U_y W})\Big\|_1 \label{eq-3192049} \\
&= \big\|\kett{V}\big\| \cdot \big\|\kett{U_xW}-\kett{U_y W}\big\| \nonumber\\
&\leq \sqrt{d_1}\cdot \big\|\kett{U_xW}-\kett{U_y W}\big\|\label{eq-3192050} \\
&= \sqrt{d_1}\cdot \sqrt{\tr\!\left(W^\dag (U_x-U_y)^\dag (U_x-U_y)W\right)} \nonumber \\
&\leq \sqrt{d_1} \sqrt{\tr\!\left((U_x-U_y)^\dag (U_x-U_y)\right)}=\sqrt{d_1} \cdot \|U_x-U_y\|_F,\label{eq-3192059}
\end{align}
where \cref{eq-3192049} is because partial trace is contractive in trace norm, \cref{eq-3192050} is due to $V^\dag V\leq I_{H_1}$, \cref{eq-3192059} is because $WW^\dag \leq I_{H_2}\otimes I_{H_3}$ since $W^\dag W\leq I_{H_1}$.
Then, we can show that 
\begin{align}
&|f(U_x,U_y)-f(U'_x,U'_y)| \nonumber \\
\leq &\frac{1}{d_1}\Big\|\tr_{\mathcal{H}_3}\!\left(\kett{V}\big(\bbra{U_x W}-\bbra{U_y W}\big)\right)-\tr_{\mathcal{H}_3}\!\left(\kett{V}\big(\bbra{U'_x W}-\bbra{U'_y W}\big)\right)\Big\|_1 \nonumber\\
\leq&\frac{1}{d_1}\Big\|\tr_{\mathcal{H}_3}\!\left(\kett{V}\big(\bbra{U_x W}-\bbra{U_x' W}\big)\right)\Big\|_1+\Big\|\tr_{\mathcal{H}_3}\!\left(\kett{V}\big(\bbra{U_y W}-\bbra{U'_y W}\big)\right)\Big\|_1 \nonumber\\
\leq &\frac{1}{\sqrt{d_1}}(\|U_x-U_x'\|_F+\|U_y-U_y'\|_F) \label{eq-3192211}\\
\leq & \sqrt{\frac{2}{d_1}}\sqrt{\|U_x-U_x'\|_F^2+\|U_y-U_y'\|_F^2},\nonumber 
\end{align}
where \cref{eq-3192211} is due to \cref{eq-3192059}.
\end{proof}

\begin{lemma}\label{lemma-3190121}
Suppose $\mathcal{H}_\mathrm{A}$, $\mathcal{H}_\mathrm{B}$ and $\mathcal{H}_\mathrm{anc}$ are Hilbert spaces with dimensions $d_1$, $d_2$ and r, respectively. Further assume that $d_1/d_2\leq r\leq kd_1d_2$ for some $k\geq 1$.
Suppose $U\in\mathbb{U}_{rd_2}$ is a Haar-random unitary acting on $\mathcal{H}_\mathrm{B}\otimes\mathcal{H}_\mathrm{anc}$ and $\Delta:\mathcal{H}_\mathrm{A}\rightarrow\mathcal{H}_\mathrm{B}\otimes \mathcal{H}_\mathrm{anc}$ is an isometry.
Then
\[\sum_{i,j=1}^r \E\!\left[\left|\tr\!\left(\Delta^\dag U^\dag \ketbra{i}{j}_\mathrm{anc} U\Delta\right)\right|^2\right]\leq \frac{2(1+k)d_1^2}{r}.\]
\end{lemma}
\begin{proof}
For any $i,j\in[r]$,
\begin{align}
&\E\!\left[\left|\tr\!\left(\Delta^\dag U^\dag \ketbra{i}{j}_\mathrm{anc} U\Delta\right)\right|^2\right] \nonumber\\
=& \sum_{k,l=1}^{d_1} \E\!\left[\tr\!\left(U^\dag \ketbra{i}{j}_{\mathrm{anc}} U\Delta\ketbra{l}{k}_\mathrm{A}\Delta^\dag U^\dag\ketbra{j}{i}_\mathrm{anc}U \Delta\ketbra{k}{l}_\mathrm{A}\Delta^\dag\right)\right]\nonumber \\
=&\sum_{k,l=1}^{d_1}\frac{1}{r^2d_2^2-1}\bigg[\tr\!\left(\Delta \ketbra{l}{k}_\mathrm{A}\Delta^\dag\right)\cdot \tr\!\left(\Delta \ketbra{k}{l}_\mathrm{A}\Delta^\dag\right) \cdot \tr\big(\ketbra{i}{j}_\mathrm{anc}\otimes I_\mathrm{B} \cdot \ketbra{j}{i}_\mathrm{anc}\otimes I_\mathrm{B}\big) \nonumber \\
&\qquad\qquad\qquad \quad + \tr\!\left(\Delta \ketbra{l}{k}_\mathrm{A}\Delta^\dag\cdot \Delta \ketbra{k}{l}_\mathrm{A}\Delta^\dag \right)\cdot \tr\big(\ketbra{i}{j}_\mathrm{anc}\otimes I_\mathrm{B}\big)\cdot \tr\big(\ketbra{j}{i}_\mathrm{anc}\otimes I_\mathrm{B}\big) \bigg]\nonumber \\
&\quad- \frac{1}{rd_2(r^2d_2^2-1)}\bigg[
\tr\!\left(\Delta \ketbra{l}{k}_\mathrm{A}\Delta^\dag\right)\cdot \tr\!\left(\Delta \ketbra{k}{l}_\mathrm{A}\Delta^\dag\right) \cdot \tr\big(\ketbra{i}{j}_\mathrm{anc}\otimes I_\mathrm{B}\big)\cdot \tr\big(\ketbra{j}{i}_\mathrm{anc}\otimes I_\mathrm{B}\big) \nonumber \\
&\qquad\qquad\qquad\qquad \quad + \tr\!\left(\Delta \ketbra{l}{k}_\mathrm{A}\Delta^\dag\cdot \Delta \ketbra{k}{l}_\mathrm{A}\Delta^\dag \right) \cdot \tr\big(\ketbra{i}{j}_\mathrm{anc}\otimes I_\mathrm{B} \cdot \ketbra{j}{i}_\mathrm{anc}\otimes I_\mathrm{B}\big) \bigg] \label{eq-3190331} \\
=&\sum_{k,l=1}^{d_1} \left(\frac{1}{r^2d_2^2-1}(\mathbbm{1}_{k=l}d_2+\mathbbm{1}_{i=j}d_2^2)-\frac{1}{rd_2(r^2d_2^2-1)}(\mathbbm{1}_{k=l}\mathbbm{1}_{i=j}d_2^2+d_2)\right)\nonumber\\
=&\frac{d_1d_2+\mathbbm{1}_{i=j}d_1^2d_2^2}{r^2d_2^2-1}-\frac{\mathbbm{1}_{i=j}d_1d_2^2+d_1^2d_2}{r^3d_2^3-rd_2},\nonumber
\end{align}
where \cref{eq-3190331} is by \cref{coro-3221651}.
Hence,
\begin{align}
\sum_{i,j=1}^r \E\!\left[\left|\tr\!\left(\Delta^\dag U^\dag \ketbra{i}{j}_\mathrm{anc} U\Delta\right)\right|^2\right] &=\frac{r^2d_1d_2+rd_1^2d_2^2}{r^2d_2^2-1}-\frac{rd_1d_2^2+r^2d_1^2d_2}{r^3d_2^3-rd_2}\nonumber \\
&\leq \frac{r^2d_1d_2+rd_1^2d_2^2}{r^2d_2^2-1} \nonumber\\
&\leq 2\left(\frac{d_1}{d_2}+\frac{d_1^2}{r}\right)\nonumber \\
&\leq \frac{2(1+k)d_1^2}{r},\label{eq-3182254}
\end{align}
where \cref{eq-3182254} uses $r\leq kd_1d_2$.
\end{proof}

\begin{lemma}\label{lemma-3180158}
Let $d_1, d_2$ and $r$ be integers such that $d_1/d_2\leq r\leq d_1d_2$. 
Let $\mathcal{H}_{\mathrm{A}}$ be a Hilbert space of dimension $d_1$ and  $\mathcal{H}_{\mathrm{B}}$ be a Hilbert space of dimension $d_2$. 
Then, there exists a set of linear operators $\{K_i\}_{i=1}^r$ with $K_i:\mathcal{H}_\mathrm{A}\rightarrow\mathcal{H}_\mathrm{B}$ satisfying
\begin{itemize}
\item $\tr(K_i^\dag K_j)=0$ for any $i\neq j$,
\item $\tr(K_i^\dag K_i)\leq 2d_1/r$ for any $i\in [r]$,
\item $\sum_{i=1}^r K_i^\dag K_i=I_\mathrm{A}$.
\end{itemize}
\end{lemma}
\begin{proof}
We distinguish between two cases depending on whether $d_1 \le d_2$ or not. 

\textbf{Case $1$:}  $d_1 \le d_2$, let $k = \bigl\lfloor \frac{d_2}{d_1} \bigr\rfloor \ge 1$. 
We can write  $\mathcal{H}_{\mathrm{B}} \cong \bigl(\bigoplus_{i=1}^k \mathcal{H}_{\mathrm{A}_i}\bigr) \oplus \mathbb{C}^{d_3}$, 
with $\mathcal{H}_{\mathrm{A}_i}\cong \mathbb{C}^{d_1}$ for each $i$ and $d_3 = d_2 - kd_1 < d_1$.

For each $i\in [k]$, there exists  $l = d_1^2$ orthogonal 
$d_1 \times d_1$ unitary matrices $\{U_{i,j}\}_{j \in [l]}$ (one can take  the generalized Pauli operators). 
We may 
find a subset $S \subseteq[k] \times [l]$ with $|S| = \bigl\lceil \frac{r}{2}\bigr\rceil$ because 
$kl = \bigl\lfloor \frac{d_2}{d_1} \bigr\rfloor d_1^2 \ge \bigl\lceil\frac{d_1d_2}{2}\bigl\rceil\ge\bigl\lceil \frac{r}{2}\bigr\rceil $. For $(i,j) \in S$, we view $U_{i,j}:\mathcal{H}_\mathrm{A}\rightarrow\mathcal{H}_{\mathrm{A}_i}$ as a linear operator that maps from $\mathcal{H}_\mathrm{A}$ to the $i$-th block $\mathcal{H}_{\mathrm{A}_i}\subseteq \mathcal{H}_\mathrm{B}$, and we define the 
linear operator $K_{i,j}$ as
\[
K_{i,j} = \frac{1}{\sqrt{|S|}}\, U_{i,j}.
\]

We then check:
\begin{itemize}
     \item[(a)]  For all $(i,j), (i',j') \in S$,
    \[
    \left|\Tr\!\big( K_{i,j}^\dagger K_{i',j'} \big)\right|
    = \frac{d_1}{|S|} \, \mathbbm{1}_{i=i'} \mathbbm{1}_{j=j'} \le \frac{2d_1}{r} \, \mathbbm{1}_{i=i'} \mathbbm{1}_{j=j'}.
    \]
    \item[(b)] 
    \[
    \sum_{(i,j) \in S} K_{i,j}^\dagger K_{i,j}
    = \sum_{(i,j) \in S} \frac{1}{|S|} \, U_{i,j}^\dagger U_{i,j}
    = I_\mathrm{A}.
    \]
    \item[(c)] The total number of constructed operators is  $|S|\le r$.
\end{itemize}

\textbf{Case $2$:} $d_1 > d_2$, let $k = \bigl\lfloor \frac{d_1}{d_2} \bigr\rfloor \in [1, r]$ and write $d_1 = k d_2 + d_3$ with $0 \le d_3 < d_2$. We can then write $\mathcal{H}_\mathrm{A} = \left(\bigoplus_{i=1}^k \mathcal{H}_{\mathrm{B}_i} \right) \oplus \mathcal{H}_{\mathrm{C}}$, with each $\mathcal{H}_{\mathrm{B}_i} \cong \mathbb{C}^{d_2}$.

    For each $i\in [k]$, there exists a family $\{U_{i,j}\}_{j \in [l]}$ of  $l = \bigl\lceil \frac{r}{2k} \bigr\rceil \in [1, d_2^2]$ orthogonal $d_2 \times d_\mathrm{2}$ unitary matrices.
    For the block $\mathcal{H}_{\mathrm{B}_i}$ ($i=1,\dots,k$), we view   
    $U_{i,j}:\mathcal{H}_{\mathrm{B}_i}\rightarrow\mathcal{H}_\mathrm{B}$ as a  linear operator that acts nontrivially only on the $i$-th block $\mathcal{H}_{\mathrm{B}_i}\subseteq\mathcal{H}_\mathrm{A}$, and we 
    define the linear operator $K_{i,j}$ as
    \[K_{i,j} = \frac{1}{\sqrt{l}}\, U_{i,j}. \]
    
    For the remaining block $\mathcal{H}_\mathrm{C}$, since $d_3 < d_2$ we can use the construction similar to that in Case $1$.
    Specifically, we split $\mathcal{H}_{\mathrm{B}}$ into $\lfloor \frac{d_2}{d_3}\rfloor$ equal-dimension blocks (possibly leaving a remainder).
    For each $i\leq \lfloor \frac{d_2}{d_3}\rfloor$, there are $d_3^2$ orthogonal $d_3\times d_3$ unitary matrices that maps from $\mathcal{H}_\mathrm{C}$ to the $i$-th block in $\mathcal{H}_\mathrm{B}$.
    These unitaries can be viewed as orthogonal isometries from $\mathcal{H}_\mathrm{C}$ to $\mathcal{H}_\mathrm{B}$ and there are $\bigl\lfloor \frac{d_2}{d_3} \bigr\rfloor\cdot d_3^2$ such isometries in total.
    We choose $r' = \bigl\lceil \frac{r d_3}{2d_1} \bigr\rceil\leq \lfloor\frac{d_2}{d_3}\rfloor \cdot d_3^2$ of them, say $V'_1,\ldots,V'_{r'}$, and define
    \[
    K_{k+1, i'} = \frac{1}{\sqrt{r'}}V'_{i'},
    \]
    for $i'\in[r']$. We can check
    \begin{itemize}
     
        \item[(a)]  For all $i, i' \in [k]$ and $j, j' \in [l]$,
        \begin{align*}
            \Tr\big(K_{i,j}^\dagger K_{i',j'}\big) 
            &= \mathbbm{1}_{i=i'}\mathbbm{1}_{j=j'} \frac{d_2}{l} 
            \le \mathbbm{1}_{i=i'}\mathbbm{1}_{j=j'} \frac{2d_1}{r},
        \end{align*}
        where we used $\frac{d_2}{l}\leq \frac{2kd_2}{r}\leq \frac{2d_1}{r}$, and for $i', i'' \in [r']$,
        \[
            \Tr\big(K_{k+1,i'}^\dagger K_{k+1,i''}\big) 
            = \mathbbm{1}_{i'=i''} \frac{d_3}{r'} 
            \le \mathbbm{1}_{i'=i''} \frac{2d_1}{r},
        \]
        and obviously $\tr(K_{i,j}^\dag K_{k+1,i'})=0$ for $i\in[k]$.
        \item[(b)] 
        \begin{align*}
            \sum_{i=1}^k \sum_{j=1}^l K_{i,j}^\dagger K_{i,j} 
            + \sum_{i'=1}^{r'} K_{k+1,i'}^\dagger K_{k+1,i'} 
            &= \sum_{i=1}^k \sum_{j=1}^l \frac{1}{l}\,I_{\mathrm{B}_i} 
               +     \sum_{i'=1}^{r'} \frac{1}{r'} I_{\mathrm{C}} \\
            &= I_\mathrm{A}.
        \end{align*}

        \item[(c)] The total number of constructed operators is
        \[
            lk + r' 
            = \Bigl\lfloor\frac{r}{2k}\Bigr\rfloor k 
            + \Bigl\lceil \frac{r d_3}{2d_1} \Bigr\rceil
            \le \Bigl\lfloor\frac{r}{2}\Bigr\rfloor + \Bigl\lceil \frac{r}{2}-\frac{r kd_2}{2d_1} \Bigr\rceil \leq r.
        \]
    \end{itemize}
\end{proof}

\subsection{Auxiliary facts}

We also need the following facts.
\begin{fact}\label{lemma-12272204}
Let $p\in [0,1]$, and $n$ and $k$ be two positive integers such that $n\ge k$. We have that 
\[\binom{n}{k} \leq \exp\!\left(n H\!\left(k/n\right)\right),\]
where $H(\cdot)$ denotes the binary entropy function.
\end{fact}
\begin{proof}
We observe that 
\[\binom{n}{k} \left(\frac{k}{n}\right)^k\left(1-\frac{k}{n}\right)^{n-k}\leq \left(\frac{k}{n}+1-\frac{k}{n}\right)^n =1,\]
so  the inequality $\binom{n}{k}\leq \frac{n^n}{k^k (n-k)^{n-k}}$ follows. 
\end{proof}

\begin{fact}\label{fact-5122103}
Let $\ket{\psi}$ be a vector in the support of a positive semidefinite operator $M$.  We have that  
\[M\geq \ketbra{\psi}{\psi}\Longleftrightarrow 1\geq \bra{\psi}M^{-1}\ket{\psi},\]
where $M^{-1}$ is the pseudoinverse of $M$.
\end{fact}
\begin{proof}
We have the following chain of equivalences
\begin{align*}
    M\geq \ketbra{\psi}{\psi}&\Longleftrightarrow  I_{\supp(M)} \geq M^{-1/2}\ketbra{\psi}{\psi} M^{-1/2}
    \\&\Longleftrightarrow 1\geq \tr(M^{-1/2}\ketbra{\psi}{\psi} M^{-1/2})
    \\&\Longleftrightarrow 1\geq \bra{\psi}M^{-1}\ket{\psi},
\end{align*}
where $M^{-1/2}$ is the pseudoinverse of $M^{1/2}$.
\end{proof}

\begin{fact}\label{lamma-12281043}
We have that for all positive numbers $x$ and $M$
\[x\ln (M/x) \leq M/e.\]
\end{fact}
\begin{proof}
The derivative of the function $f(x)=x\ln(M/x)$ is $f'(x)=\ln(M/x)-1$. So $f'$ is strictly monotonically decreasing and  $f'(M/e)=0$. Hence  $f(x)\leq f(M/e)=M/e$.
\end{proof}

\begin{fact}\label{fact-3191602}
For any square matrix $M$ such that $M^2=0$, we have 
\[\|M+M^\dag\|_1=2\|M\|_1.\]
\end{fact}
\begin{proof}
Note that
\[\|M+M^\dag\|_1=\tr\!\left(\sqrt{(M+M^\dag)^2}\right)=\tr\!\left(\sqrt{MM^\dag+M^\dag M}\right)=\tr\!\left(\sqrt{MM^\dag}\right)+\tr\!\left(\sqrt{M^\dag M}\right),\]
where the last equality is by using $\supp(MM^\dag)\perp\supp(M^\dag M)$.
\end{proof}

\begin{theorem}[{\cite[Corollary~17]{meckes2013spectral}}]\label{lem:corollary17} Let $k,d\geq 1$. Suppose that $f: \big( \mathbb{U}_d\big)^k\to \mathbb{R}$ is $L$-Lipschitz with respect to the $\ell_2$-sum of the $2$-norms (Frobenius norm), i.e.
\begin{equation}
   \big|f(U_1,\dots, U_k)-f(U'_1,\dots, U'_k)\big|\leq L\sqrt{\sum_{i=1}^k\|U_i-U_i'\|_F^2} 
\end{equation}
for all $U_i,U_i'\in \mathbb{U}_d$, with $i=1,\dots, k$. Then, if we independently sample $U_1,\dots, U_k$ according to the Haar measure on $\mathbb{U}_d$, the following inequality holds for each $t>0$:
\begin{equation}
    \Pr\!\left[f(U_1,\dots, U_k)  \geq \mathbb{E}[f(U_1,\dots, U_k)] +t\right] \le  \exp\left(-\frac{dt^2}{12L^2}\right).
\end{equation}
    
\end{theorem}

The following facts from Weingarten calculus are needed in this work.  
For $\pi\in \mathfrak{S}_n$  a permutation of $[n]$, we denote by $\operatorname{Wg}(\pi,d)$  the  Weingarten function of dimension $d$. 
\begin{lemma}\label{lem:Wg} 
Let $U\in\mathbb{U}_d$ be a Haar-random unitary and let $\{A_i,B_i\}_{i=1}^n$ be a sequence of complex $(d\times d)$-matrices. We have the following formula for the expectation value:
\begin{equation}\label{eq: Wg-product}
    \begin{split}
&\ex{\Tr(UB_1U^\dagger A_1U\dots UB_nU^\dagger A_n)}
\\&\qquad =\sum_{\alpha,\beta \in \mathfrak{S}_n}\operatorname{Wg}(\beta\alpha^{-1},d)\tr_{\beta^{-1}}(B_1,\dots,B_n)\tr_{\alpha\gamma_n}(A_1,\dots,A_n),
\end{split}
\end{equation}
where $\gamma_n=(12\dots n)\in\mathfrak{S}_n$ and, writing $\sigma$ in terms of cycles $\{C_j\}$ as $\sigma=\prod_j C_j $,
\begin{equation*}
    \tr_{\sigma}(M_1,\dots,M_n)\coloneqq\prod_j \tr\prod_{i\in C_j} M_i.
\end{equation*}
\end{lemma}
\begin{proof} For elementary matrices $A_1, \dots, A_n, B_1, \dots, B_n$, the lemma is exactly 
    \cite[Corollary 2.4]{collins2006integration}. 
    To see this, let us assume that $A_1 = \ketbra{i_1'}{i_2}, A_2 = \ketbra{i_2'}{i_3}, \dots,  A_n = \ketbra{i_n'}{i_1}$ and 
    $B_1 = \ketbra{j_1}{j_1'}, B_2 = \ketbra{j_2}{j_2'}, \dots,  B_n = \ketbra{j_{n}}{j_n'}$ for some $i_1, \dots, i_n, i_1', \dots, i_n' \in [d]$
 and     $j_1, \dots, j_n, j_1', \dots, j_n' \in [d]$. 
The LHS of \eqref{eq: Wg-product} can be computed using \cite[Corollary 2.4]{collins2006integration}:
\begin{align*}
    \ex{\Tr(UB_1U^\dagger A_1U\dots UB_nU^\dagger A_n)}&=\ex{U_{i_1, j_1}\cdots U_{i_n, j_n} \cdot \bar{U}_{i_1', j_1'}\cdots \bar{U}_{i_n', j_n'} }
    \\&= \sum_{\alpha,\beta \in \mathfrak{S}_n}\operatorname{Wg}(\beta\alpha^{-1},d) \mathbbm{1}_{i_1=i'_{\alpha(1)}, \dots ,\,  i_n=i'_{\alpha(n)}}\cdot \mathbbm{1}_{j_1=j'_{\beta(1)}, \dots ,\, j_n=j'_{\beta(n)}}.
\end{align*}
 Since the matrices $A_1, \dots, A_n$ and $B_1, \dots, B_n$ are elementary, the RHS of \eqref{eq: Wg-product} can be expressed as follows:

\begin{align*}
   &\sum_{\alpha,\beta \in \mathfrak{S}_n}\operatorname{Wg}(\beta\alpha^{-1},d)\tr_{\beta^{-1}}(B_1,\dots,B_n)\tr_{\alpha\gamma_n}(A_1,\dots,A_n) 
   \\\quad &= \sum_{\alpha,\beta \in \mathfrak{S}_n}\operatorname{Wg}(\beta\alpha^{-1},d) \mathbbm{1}_{j_{\beta^{-1}(1)}=j'_{1}, \dots ,\, j_{\beta^{-1}(n)}=j'_{n}}\cdot \mathbbm{1}_{i_2=i'_{\alpha\gamma_n(1)}, \dots ,\,  i_1=i'_{\alpha\gamma_n(n-1)}}
   \\&= \sum_{\alpha,\beta \in \mathfrak{S}_n}\operatorname{Wg}(\beta\alpha^{-1},d)\mathbbm{1}_{j_1=j'_{\beta(1)}, \dots ,\, j_n=j'_{\beta(n)}}\cdot  \mathbbm{1}_{i_1=i'_{\alpha(1)}, \dots ,\,  i_n=i'_{\alpha(n)}}.
\end{align*}
Hence, \eqref{eq: Wg-product} holds for elementary matrices. 
    The generalization is obtained by linearity. 
\end{proof}

The following values of the Weingarten function are known \cite[Section 6]{collins2006integration}.
\begin{lemma}\label{lem:wg2}
The function $\operatorname{Wg}(\pi,d)$ has the following values:
\begin{itemize}
    \item $\operatorname{Wg}((1),d)=\frac{1}{d}$,
    \item $\operatorname{Wg}((12),d)=\frac{-1}{d(d^2-1)}$,
    \item $\operatorname{Wg}((1)(2),d)=\frac{1}{d^2-1}$.
\end{itemize}
\end{lemma}

Then, we can easily see the following result.
\begin{corollary}\label{coro-3221651}
Let $U\in\mathbb{U}_d$ be a Haar-random unitary and $A_1,B_1,A_2,B_2$ are complex $(d\times d)$-matrices. We have:
\begin{align}
\E\!\left[\tr(UB_1 U^\dag A_1 U B_2U^\dag A_2)\right]&=\frac{1}{d^2-1}\Big[\tr(B_1)\tr(B_2)\tr(A_1A_2)+\tr(B_1B_2)\tr(A_1)\tr(A_2)\Big] \nonumber\\
&\qquad\quad -\frac{1}{d(d^2-1)}\Big[\tr(B_1B_2)\tr(A_1A_2)+\tr(B_1)\tr(B_2)\tr(A_1)\tr(A_2)\Big]. \nonumber
\end{align}
\end{corollary}
\begin{proof}
By \cref{lem:Wg}, we have
\begin{align}
&\E\!\left[\tr(UB_1 U^\dag A_1 U B_2U^\dag A_2)\right] \nonumber\\
=&\operatorname{Wg}((1)(2),d)\Big[\tr_{(1)(2)}(B_1,B_2)\tr_{(12)}(A_1,A_2)+\tr_{(12)}(B_1,B_2)\tr_{(1)(2)}(A_1,A_2)\Big] \nonumber \\
&\qquad\qquad\qquad +\operatorname{Wg}((12),d)\Big[\tr_{(12)}(B_1,B_2)\tr_{(12)}(A_1,A_2)+\tr_{(1)(2)}(B_1,B_2)\tr_{(1)(2)}(A_1,A_2)\Big]\nonumber \\
=&\frac{1}{d^2-1}\Big[\tr(B_1)\tr(B_2)\tr(A_1A_2)+\tr(B_1B_2)\tr(A_1)\tr(A_2)\Big] \nonumber\\
&\qquad\qquad\qquad -\frac{1}{d(d^2-1)}\Big[\tr(B_1B_2)\tr(A_1A_2)+\tr(B_1)\tr(B_2)\tr(A_1)\tr(A_2)\Big], \nonumber
\end{align}
where the last equality is due to \cref{lem:wg2}.
\end{proof}

\section{Isometry channel tomography}
\label{sec:isometry_channel_tomography}

\subsection{Diamond norm tomography}\label{sec-3102118}
In this section, we prove a lemma that extends the unitary channel tomography algorithm that uses $O(d^2/\varepsilon^2)$ queries in \cite{haah2023query}  
to isometry channel tomography.

\begin{lemma}[Isometry channel tomography in diamond norm]\label{thm-1240018}
Suppose that $d_1\leq d_2$ are two positive integers and $\varepsilon\in (0,1)$. Let $V:\mathbb{C}^{d_1}\rightarrow\mathbb{C}^{d_2}$ be an isometry and $\mathcal{V}=V(\cdot)V^\dag\in \isochannel_{d_1,d_2}$ be the associated isometry channel.
Then, there exists an algorithm using $O(d_1d_2/\varepsilon^2)$ queries to channel $\mathcal{V}$ and it outputs an isometry channel estimate $\widehat{\mathcal{V}}$ such that $\|\mathcal{V}-\widehat{\mathcal{V}}\|_\diamond\leq \varepsilon$ with probability $\geq 2/3$. Further, the algorithhm uses these queries in parallel.
\end{lemma}

To prove \Cref{thm-1240018}, one of the key techniques is the following lemma~\cite{CL14,KRT17,GKKT20,haah2023query} for pure state tomography.

\begin{lemma}[Pure state tomography, cf.\ {\cite[Proposition 2.2]{haah2023query}}]
    \label{lmm:pure-state-tomo}
    Suppose $d$ is a positive integer.
    Then, there exists an algorithm for pure state tomography using $O(d/\varepsilon_{\max})$ copies of the unknown quantum state $\ket{v}\in \Co^d$, and it outputs a pure state estimate (by a classical description) $\ket{\widehat{v}}$ such that
    \begin{equation*}
        \ket{\widehat{v}}= \phi \sqrt{1-\varepsilon} \ket{v} + \sqrt{\varepsilon} \ket{w},
    \end{equation*}
    where $\phi$ is a random phase, $\varepsilon\in [0,1]$ is a random number with $\Pr[\varepsilon\leq \varepsilon_{\max}]\geq 1- \exp(-5d)$, and $\ket{w}$ is a Haar random state orthogonal to $\ket{v}$.
\end{lemma}

We also need the following lemma that allows us to convert a weak tomography algorithm into a standard tomography algorithm, as required by \Cref{thm-1240018}.

\begin{lemma}
    \label{lmm:weak-tomo}
    Suppose $d_1\leq d_2$ are two positive integers. Let $V:\mathbb{C}^{d_1}\rightarrow\mathbb{C}^{d_2}$ be an isometry and $\mathcal{V}=V(\cdot)V^\dag\in \isochannel_{d_1,d_2}$ be the associated isometry channel.
    Let $\calA$ be an algorithm for weak isometry channel tomography that satisfies the following condition: using queries to $\calV$, it outputs an isometry estimate $\widehat{V}:\mathbb{C}^{d_1}\rightarrow\mathbb{C}^{d_2}$ such that 
    \begin{equation}
        \label{eq:weak-tomo-precond}
        \Pr\bracks*{\exists \textup{ diagonal unitary } \Phi:\Co^{d_1}\rightarrow \Co^{d_1}, \norm*{V\Phi- \widehat{V}}_{\textup{op}}\leq \varepsilon\leq \frac{1}{8}} \geq 1-\eta,
    \end{equation}
    where $\norm*{\cdot}_{\textup{op}}$ denotes the operator norm.
    Then, there exists an algorithm for isometry channel tomography that uses $\calA$ twice in parallel and outputs an isometry estimate $\widehat{\calV'}$ such that
    \begin{equation*}
        \Pr\bracks*{\norm*{\calV -\widehat{\calV'}}_{\diamond}\leq 98\varepsilon}\geq 1-2\eta.
    \end{equation*}
\end{lemma}

\begin{proof}
    Our proof extends that of \cite[Proposition 2.3]{haah2023query} for unitary channel tomography to the setting of isometry channel tomography.
    Let $\calA$ be an algorithm for weak isometry channel tomography as described in \Cref{lmm:weak-tomo}.
    We first apply $\calA$ that uses queries to the original channel $\calV$ to obtain an isometry estimate $\widehat{V_1}:\Co^{d_1}\rightarrow \Co^{d_2}$.
    In parallel, we can apply $\calA$ that uses queries to the modified channel $\calV \circ \calF$
    to obtain another isometry estimate $\widehat{V_2}:\Co^{d_1}\rightarrow \Co^{d_2}$,
    where $\calF$ is the quantum channel for implementing the quantum Fourier transform $\mathit{F}:\Co^{d_1}\rightarrow \Co^{d_1}$.

    Combining our condition of $\calA$ and the union bound, we obtain
    \begin{equation}
        \label{eq:two-weak-tomo}
        \norm*{V\Phi_1 - \widehat{V_1}}_{\textup{op}}\leq \varepsilon\quad \textup{and}\quad \norm*{V F \Phi_2 -\widehat{V_2}}_{\textup{op}}\leq \varepsilon 
    \end{equation}
    for some diagonal unitaries $\Phi_1,\Phi_2:\Co^{d_1}\rightarrow \Co^{d_1}$,
    with probability $\geq 1-2\eta$.
    Because $V$ is an isometry, we have $V^\dagger V = \sum_{j=0}^{d_1-1} \ket{j}\!\bra{j}=I_{d_1}$. As a result,
    \begin{align*}
        \norm*{\widehat{V_1}^\dagger \widehat{V_2}- \Phi_1^\dagger F \Phi_2}_{\textup{op}}
        &\leq \norm*{\parens*{\widehat{V_1}^\dagger -  \Phi_1^\dagger V^\dagger} \widehat{V_2}}_{\textup{op}} + \norm*{ \Phi_1^\dagger V^\dagger\parens*{\widehat{V_2} - VF \Phi_2}}_{\textup{op}}\\
        &\leq \norm*{V\Phi_1 - \widehat{V_1}}_{\textup{op}} + \norm*{V F \Phi_2 -\widehat{V_2}}_{\textup{op}}\\
        &\leq 2\varepsilon.
    \end{align*}
    Now we define $p(k,j)$ to be the proposition 
    \begin{equation}
        \label{eq:def-p-k-j}
        \abs*{\bra{k}\parens*{\widehat{V_1}^\dagger \widehat{V_2}- \Phi_1^\dagger F \Phi_2}\ket{j}}\leq \frac{4\varepsilon}{\sqrt{d_1}}.
    \end{equation}
    Using the pigeonhole principle, one can obtain
    \begin{equation}\label{eq:condition-k-j}
        \textup{for any } j=0 \textup{ to } d_1-1,\quad \#\braces*{k: p(k,j)} \geq \frac{3d_1}{4},
    \end{equation}
    with probability $\geq 1-2\eta$, 
    where $\#\braces*{k: p(k,j)}$
    stands for the number of $k$ such that $p(k,j)$ is satisfied.
    
    Define $\Phi_3= \sum_{k,j=0}^{d_1-1}\frac{\bra{k}\widehat{V_1}^\dagger \widehat{V_2} \ket{j}}{\bra{k}F\ket{j}}\ket{k}\!\bra{j}$.
    If $p(k,j)$ is satisfied, we obtain, from \Cref{eq:def-p-k-j},
    \begin{equation}
        \label{eq:phi321}
        \abs*{\bra{k}\Phi_3\ket{j} - \bra{k}\Phi_1^\dagger\ket{k}\cdot \bra{j}\Phi_2\ket{j}} \leq 4\varepsilon,
    \end{equation}
     where we use the fact that $\abs*{\bra{k}F\ket{j}}=\frac{1}{\sqrt{d_1}}$ for any $k,j$.
    Moreover, if both $p(k,0)$ and $p(k,j)$ are satisfied, then from \Cref{eq:phi321} one can derive
    \begin{equation}
        \label{eq:estimate-Phi3}
        \abs*{\frac{\bra{k}\Phi_3\ket{j}}{\bra{k}\Phi_3\ket{0}}-\frac{\bra{j}\Phi_2\ket{j}}{\bra{0}\Phi_2\ket{0}}} \leq \frac{2\cdot 4\varepsilon}{1-4\varepsilon} \leq 16 \varepsilon.
    \end{equation}
    Note that from \Cref{eq:condition-k-j}, we have
    \begin{equation}
        \label{eq:both-pkj-pk0}
        \textup{for any } j=0 \textup{ to } d_1-1,\quad \#\braces*{k: p(k,j) \wedge p(k,0)} \geq \frac{d_1}{2},
    \end{equation}
    with probability $\geq 1-2\eta$.
    For each $j=0$ to $d_1-1$, define $a_j$ and $b_j$ to be the medians of the real parts and the imaginary parts of the set $\braces*{\frac{\bra{k}\Phi_3\ket{j}}{\bra{k}\Phi_3\ket{0}}}_k$, respectively. 
    In this case, \Cref{eq:estimate-Phi3,eq:both-pkj-pk0} together lead to
    \begin{equation*}
        \abs*{\parens*{a_j + i b_j} - \frac{\bra{j}\Phi_2\ket{j}}{\bra{0}\Phi_2\ket{0}}}\leq \sqrt{(16\varepsilon)^2 + (16\varepsilon)^2},
    \end{equation*}
    which further means $\phi_j =\frac{a_j + i b_j}{\abs*{a_j + i b_j}}$ satisfies
    $\abs*{\phi_j-\frac{\bra{j}\Phi_2\ket{j}}{\bra{0}\Phi_2\ket{0}}}\leq 48\varepsilon$.
    Let $\Phi=\sum_{j=0}^{d_1-1} \phi_j\ket{j}\!\bra{j}$,
    then we have
    \begin{equation}
        \label{eq:Phi-estimate}
        \norm*{\bra{0}\Phi_2\ket{0}\cdot\Phi - \Phi_2}_{\textup{op}}\leq 48\varepsilon
    \end{equation}
    with probability $\geq 1-2\eta$.
    Consequently,
    \begin{align}
        \norm*{\calV-\widehat{\calV_2} \Phi^\dagger \calF^\dagger}_{\diamond}&\leq 2
        \norm*{V\bra{0}\Phi_2\ket{0} -\widehat{V_2} \Phi ^\dagger F^\dagger }_{\textup{op}} \label{eq-221646}\\
        &\leq 
        2\norm*{V\bra{0}\Phi_2\ket{0}- \widehat{V_2} \Phi_2^\dagger F^\dagger \bra{0}\Phi_2\ket{0}}_{\textup{op}}
        + 2\norm*{\widehat{V_2} \Phi_2^\dagger F^\dagger \bra{0}\Phi_2\ket{0} - \widehat{V_2} \Phi^\dagger F^\dagger}_{\textup{op}} \nonumber\\
        &=
        2\norm*{V- \widehat{V_2} \Phi_2^\dagger F^\dagger}_{\textup{op}}
        +
        2\norm*{\bra{0}\Phi_2\ket{0}\cdot\Phi - \Phi_2}_{\textup{op}} \label{eq-221647}\\
        &\leq 98\varepsilon,\label{eq-221648}
    \end{align}
    with probability $\geq 1-2\eta$.
    Here, \cref{eq-221646} comes from \cite[Lemma 12]{AKN98} (see also \cite{kretschmann2008information}). \cref{eq-221647} exploits the fact that $\Phi_2$ is a diagonal unitary, $\widehat{V_2}$ is an isometry and $F$ is a unitary. \cref{eq-221648} is due to \Cref{eq:two-weak-tomo,eq:Phi-estimate}.
    The algorithm can output the isometry channel corresponding to $\widehat{V'}=\widehat{V_2} \Phi^\dagger F^\dagger$ as the final estimate.
\end{proof}

Given the above lemma,
one can prove \Cref{thm-1240018} for isometry channel tomography.

\begin{proof}[Proof of \Cref{thm-1240018}]
    The proof extends that of \cite[Theorem 2.1]{haah2023query} to isometry channel tomography.
    From \Cref{lmm:weak-tomo} (with appropriate rescaling of $\varepsilon$), it suffices to construct an algorithm for weak isometry channel tomography that satisfies \Cref{eq:weak-tomo-precond} with $\eta= \frac{1}{6}$.
    Our algorithm works as follows:
    \begin{enumerate}
        \item 
            Given queries to $\calV$, we first use the algorithm in~\Cref{lmm:pure-state-tomo} for pure state tomography (taking $d=d_2$ and $\varepsilon_{\max}=\Theta(\varepsilon^2)$ to be determined later) on computational basis input states $\ket{0},\ket{1},\ldots, \ket{d_1-1}$ to get estimates $\ket{\widetilde{v_j}}$ of
            $\ket{v_j}=V\ket{j}$ for all $j$ in parallel. Then, the following holds
            \begin{equation}
                \label{eq:cond-each-vj}
                \ket{\widetilde{v_j}}= \phi_j \sqrt{1-\varepsilon_j} \ket{v_j} + \sqrt{\varepsilon_j} \ket{w_j},
            \end{equation}
            where for each $j=0$ to $d_1-1$, the random variables $\phi_j, \varepsilon_j, \ket{w_j}$ are as in \Cref{lmm:pure-state-tomo}.
        \item 
            Let $\widetilde{V}=\sum_j \ket{\widetilde{v_j}}\!\bra{j}$.
            Let $\widetilde{V}=U_2 \Lambda U_1$ be the singular value decomposition of  $\widetilde{V}$ with $U_1\in \mathbb{U}_{d_1}$ and $U_2\in \mathbb{U}_{d_2}$, respectively.
            We then output the quantum channel $\widehat{\calV}$ corresponding to the isometry $\widehat{V}=U_2 \sum_{j=0}^{d_1-1} \ket{j}\!\bra{j} U_1$.
    \end{enumerate}
    One can easily compute that the number of queries to $\calV$ in the algorithm above is $O\parens*{d_1d_2 /\varepsilon^2}$. Next, to see that $\widehat{V}$ satisfies \Cref{eq:weak-tomo-precond} in \Cref{lmm:weak-tomo},
    we prove 
    \begin{equation}
        \label{eq:target-iso-tomo}
        \norm*{V \Phi- \widetilde{V}}_{\textup{op}}\leq \varepsilon/2
    \end{equation}
    for some diagonal unitary $\Phi:\Co^{d_1}\rightarrow \Co^{d_1}$,
    with probability $\geq 0.97 \geq \frac{5}{6}$.
    Whenever \Cref{eq:target-iso-tomo} holds,
    the estimate $\widehat{V}$
    in the algorithm above satisfies
    \begin{equation*}
        \norm*{V\Phi -\widehat{V}}_{\textup{op}}
        \leq \norm*{V\Phi- \widetilde{V}}_{\textup{op}}
        + \norm*{\widetilde{V}-\widehat{V}}_{\textup{op}}
        \leq \varepsilon.
    \end{equation*}
    Here, we use $\|\widetilde{V}-\widehat{V}\|_{\mathrm{op}} \leq \varepsilon/2$, because once \cref{eq:target-iso-tomo} holds, the operator norm between $\widetilde{V}$ and an isometry is at most $\varepsilon/2$ and therefore the differences between the singular values of $\widetilde{V}$ and $1$ are at most $\varepsilon/2$.
    
    Suppose $W=\sum_{j=0}^{d_1-1} \ket{w_j}\!\bra{j}$, $\Phi=\sum_{j=0}^{d_1-1} \phi_j \ket{j}\!\bra{j}$, $B_1=\sum_{j=0}^{d_1-1} \sqrt{\varepsilon_j}\ket{j}\!\bra{j}$, and $B_2=\sum_{j=0}^{d_1-1} \sqrt{1-\varepsilon_j}\ket{j}\!\bra{j}$.
    Then, $\abs*{\sqrt{1-\varepsilon_j}- 1}\leq \sqrt{\varepsilon_j}$ for all $j=0$ to $d_1-1$ implies
    $\norm*{B_2-I_{d_1}}_{\textup{op}}\leq \norm*{B_1}_{\textup{op}}$,
    where $I_{d_1}=\sum_{j=0}^{d_1-1} \ket{j}\!\bra{j}$.
    Using \Cref{lmm:pure-state-tomo} (taking $d=d_2$, where $d_2$ is sufficiently large), we obtain
    \begin{equation*}
        \norm*{B_1}_{\textup{op}} \leq \sqrt{\varepsilon_{\max}},
    \end{equation*}
    with probability $\geq 0.99$.
    From the triangle inequality, we have
    \begin{equation*}
        \norm*{V\Phi- \widetilde{V}}_{\textup{op}} = \norm*{V\Phi (B_2- I_{d_1}) + W B_1}_{\textup{op}}
        \leq \norm*{V}_{\textup{op}}\cdot\norm*{\Phi}_{\textup{op}}\cdot\norm*{B_2- I_{d_1}}_{\textup{op}} + \norm*{W}_{\textup{op}}\cdot\norm*{B_1}_{\textup{op}},
    \end{equation*}
    which implies
    \begin{equation}
        \label{eq:bound-V-tildeV}
        \norm*{V\Phi- \widetilde{V}}_{\textup{op}}\leq \sqrt{\varepsilon_{\max}} (1+\norm*{W}_{\textup{op}}),
    \end{equation}
    with probability $\geq 0.99$.

    Now we prove
    \begin{equation}
        \label{eq:bound-W-cw}
        \norm*{W}_{\textup{op}}\leq c_W
    \end{equation}
    for some constant $c_W>0$,
    with probability $\geq 0.98$.
    This will imply \Cref{eq:target-iso-tomo} by combining with \Cref{eq:bound-V-tildeV} and taking $\varepsilon_{\max}=\Theta(\varepsilon^2)$ to be sufficiently small.
    
    We define the following quantum states: for each $j=0$ to $d_1-1$, 
    \begin{equation}
        \label{eq:def-yj}
        \ket{y_j} = \sqrt{\delta_j} \psi_j \ket{v_j} + \sqrt{1-\delta_j} \ket{w_j},
    \end{equation}
    where $\sqrt{\delta_j}= \abs*{\braket{0}{x_j}}$ is the overlap between a Haar random state $\ket{x_j}\sim \Co^{d_2}$ and the state $\ket{0}$, and $\psi_j\sim [0,2\pi)$ is an uniformly random phase.
    Here, we require that $\ket{x_0},\ldots,\ket{x_{d_1-1}}$ are independent
    and $\psi_0,\ldots, \psi_{d_1-1}$ are independent.
    In this case, $\ket{y_j}\sim \Co^{d_2}$ and $\ket{y_0},\ldots,\ket{y_{d_1-1}}$ are independent. 
    Let $Y=\sum_{j=0}^{d_1-1} \ket{y_j}\!\bra{j}$.
    According to \cite[Theorem 3.4.6, complex version]{Ver18}, $\sqrt{d_2} Y$ has its column vectors being independent sub-gaussian isotropic random in $\Co^{d_2}$,
    and we can upper bound the maximal singular value of $Y$ with high probability by \cite[Theorem 4.6.1, complex version]{Ver18}:
    \begin{equation}
        \label{eq:bound-Y}
        \norm*{Y}_{\textup{op}}\leq c_Y
    \end{equation}
    for some constant $c_Y>0$, with probability $\geq 0.99$.
    Suppose $E_1=\sum_{j=0}^{d_1-1} \sqrt{\delta_j}\psi_j \ket{j}\!\bra{j}$ and
    $E_2=\sum_{j=0}^{d_1-1} \sqrt{1-\delta_j} \ket{j}\!\bra{j}$.
    As 
    \begin{equation*}
        \norm*{W}_{\textup{op}}= \norm*{\parens*{Y - V E_1} E_2^{-1}}_{\textup{op}}\leq
        \parens*{\norm*{Y}_{\textup{op}}+\norm*{V}_{\textup{op}}\cdot \norm*{E_1}_{\textup{op}}}\cdot \norm*{E_2^{-1}}_{\textup{op}},
    \end{equation*}
    by combining with \Cref{eq:def-yj,eq:bound-Y}, we obtain
    \begin{equation}
        \label{eq:bound-W}
        \norm*{W}_{\textup{op}}
        \leq \parens*{c_Y + 1} \cdot \parens{1-\max_{j} \delta_j}^{-1/2}
    \end{equation}
    with probability $\geq 0.99$.
    Since $\sqrt{d_2}\ket{x_j}$ are sub-Gaussian (similar to the case of $\sqrt{d_2}\ket{y_j}$, by \cite[Theorem 3.4.6, complex version]{Ver18}), $\sqrt{d_2\delta_j}$ are also sub-gaussian by definition, which yields
    \begin{equation*}
        \Pr\bracks*{\sqrt{\delta_j} \leq 0.1} \geq 1- e^{-\Theta(d_2)}.
    \end{equation*}
    Using the union bound and $d_1\leq d_2$, we have $\Pr\bracks*{(1-\max_j \delta_j)^{-1/2}\leq 2}\geq 0.99$ for sufficiently large $d_2$. Then, we can establish \Cref{eq:bound-W-cw}.
\end{proof}

\subsection{Choi trace norm tomography}
In this section, we prove the following lemma,
which essentially builds on an algorithm for isometry channel tomography in \cite{yoshida2025quantum}, with a slightly different analysis.

\begin{lemma}[Isometry channel tomography in Choi trace norm, adapted from \cite{yoshida2025quantum}]\label{lemma-2232322}
Suppose $d\leq D$ are two positive integers and $\varepsilon\in (0,1)$. Let $V:\mathbb{C}^{d}\rightarrow\mathbb{C}^{D}$ be an isometry and $\mathcal{V}=V(\cdot)V^\dag\in \isochannel_{d,D}$ be the associated isometry channel.
Then, there exists an algorithm that uses $n=O((D-d)d/\varepsilon^2+d^2/\varepsilon)$ queries to $\mathcal{V}$ and outputs an isometry channel estimate $\widehat{\mathcal{V}}$ such that $\left\|\frac{1}{d}C_\mathcal{V}-\frac{1}{d}C_{\widehat{\mathcal{V}}}\right\|_1\leq \varepsilon$ with probability $\geq 2/3$. Further, the algorithm uses these queries in parallel.
\end{lemma}
\begin{proof}
Suppose $n\geq 3d^2/2$. Using \cite[Lemma S5]{yoshida2025quantum} together with
$g(n)=\left\lfloor\frac{2}{3d^2}n\right\rfloor\geq \frac{1}{3d^2}n$,
 there exists an algorithm for isometry channel tomography that uses $n$ queries to achieve the average channel fidelity:
\begin{align}
\mathrm{F}&\geq 1-\frac{\pi^2(d-1)^2}{d^2g(n)^2}-\frac{D-d}{\frac{n}{d}+\frac{d-1}{2}g(n)+D-d} \geq 1- \frac{144d^4}{n^2}-\frac{D-d}{\frac{n}{d}}\geq 1-144\left(\frac{d^4}{n^2}+\frac{(D-d)d}{n}\right), \nonumber 
\end{align}
where the average channel fidelity $\mathrm{F}$ is defined as follows:
\[\mathrm{F}\coloneqq \EE{\mathcal{V}\sim \isochannel_{d,D}}\left[\EE{\widehat{\mathcal{V}}}\left[\mathrm{F}_\textup{ch}(\mathcal{V},\widehat{\mathcal{V}})\right]\right].\]
Here, $\widehat{\mathcal{V}}$ denotes the output of the tomography algorithm when the input channel is $\mathcal{V}$. 
It is worth noticing that this algorithm is also covariant. Thus, for any $\mathcal{V}\in\isochannel_{d,D}$, we have
\[1-\EE{\widehat{\mathcal{V}}}\left[\mathrm{F}_{\textup{ch}}(\mathcal{V},\widehat{\mathcal{V}})\right]=1-\mathrm{F}\leq 144\left(\frac{d^4}{n^2}+\frac{(D-d)d}{n}\right)\leq \frac{\varepsilon^2}{12},\]
if we set 
\[n=\Theta\!\left(\frac{(D-d)d}{\varepsilon^2}+\frac{d^2}{\varepsilon}\right).\]
Moreover, this implies that with probability at least $2/3$, we have
\[\left\|\frac{1}{d}C_{\mathcal{V}}-\frac{1}{d}C_{\widehat{\mathcal{V}}}\right\|_1=2\sqrt{1-\mathrm{F}_\textup{ch}(\mathcal{V},\widehat{\mathcal{V}})}\leq2\sqrt{3\cdot \frac{\varepsilon^2}{12}}= \varepsilon,\]
where the inequality comes from the Markov's inequality.
Further, note that this algorithm uses queries to $\mathcal{V}$ in parallel.
\end{proof}

\section*{Acknowledgments}
We thank Antonio Anna Mele and Lennart Bittel for helpful discussions regarding the results presented in their paper~\cite{mele2025optimallearningquantumchannels}.
FG acknowledges financial support from the European Union under the European Research Council (ERC Grant Agreement No.~101165230).
The work of ZZ was supported in part by the Australian Research Council under Grant DP250102952.
Part of the work of ZZ was done when the author was with the University of Technology Sydney, Australia.

\bibliographystyle{alpha}
\bibliography{main}
\end{document}